\documentclass[12pt, a4paper]{article}
\usepackage[top=1.5in, bottom=1.5in, left=1.5in, right=1.5in]{geometry}
\usepackage[utf8]{inputenc}
\usepackage{authblk}
\usepackage[toc,page]{appendix}
\usepackage{amssymb, amsmath, latexsym}
\usepackage{footnote}
\usepackage{caption}
\usepackage{graphicx}
\usepackage{floatrow}
\usepackage{geometry,mathtools}
\usepackage{amsmath}
\usepackage{caption,tabularx,booktabs}
\usepackage{mathtools}
\usepackage{lipsum}
\usepackage{setspace}   
\usepackage{bbm}
\usepackage{booktabs,tabularx}
\usepackage[input-decimal-markers=.]{siunitx}
\usepackage{bigints}
\usepackage{subcaption}
\usepackage[labelformat=parens,labelsep=quad,skip=3pt]{caption}
\usepackage{booktabs,caption}
\usepackage[flushleft]{threeparttable}
\usepackage{tablefootnote}
\usepackage{array}
\newcolumntype{L}{>{\centering\arraybackslash}m{3cm}}
\usepackage{amsmath}
\usepackage{array} 
\usepackage{arydshln} 
\newcolumntype{R}[1]{>{\raggedleft\arraybackslash}p{#1}}
\usepackage{multirow}
\usepackage{siunitx}
\usepackage{bbm}
\usepackage{pgfplots}
\usetikzlibrary{decorations.markings}
\usepackage{subcaption}
\usepackage{bbding}
\DeclareMathOperator*{\argmax}{arg\,max}

\usepackage[utf8]{inputenc}
\usepackage[english]{babel}
\usepackage{booktabs,caption}
\usepackage[flushleft]{threeparttable}
\usepackage{apacite}
\usepackage{natbib}
\usepackage{enumitem}
\usepackage{multirow}
\usepackage{siunitx} 
\usepackage{booktabs} 
\usepackage{caption} 
\usepackage{amsmath} 

\usepackage{bigdelim}
\usepackage{tabvar}
\usepackage{setspace}
\usepackage[flushleft]{threeparttable}
\usepackage{booktabs}
\usepackage[gen]{eurosym}
\usepackage[colorlinks=true,
            linkcolor=red,
            urlcolor=blue,
            citecolor=blue]{hyperref}
\usepackage{bm}
\usepackage{ragged2e}
\usepackage{booktabs}
\usepackage{lscape}
\usepackage{color}
\usepackage{alltt}

\usepackage{xparse}
\usepackage{amsmath,amsthm,amssymb}

\newtheorem{proposition}{Proposition}
\newtheorem{theorem}{Theorem}
\newtheorem{lemma}{Lemma}
\newtheorem{remark}{Remark}
\newtheorem{definition}{Definition}
\usepackage{booktabs,array,lmodern}

\usepackage{geometry,mathtools}
\newtheorem{assumption}{Assumption}

\newtheorem{example}{Example}
\newtheorem{corollary}{Corollary}

\begin{document}
\title{Nonparametric Analysis of Random Utility Models Robust to Nontransitive Preferences
 \author{Wilfried Youmbi
\thanks{Department of Economics, the University of Western Ontario, E-mail: \href{mailto:wyoumbif@uwo.ca}{wyoumbif@uwo.ca}.  I express my deepest gratitude to Victor Aguiar, Nail Kashaev, Gregory Pavlov, and Roy Allen for their continuous guidance, support, and insightful suggestions while I prepared this paper. I also want to thank Sergio Ocampo, Baxter Robinson, Nirav Mehta, Rory McGee, Salvador Navarro, Michael Sullivan, Lisa Tarquinio, Charles Mao Takongmo, and all participants of the applied/metrics lunch talk and the AMM reading group for helpful comments. All mistakes are mine.
}}

}
\maketitle
\begin{abstract}
The Random Utility Model (RUM) is the gold standard in describing the behavior of a population of consumers. The RUM operates under the assumption of transitivity in consumers' preference relationships,  but the empirical literature has regularly documented its violation. In this paper, I introduce the Random Preference Model (RPM), a novel framework for understanding the choice behavior in a population akin to RUMs, which preserves monotonicity and accommodates nontransitive behaviors. The primary objective is to test the null hypothesis that a population of rational consumers generates cross-sectional demand distributions without imposing constraints on the unobserved heterogeneity or the number of goods. I analyze data from the UK Family Expenditure Survey and find evidence that contradicts RUMs and supports RPMs. These findings underscore RPMs' flexibility and capacity to explain a wider spectrum of consumer behaviors compared to RUMs. This paper generalizes the stochastic revealed preference methodology of \cite{mcfadden1990stochastic} for finite choice sets to settings with nontransitive and possibly nonconvex preference relations.\vspace{0.15cm}

\vspace{0in}\\
\end{abstract}
\thispagestyle{empty}
\date{}
\section{Introduction}
The random utility model (RUM) stands as the cornerstone for understanding the decisions of a population \cite[][]{aguiar2023random, mcfadden1990stochastic, deb2023revealed}.\footnote{Early applications of RUMs have focused on transportation decisions \cite[][]{domencich1975urban, ben1985discrete}, and later studies have used this modeling framework to analyze problems in industrial organizations \cite[][]{berry1993automobile, pakes1984patents, tebaldi2023nonparametric} such as demand and welfare analysis \cite[][]{deb2023revealed, frick2019dynamic, aguiar2021stochastic, bhattacharya2019demand, tebaldi2023nonparametric}.} Rooted in the assumption that consumers maximize their utility, which is considered random, the RUM posits complete and transitive preference relationships.\vspace{0.20cm}

Transitivity is a fundamental property of rational decision-making.\footnote{Transitivity as an assumption is based on the belief that if an individual prefers an apple to a banana and a banana to an orange, this individual should prefer an apple to an orange.} Despite its central role, empirical studies have repeatedly shown its violation \cite[][]{tversky1969intransitivity, quah2006weak, loomes1991observing, cherchye2018transitivity, alos2022identifying}. For example, \cite{cherchye2018transitivity} find that more than $65 \%$ of consumers in the Spanish Continuous Family Expenditure Survey did not meet the transitivity criterion. Moreover, the transitivity assumption may fail when individuals compare pairs of alternatives and choose their preferred option \cite[][]{david1963method, tversky1969intransitivity}.\footnote{One can imagine a model with a menu of at least three goods in which individuals compare pairs of alternative options and select their preferred option. From these pairwise preferences overall preferences can be derived, which may not be transitive.}\vspace{0.20cm}

Transitivity violations can be due to various factors, such as inattentiveness, perceptual difficulties \cite[][]{luce1956semiorders}, decision-making based on a similarity or regret \cite[][]{rubinstien1988similarity, loomes1982regret}, or to temporal inconsistencies due to relative time discounting \cite[][]{roelofsma2000intransitive, ok2007theory}.\
The violation of the transitivity assumption translates into the inability of the RUM to explain the behavior of consumers, especially when the dimension of the commodity space is at least three \cite[][]{deb2023revealed, aguiar2023random, kitamura2018nonparametric}. Consequently, recovering the distribution of heterogeneous preferences simply by observing the probability of choosing a finite set of alternatives from different menus becomes impossible \cite[][]{deb2023revealed}. \vspace{0.20cm}

The main contribution of this paper is the introduction, characterization, and implementation of the Random Preference Model (RPM), that is a novel application of rationality in choice behavior analogous to RUMs. The RPM preserves monotonicity and helps in the identification of preferences when transitivity is violated. By comparing RPMs with RUMs, I find evidence contradicting RUMs and supporting RPMs. These empirical results highlight the flexibility of RPMs and their ability to explain a broader range of consumer behaviors beyond the limitations of RUMs. Furthermore, I show that meaningful welfare and counterfactual analysis are possible using a finite, repeated cross-sectional dataset on observed prices and the distribution of choices consistent with RPMs but not RUMs. As far as I know, this paper is the first to develop and implement a stochastic theory of revealed preference that is robust to transitivity violations. \vspace{0.20cm}

As \cite{nishimura2016utility} and \cite{diecidue2017regret} have pointed out, addressing non-transitivity requires reexamining the fundamental ideas about maximization, indifference, and utility. To effectively model nontransitive preferences, I adopt pairwise comparison models that work with binary menus \cite[][]{aguiar2020rationalization, shafer1974nontransitive, john2001concave, kihlstrom1976demand}. I use preference functions to represent nontransitive binary relations among choices.\footnote{A preference function, $r: X \times X \to \mathbb{R}$, maps ordered pairs of bundles of goods to real numbers. Unlike traditional utility functions that may struggle to capture nontransitive preferences, preference functions serve as numerical representations of consumer preferences and offer a more flexible framework (\citealp{aguiar2020rationalization}; \citealp{nishimura2016utility}).} These functions gauge the ``strength of the preference," where bundle $x$ is preferred to $y$ if and only if the preference function yields a positive value ($r(x, y) \geq 0$). This representation enables nontransitive decision-making, as $r(x, y) > 0$ and $r(y, z) > 0$ does not necessarily mean $r(x, z) > 0$.\vspace{0.20cm}

The RPM assumes that individuals' preferences remain stable over time, and a stochastic demand system is rationalizable if and only if it is a mixture of rationalizable nonstochastic demand systems. Therefore, I define the rationalization of these nonstochastic demand systems with the weak axiom of revealed preference (WARP). In contrast to the strong axiom of revealed preference (SARP) used in \cite{kitamura2018nonparametric} (henceforth KS) and \cite{deb2023revealed} (henceforth DKSQ), the WARP provides a more minimal, normatively appealing, and empirically robust consistency condition for a choice \cite[][]{aguiar2020rationalization, blundell2008best, cosaert2018nonparametric, quah2006weak}. In addition, the WARP can accommodate consumers with non-convex preferences, but this accommodation can lead to indecision that can manifest itself in the form of empty-valued demand correspondences \cite[][]{aguiar2020rationalization} (Henceforth AHS).\footnote{Convexity means the preference for mixed consumption bundles over their components. AHS characterizes the necessary and sufficient property that guarantees that the associated consumer preference maximization problem always has a solution.} However, it is important to note that the framework for analyzing the rationality of nonstochastic demand systems does not account for unobserved individual heterogeneity that can stem from preference variations. \vspace{0.20cm}

To characterize RPMs, I perform a WARP-based analysis in an environment with unobserved preference heterogeneity.\footnote{Unobserved preference heterogeneity means that different consumers behave differently even when all observable household characteristics are taken into account (\citealp{cherchye2019bounding}).} In my framework, I use an alternative approach similar to the random parameter model proposed by \cite{apesteguia2018monotone} that I have adapted to address nontransitive preferences. I avoid adopting the standard household-specific additive error structure like in \cite{alos2022identifying} and \cite{fosgerau2023nontransitive} which often leads to counterintuitive outcomes due to a logit shock. The additive shock is tractable but can cause problems such as the violation of monotonicity. In contrast, the random parameter model preserves the functional form and guarantees monotonicity by its own construction. In this model, a single random error in the parameter of the preference function transforms it into a different preference function within the same original space of preference functions, and the evaluation of the alternatives remains consistent. The RPM, therefore, serves as a counterpart to the single-crossing RUM while being resilient to deviations from the transitivity of preferences.\vspace{0.20cm}

Following DKSQ, I consider a scenario in which an analyst has a repeated cross-sectional data set in which each observation comprises the prevailing prices and the distribution of demand in a population at those prices. The central question is whether it is possible to find a necessary and sufficient condition that ensures the rationalizability of the data set with a population of consumers who maximize preference functions instead of utility functions. To answer this question, I provide an exact analog of the empirical characterization of KS for RUMs. This RPM characterization is independent of a parametric specification of the underlying household preference functions or the unobserved heterogeneity distribution. Furthermore, I show that RPMs, like RUMs, can also be statistically tested using the tools in KS.\vspace{0.20cm}

A direct application of the RPM consists of rationalization tests to partially identify the distribution over preference functions when observable choices are repeated in cross-sections, are unrestricted, and unobserved heterogeneity is allowed. This is a crucial step in a welfare analysis. In particular, my methodology allows for insightful studies of the effects of the substitution or addition of a product on the consumers' welfare.\footnote{Consider a government that wants to introduce a subsidy program to encourage the adoption of renewable energy sources, such as solar panels, among homeowners. The aim is to reduce carbon emissions and reliance on non-renewable energy sources. I can measure the proportion of homeowners benefitting from this policy in my welfare analysis.} Welfare analysis in the presence of general preference heterogeneity is a challenging empirical topic that has been the subject of several recent research papers 
 \cite[][]{hausman2016individual, deb2023revealed}. \vspace{0.20cm}

This paper contributes to the literature on the nonparametric analysis of RUMs \cite[][]{mcfadden1971extension, mcfadden2006revealed, mcfadden1990stochastic, falmagne1978representation, stoye2019revealed, koida2024dual, kitamura2018nonparametric, im2022non, manzini2007sequentially, kawaguchi2017testing} and their exceptional cases such as single-crossing RUMs \cite[][]{apesteguia2018monotone}, random augmented utility models (RAUMs) \cite[][]{deb2023revealed}, random quasi-linear utility models \cite[][]{yang2023random}, and dynamic RUMs \cite[][]{kashaev2023dynamic, frick2019dynamic}. My framework for testing RPMs is closer to the frameworks of \cite{kitamura2018nonparametric} and \cite{deb2023revealed}. \cite{kitamura2018nonparametric} have developed a statistical test to verify whether a population of heterogeneous households obeys the axiom of stochastic revealed preference (ARSP) for a finite collection of budget sets, thereby explicitly considering the preference relations' transitivity. \cite{deb2023revealed} have developed a random version of the expenditure-augmented utility models and have shown that this class of RUMs is amenable to statistical testing.\footnote{An expenditure augmented utility function is a function $u:\mathbb{R}^{L}_{+} \times \mathbb{R}_{-} \to \mathbb{R}$, where $u(y, -e)$ is the consumer's utility when they spend $e$ to buy the bundle of goods $y$.} \cite{kashaev2023dynamic} have developed and characterized a dynamic version of RUMs. In this paper, I take a different approach by focusing on a static model that imposes fewer constraints on consumer behavior. My work builds upon the stochastic revealed preference methodology proposed by \cite{mcfadden1971extension, mcfadden1990stochastic} (henceforth MR), KS, and DKQS but generalizes it to a setting where preference relationships are nontransitive.\vspace{0.20cm}

 The paper also contributes to the extensive literature on the (partial) identification of preferences \cite[][]{deb2023revealed, hausman2016individual, kashaev2023dynamic, apesteguia2018monotone, manski2003partial, lu2019bayesian, turansick2022identification, dardanoni2020inferring, wei2024random}, and to the literature on counterfactual stochastic choices \cite[][]{cherchye2019bounding, kitamura2019nonparametric, deb2023revealed, kashaev2021random, kashaev2023dynamic}. My welfare results allow for a more comprehensive understanding of the impact of policy changes on consumers' welfare because I can partially identify nontransitive preferences without restricting the dimension of the commodity space to two. More importantly, my welfare and counterfactual bounds are sharp and can only be improved with further information. Therefore, the SARP-based approach to stochastic revealed-preference should significantly improve the bounds, but \cite{cosaert2018nonparametric} and \cite{blundell2015sharp} found that using the SARP instead of the WARP does not substantially improve the revealed preferences' bounds. Specifically, they show that bounds under the SARP improve only marginally. \vspace{0.20cm} 
 
The final strand of literature my paper contributes to is the long-standing work on nontransitive preferences \cite[][]{fishburn1982nontransitive, fishburn1991nontransitive, fountain1981consumer, alos2022identifying, fosgerau2023nontransitive, kalai2002rationalizing, tversky1969intransitivity, hoderlein2014revealed, aguiar2020rationalization}. My work extends the revealed preference theory of AHS by explicitly incorporating individual heterogeneity. In particular, my framework is consistent with some recent work investigating stochastic choices under nontransitive preferences \cite[][]{alos2022identifying, fosgerau2023nontransitive, aguiar2020rationalization}. Along this line of thought, \cite{alos2022identifying} have developed a method to distinguish between nontransitive preferences and transitivity violations due to noise. Their nonparametric approach uses response times and choice frequencies to distinguish the revealed preferences from random fluctuations. \cite{fosgerau2023nontransitive} have found that nontransitive behavior in combination with a monotonicity assumption corresponds to a transitive stochastic choice model. \vspace{0.20cm}

The rest of the paper is structured as follows: In Section $2$, I introduce the nonstochastic choice setup and define the concept of preference function. Section $3$ deals with the development of a stochastic revealed preference methodology. In Section $4$, I describe the procedure for deciding whether a given repeated cross-sectional data set can be rationalized by an RPM. Section $5$ presents the empirical application. In Section $6$, I discuss the implications of the empirical content of the RPM for welfare comparisons and counterfactual analysis. Section $7$ is the conclusion of the paper. 
 
\section{Nonstochastic Choice Setup}
Suppose a consumer chooses bundles consisting of $L$ goods ($L \geq 2$). I use the following notation: $\mathbb{R}^{L}_{+}=\left\lbrace y \in \mathbb{R}^{L}: y \geqq (0, \ldots, 0) \right\rbrace$,~ $\mathbb{R}^{L}_{++}=\left\lbrace y \in \mathbb{R}^{L}: y > (0, \ldots, 0) \right\rbrace$.\footnote{Given $x=(x^{1},\ldots, x^{l},\ldots x^{L})$ and $y=(y^{1},\ldots, y^{l},\ldots y^{L})$; I adopt the following definitions: $x \geqq y$ means $x^{l} \geq y^{l}$ for all $l \in \left\lbrace 1, \ldots L\right\rbrace $; $x \geq y$ means $x \geqq y$ and $x \neq y$; and $x> y$ means $x^{l} > y^{l}$ for all $l \in \left\lbrace 1, \ldots L\right\rbrace $.} $X\equiv \mathbb{R}^{L}_{+}$ denotes a metric space that is interpreted as a universal set of alternatives. Suppose I had access to a finite number of observations, denoted by $T$, on prices and the chosen quantities of these goods, where the observations are indexed by $\mathbb{T}=\left\lbrace 1, \ldots, T\right\rbrace$. In each period $t \in \mathbb{T}$, consumers choose a single bundle $y_{t}=(y^{1}_{t}, \ldots, y^{L}_{t})^{'} \in X$ when facing a price vector $p_{t}=(p^{1}_{t}, \ldots, p^{L}_{t}) \in \mathbb{R}^{L}_{++}$. Wealth in the period $t \in \mathbb{T} $ is equivalent to $w_{t} \in \mathbb{R}_{++}$. $O^{T}=\left\lbrace p_{t}, y_{t}\right\rbrace_{t \in \mathbb{T}}$ that denotes all price-quantity observations. I refer to $O^{T}$ as the data which describe a single consumer's demand observed over time. \vspace{0.25cm} 
 \begin{definition} [Direct revealed preferred relations]
 $y_{t}$ is directly revealed as being preferred to $y_{s}$. The preference is denoted by $y_{t} \succeq^{R, D} y_{s}$ if $p_{t}y_{t} \geq p_{t}y_{s}$. $y_{t}$ and is strictly and directly revealed as being preferred to $y_{s}$. This preference is denoted by $y_{t} \succ^{R, D} y_{s}$, if $p_{t}y_{t} > p_{t}y_{s} $.
 \end{definition}

 \begin{definition} [Revealed preferred relations]
$y_{t}$ is revealed preferred to $y_{s}$, written $y_{t} \succeq^{R} y_{s}$, when there is a chain $\left(y_{1}, y_{2}, \ldots, y_{n} \right)$ with elements on $X$ with $y_{1} = y_{t}$ and $y_{n} = y_{s}$ such that 

$y_{1} \succeq^{R, D} y_{2} \succeq^{R, D} \ldots \succeq^{R, D} y_{n}$. Also, $y_{t}$ is strictly revealed preferred to $y_{s}$, written  $y_{t} \succ^{R} y_{s}$, when at least one of the directly revealed relations in the revealedp referred chain is strict. 
\end{definition}

 These definitions compare pairs of bundles and states that the consumer chooses $y_{t}$ if $y_{s}$ is affordable. Using these definitions of binary relations, I define the WARP, WGARP, and the SARP.
 
\begin{definition} [WARP]
Given a finite data set $O^{T}=\left\lbrace p_{t}, y_{t}\right\rbrace_{t \in \mathbb{T}}$, the WARP holds when there is no pair of observations $(s, t) \in \mathbb{T} \times \mathbb{T}$ such that $y_{t} \succeq^{R, D} y_{s}$ and $y_{s} \succeq^{R, D} y_{t}$ with $y_{t} \neq y_{s}$.
\end{definition}
This axiom is called the \cite{samuelson1938note} 's WARP. \cite{kihlstrom1976demand} later introduced a generalized version of the WARP known as the WGARP.
\begin{definition} [WGARP]
Given a finite data set $O^{T}=\left\lbrace p_{t}, y_{t}\right\rbrace_{t \in \mathbb{T}}$, the WGARP holds when there is no pair of observations $(s, t) \in \mathbb{T} \times \mathbb{T}$ such that $y_{t} \succeq^{R, D} y_{s}$ and $y_{s} \succ^{R, D} y_{t}$.
\end{definition}
The WGARP thus holds when there is no pair of observations $(s, t)$ in which if a bundle of goods $y_{t}$ is directly revealed as being preferred to a bundle of goods $y_{s}$, then $y_{s}$ cannot be directly revealed as being preferred to $y_{t}$ in the strict sense. The WGARP is a weaker version of the generalized axiom of revealed preference (GARP). The following axiom is the SARP that was introduced by \cite{houthakker1950revealed}. The SARP adds transitivity to the WARP.
\begin{definition} [SARP]
Given a finite data set $O^{T}=\left\lbrace p_{t}, y_{t}\right\rbrace_{t \in \mathbb{T}}$, the SARP holds if there is no pair of observations $s, t \in \mathbb{T}$ such that $y_{t} \succeq^{R} y_{s}$ and $y_{s} \succeq^{R, D} y_{t}$,  with  $y_{t} \neq y_{s}$.
\end{definition}

\subsection{Preference Function}
\cite{shafer1974nontransitive} showed that preferences that are not transitive can have a numerical representation. Such a numerical representation is called a preference function.
\begin{definition}[Preference function]
A preference function, $r: X \times X \to \mathbb{R}$, maps ordered pairs of bundles to real numbers.
\end{definition}
 Given a nontransitive binary relation $\succeq \subseteq X \times X$ represented by $r$, $x$ is preferred to $y$ and is denoted by $x\succeq y$ whenever $r(x, y) \geq 0$; $x$ is said to be indifferent to $y$ that is denoted by $x \sim y$ whenever $r(x, y) \geq 0$ and $r(y, x) \geq 0$; $x$ is strictly preferred over $y$ that is denoted by $x\succ y$, if $r(x, y) \geq 0$ when $r(y, x)<0$. This representation allows for nontransitivity, since $r(x, y) \geq 0$ and $r(y,z) \geq 0$ together do not indicate the sign of  $r(x,z)$.\footnote{For example, in the set $\left\lbrace x, y, z\right\rbrace $ it could be the case that $r(x, y)=r(y, z)=1/2$ and $r(x, z)=-1/2$. Consequently, $x$ is chosen over $y$, $y$ over $z$, and $z$ over $x$. The preference function generalizes the idea of utility in the sense that $r$ can be defined by $r(x, y)= u(x)-u(y)$ if the consumer's preference is described by a single utility function $u$.} \vspace{0.25cm}
 
The following definition, which comes from AHS, illustrates the properties of the preference function.
\begin{definition}[Properties of the preference function]
Consider a preference function $r: X \times X \to \mathbb{R}$.
\begin{itemize}
 \item[(i)] $r$ is complete if for all $x, y \in X$, either $r(x, y) \geq 0$ or $r(y, x) \geq 0$.
 \item[(ii)] $r$ is continuous if for all $y \in X$ and every sequence $\left\lbrace x_{n}\right\rbrace _{n \in \mathbb{N}}$ of elements in $X$ converge to $x \in X$, $\lim_{n \to +\infty} r(x_{n}, y)=r(x, y)$ \label{def11}.
 \item[(iii)] $r$ strictly increases if for all $x, y, z \in X$, $x > z$ means that $r(x, y) > r(z, y)$.\label{def12}
 \item[(iv)] $r$ is piecewise concave if there is a set of concave functions in the first argument $f_{t}(x, y)$ for $t \in \mathbb{K}$, where $\mathbb{K}$ is a compact set, such that $r(x, y)=\max_{t \in \mathbb{K}}\left\lbrace f_{t}(x, y)\right\rbrace $; and strictly piecewise concave if there is a similar set of strictly concave functions. \label{def13}
 \item[(v)] $r$ is skew-symmetric if for all $x, y \in X$, $r(x, y) = -r(y, x)$.
\end{itemize}
\end{definition}
Each property is only checked on the first argument of the preference function. The continuity condition is helpful to ensure the existence of a maximum in the constrained maximization of the preference function \cite[][]{sonnenschein1971demand}. Strict monotonicity means that ``more is better" and is stronger than local nonsatiation, which excludes thick indifference curves. Piecewise concavity and its strict version are new properties that are particularly important for my characterizations of the WARP. When the preference function is piecewise concave, the preference relation it represents has the property of star sharpness, which is a weaker condition than convexity \cite[][]{aguiar2020rationalization}.\footnote{A preference relation $\succeq$ satisfies the star sharpness property of the upper contour set if for two alternatives $x, y \in X$ and $x \succeq y$; then $ax+(1-a)y \succeq y$, for all $a \in \left[0, 1\right] $.} \vspace{0.25cm} 
The following example illustrates the Shafer's preference function. 

\begin{example}[Shafer's preference function]{\label{ex2}}
The preference function of \cite{shafer1974nontransitive} is defined as follows:
$$r(x, y)=(y^{1})^{\frac{-1}{2}}(x^{2})^{\frac{1}{2}}+ \ln (x^{3})-(x^{1})^{\frac{-1}{2}}(y^{2})^{\frac{1}{2}}-\ln (y^{3})$$ for all $x=(x^{1}, x^{2}, x^{3}), y=(y^{1}, y^{2}, y^{3}) \in X=\mathbb{R}_{+}^{2} \times \mathbb{R}_{+}\setminus\left\lbrace 0\right\rbrace$, where $ln$ denotes the natural logarithm. This function is skew-symmetric, continuous, strictly increasing, and concave over the range $X \times X$. The preference relation associated with this preference function is intransitive (\citealp{kihlstrom1976demand}; \citealp{shafer1974nontransitive}).
\end{example}

The following definition provides the condition under which a preference function strictly rationalizes a finite set of observations.
\begin{definition} [Strict Preference Function Rationalization]{\label{def8}}
Consider a finite data set $O^{T}$ and a preference function, $r: X \times X \to \mathbb{R}$. The data set $O^{T}$ is strictly rationalized by $r$ if for all $t \in \mathbb{T}$, then $r(y_{t}, y)>0$ for all $y$ such that $p_{t}y_{t} \geq p_{t}y$ whenever $y_{t} \neq y$.
\end{definition}
According to this definition, a preference function is said to rationalize the data on different goods if, at each observation, it assigns a strictly higher ``strength" to the chosen bundle than to any other bundle that is weakly cheaper at the prevailing prices. \vspace{0.25cm} 

The demand function is determined as follows:
 Suppose $y_{t}$ is the optimal bundle at prices $p_{t}$ and income $w_{t}$, then $r(y_{t}, y) \geq 0$ for all $y$ such that $p_{t}y \leq w_{t}$. As $r(y_{t}, y)=-r(y, y_{t})$,  $r(y, y_{t})\leq 0$ for all $y \in X$ such that $p_{t}y \leq w_{t}$. Since $r(y_{t}, y_{t})=0$, $r(y_{t}, y)$ is at a maximum when $y_{t} = y$ for all $y$ such that $p_{t}y \leq w_{t}$. The Lagrangian is then defined as follows:
 $$L(x, y; \beta)=r(y, y_{t})+\beta (w_{t}-p_{t}y_{t})$$
 As in \cite{shafer1974nontransitive}, the first order conditions for the maximum
of  $$\frac{r_{l}(y, y_{t})}{r_{1}(y, y_{t})}=\frac{p^{l}}{p^{1}},~l=2,\ldots, L,~\text{and}~p_{t}y_{t}=w_{t}$$
 are satisfied for $y=y_{t}$, and $r_{l}$ is the partial derivative of $r$ with respect to $y^{l}$. Thus, to find $y_{t}$, I solve the $L$ equations for $y_{t}$.
$$\frac{r_{l}(y_{t}, y_{t})}{r_{1}(y_{t}, y_{t})}=\frac{p^{l}}{p^{1}},~l=2,\ldots, L,~\text{and}~p_{t}y_{t}=w_{t}$$ 
\vspace{0.25cm}

The following lemma is equivalent to Afriat's Theorem but for the WARP and uses a strict preference function rationalization.
\begin{lemma}{\label{lem1}}
Consider a finite data set $O^{T}=\left\lbrace p_{t}, y_{t}\right\rbrace_{t \in \mathbb{T}}$. $O^{T}$ is strictly rationalized by a continuous, strictly increasing, skew-symmetric, and (strictly) piecewise concave preference function if and only if it satisfies the WARP.
\end{lemma}
\begin{proof}
 See Appendix
\end{proof}

This Lemma provides a revealed-preference characterization of WARP for finite data sets and the model behind WARP. With this lemma, the analysis of the nonstochastic demand remains independent of a parametric specification of the underlying household preference functions. This lemma is analogous to Corollary $1$ in AHS. Compared to this corollary, I tighten the constraints on the preference function by requiring the WARP as opposed to the WGARP; to strictly increase as opposed to being monotonic; and, as a result, requiring strict rationalization as opposed to rationalization. This lemma generalizes \cite{matzkin1991testing}'s Theorem about strict rationalization with a monotone utility function.\footnote{\cite{matzkin1991testing}'s Theorem is a special case of Lemma \ref{lem1} when the preference function $r$ is defined as: $r(x, y)=u(x)-u(y)$, where $u$ is the only utility function that rationalizes the data.} This result highlights the difference between this paper and KS and DKSQ, as the WARP characterizes nonstochastic demand systems.\vspace{0.25cm}
 
A limitation of Lemma \ref{lem1} is that the test of the WARP requires repeated observations of a consumer across different price regimes \cite[][]{echenique2011money}. Therefore, it does not account for differences in consumer preferences across the population and may not be helpful if using repeated cross-sectional data. In the next section, following MR, I perform a WARP-based analysis in a setting with unobserved preference heterogeneity.\vspace{0.20cm}

The following example is a three-price-commodity pair that obeys the WARP but violates the SARP.
\begin{example}[Intransitive choices]{\label{ex3}}
 Consider the data set consisting of the following three choices:
 $$p_{1}=(2, 3, 3), ~x_{1}=(3, 1, 7), ~p_{2}=(3, 2, 3),~x_{2}=(7, 3, 1), ~p_{3}=(3, 3, 2), \text{and}~x_{3}=(1, 7, 3)$$ None of these price and consumption pairs violates the WARP. However, there is a cycle with this set of choices: $x_{1} \succ^{R, D} x_{2}$, $x_{2} \succ^{R, D} x_{3}$, and $x_{3} \succ^{R, D} x_{1}$.
Therefore, the SARP is violated because transitivity does not apply.
 Therefore, the weak axiom is valid, but no (locally nonsatiated) utility function rationalizes these data. The consumer with this behavior, therefore, is not a utility maximizer. According to the previous Lemma, there exist a continuous, strictly increasing, skew-symmetric, and (strictly) piecewise concave preference function that rationalizes these observations. 
\end{example}

\section{Stochastic Choice Setup}
\subsection{Stochastic Demand System Setup}
In this section, I develop the random version of the preference function model. I provide a stochastic revealed preference theory for consumers with intransitive behaviors similar to those described in examples \ref{ex2} and \ref{ex3}.  
$$ B_{t}=B_{t}(p_{t}, w_{t})=\left\lbrace y \in X: p_{t}^{'}y \leq w_{t} \right\rbrace$$
where $p_{t} \in \mathbb{R}^{L}_{++} $ is the vector of prices and $w_{t} \in \mathbb{R}_{++}$ is the income. $B_{t}$ is the hyperplane budget in period $t$. For every $t \in \mathbb{T}$, let $P_{t}$ there is a probability measure on the set of all the Borel measurable subsets $B_{t}$. The primitive in the demand framework is the collection of all observed $P_{t}$, $P =\left\lbrace P_{t}\right\rbrace_{t \in \mathbb{T}}$. This collection is called a \textit{stochastic} demand system. \vspace{0.20cm} 

I consider a population of consumers in which each is endowed with a fixed preference function $r: X \times X \to \mathbb{R}$, and each is unobserved by the analyst. I restrict attention to strictly increasing preference functions (``more is better").  

 \begin{example} [Random version of the Shafer's preference function] {\label{ex4}}
Considering the previous \cite{shafer1974nontransitive} preference function, given $x=(x^{1}, x^{2}, x^{3}),~ y=(y^{1}, y^{2}, y^{3}) \in X=\mathbb{R}_{+}^{2} \times \mathbb{R}_{+}\setminus\left\lbrace 0\right\rbrace$ and $\alpha \in \left(0, 1 \right)$, the generalized \cite{shafer1974nontransitive} preference function is defined as follows:
$$r_{\alpha}(x, y)=(x^{2})^{\alpha} (y^{1})^{\alpha-1}+ \ln (x^{3})-(y^{2})^{\alpha}(x^{1})^{\alpha-1}-\ln (y^{3})$$

In the appendix, I show that $r_{\alpha}$ is continuous, skew-symmetric, strictly increasing, and concave in $x$ for every fixed $y$. If $\alpha=1/2$, then $r_{1/2}$ is the classical preference function of \cite{shafer1974nontransitive}.
Assume that $\alpha$ is a random variable that follows a uniform distribution in the set $\left[0, 1 \right]$, then $r_{\alpha}$ is a random variable. Consider a population of households, each of which draws $\alpha$ from a uniform distribution supported by $\left[0, 1 \right]$. The preference functions are thus drawn randomly over a probability distribution. Hence, for each realization $\tilde{\alpha}$ of $\alpha$, the realized preference function is $\tilde{r}$, and the household behaves like a maximizer of $\tilde{r}$. The collection of $r_{\alpha}$ and the associated distributions are therefore a special case of RPMs.
\end{example}

In this example, the defined random preference function looks similar to a model with random parameters. Thus, a single random error in the parameters of the preference function transforms it into a different preference function within the same original space of preference functions. At the same time, the evaluation of the alternatives remains consistent.\vspace{0.25cm}

Given the \textit{stochastic} demand system $P$, I can define a nontransitive random demand model (NRDM). Denote $\mathcal{R}$ as the set of all continuous, strictly increasing, skew-symmetric, and (strictly) piecewise concave preference functions that map $X \times X$ to $\mathbb{R}$.

\begin{definition}{\label{def9}}
 The stochastic demand system $P =\left\lbrace P_{t}\right\rbrace_{t \in \mathbb{T}}$ is consistent with the NRDM or is \textit{stochastically} rationalizable if there is a probability measure over $\mathcal{R}$, $\mu$, such that:
 $$P_{t}(\mathcal{Y})=\int \mathbf{1}\left\lbrace y \in \mathcal{Y}:~r(y, z) \geq 0, ~\forall z \in B_{t} \right\rbrace d\mu(r)$$
 for all $t \in \mathbb{T}$ and for all the Borel measurable $\mathcal{Y} \subseteq X$.
\end{definition}
From this definition, all possible Borel sets of $ X$ must be examined to check whether a stochastic demand system $P$ is consistent with the NRDM. Later, I propose an equivalent characterization of \textit{stochastic} rationalizability that does not require examining all possible Borel sets.\vspace{0.25cm}

\subsection{Stochastic Choice Function Setup}
In this subsection, I show that stochastic rationality according to the RPMs depends solely on the probabilities associated with some areas of the budget hyperplane, which are referred to as \textit{patches}.
For any $t \in \mathbb{T}$, let $\left\lbrace x_{i \mid t}\right\rbrace_{i \in \mathcal{I}_{t}}$, $\mathcal{I}_{t}=\left\lbrace 0, 1, \ldots, I_{t}\right\rbrace$ denote a finite partition of $B_{t}^{*}$ (I index each element of the partition by $i$). $I_{t}$ is the number of elements in menu $B_{t}^{*}$ and $x_{i \mid t}$ is a subset of $B_{t}$.
\begin{definition}[Patches]
    Let $\cup_{t \in \mathbb{T}} x_{i \mid t}$ be the coarsest partition of $\cup_{t \in \mathbb{T}} B_{t}$
such that: 
$$ x_{i \mid t} \cap B_{t} \in \left\lbrace x_{i \mid t}, \emptyset\right\rbrace$$
for any $t, s \in \mathbb{T}$ and $i \in \mathcal{I}_{t}$. A set $x_{i \mid t}$ is called a patch. If $x_{i \mid t} \subseteq B_{s}$  for some $i$ and $t \neq s$ exist, then $x_{i \mid t}$ is called an intersection patch.
\end{definition}
    
By definition, patches can only be strictly above, strictly below, or on budget hyperplanes. A typical patch belongs to one budget hyperplane. However, intersection patches always belong to several budget hyperplanes. The case for period, $L = 2$ goods and $T=2$ budgets, is depicted in \ref{fig1}. \vspace{0.25cm}

The (discretized) choice set is:
$$X^{*}=\cup_{i \in \mathcal{I}_{t}, t \in \mathbb{T}} \left\lbrace x_{i \mid t} \right\rbrace$$
 
I define a menu as the collection of patches from the same budget hyperplane: 
$$ B_{t}^{*}=\left\lbrace x_{i \mid t}\right\rbrace_{i \in \mathcal{I}_{t}} \subseteq  X^{*}$$
Henceforth, I refer to a menu or budget interchangeably. For every $t \in \mathbb{T}$, $\rho_{t}$ is a probability measure on $B_{t}^{*}$. $\rho_{t} (x_{i \mid t})=\rho_{i \mid t} \geq 0$ for all $i \in \mathcal{I}_{t}$ and $\Sigma_{i \in \mathcal{I}_{t}}\rho_{i \mid t} =1$. The primitive in my framework is the collection of all observed $\rho_{t}$, $\rho=\left(\rho_{t}\right)_{t \in \mathbb{T}}$. I call this collection a \textit{stochastic choice function}. \footnote{Let $>$ be a binary relation with $X^{*}$. Completeness: For all $A, B \in X^{*}$, $A>B$ or $B>A$. Transitivity: For all $A, B, C \in X^{*}$, if $A>B$, $B>C$, then $A>C$. I allow the violation of transitivity in this setting, but I do not allow preference reversal: for all $A, B \in X^{*}$, if $A>B$, then $\neg(B>A)$. That is, if $A$ is strictly preferred to $B$, then $B$ cannot be strictly preferred to $A$.}\vspace{0.25cm}
 
Let $\rho_{t}(x_{i \mid t})=P_{t}(x_{i \mid t})$ denote the fraction of consumers who pick from a choice set $x_{i \mid t}$ given a budget set $t$. The most important building block of my demand framework is the \textit{stochastic} choice function,
 $$\rho=\left( \rho_{t}(x_{i \mid t}) \right)_{i \in \mathcal{I}_{t}, t \in \mathbb{T}}=\left(\rho_{i \mid t} \right)_{i \in \mathcal{I}_{t}, t \in \mathbb{T}} $$
The vector $\rho$ represents the distribution over a finite group of patches and contains all the essential details required to determine whether $P$ aligns with the NRDM.
Given a period $t$, a representative consumer chooses a patch $x_{i \mid t}$ from the discretized budget $B^{*}_{t}$. This representative consumer carries the characteristics of all consumers in the population whose nonstochastic choices in period $t$ fall within the patch $x_{i \mid t}$. Thus, there is a realization $r^{*}$ of the preference function $r$ with probability $\rho_{t}(x_{i \mid t})$ that generates nonstochastic choices belonging to $x_{i \mid t}$. 


\begin{example}
Considering Figure \ref{fig1}, this example has a total of five patches; namely, the intersection and the four line segments identified in the figure: $x_{0\mid 1}$, $x_{1\mid 1}$, $x_{0 \mid 2}$, and $x_{1\mid 2}$. If the intersection patches are disregarded, the vector representation of $(P_{1},P_{2})$ is $\left(\rho_{0\mid 1}, \rho_{1\mid 1}, \rho_{0 \mid 2}, \rho_{1\mid 2}\right) $.
\end{example}

\begin{figure}
  \caption{Choice probabilities with $L=2$ goods and two budgets $B_{1}^{*}=\left\lbrace x_{0\mid1}, x_{1\mid1} \right\rbrace $ and $B_{2}^{*}=\left\lbrace x_{0\mid2}, x_{1\mid 2} \right\rbrace$, $\hat{\rho}_{0 \mid 2}=2/3$,  $\hat{\rho}_{1 \mid 2}=1/3$, $\hat{\rho}_{0 \mid 1}=1/4$, and $\hat{\rho}_{1 \mid 1}=3/4$}{\label{fig1}.}
\begin{tikzpicture}[scale=0.6]Axis
\draw [->] (0,0) node [below] {0} -- (0,0) -- (6,0) node [right] {$y_{1}$};
\draw [->] (0,0) node [below] {0} -- (0,0) -- (0,6) node [above] {$y_{2}$};
\node [below][scale=0.7] at (5,0) {$B_{1}^{*}$};
\node [left][scale=0.7] at (0, 2.5) {};
\node [below][scale=0.7] at (2.5,0) {$B_{2}^{*}$};
\node [left][scale=0.7] at (0, 5) {};
\draw [ultra thick] (0, 2.5)--(5,0) node[above][scale=0.7]{};
\draw [ultra thick] (0,5)--(2.5, 0) node[above left][scale=0.7]{};
\draw[thick] (1.6, 1.8) -- (2.5, 0);
\draw[thick] (0, 2.5) -- (1.6, 1.8);
\node [below][scale=0.7] at (1.2, 4) {$x_{0 \mid 2}$};
\node [below][scale=0.7] at (1.6, 1) {$x_{1 \mid 2}$};
\node [below][scale=0.7] at (0.8, 2) {$x_{0 \mid 1}$};
\node [below][scale=0.7] at (4, 1.4) {$x_{1 \mid 1}$};
\end{tikzpicture}
  \end{figure}

\section{Rationalization by RPMs}
I now delve into an observed repeated cross-sectional data set $S= \left\lbrace p_{t}, \rho_{t}\right\rbrace_{t \in \mathbb{T}}$, where each observation represents the prevailing price $p_{t}$ and the corresponding distribution of choices made by a population of consumers. This data set also reflects the distribution of demand in the population at this price and is represented by the probability measure $\rho_{t}$ on $B_{t}^{*}$ for all $t \in \mathbb{T}$. 
\vspace{0.20cm}

\begin{definition} \label{def3}
The RPM rationalizes the repeated cross-sectional data set $S=\left\lbrace p_{t}, \rho_{t}\right\rbrace_{t \in \mathbb{T}}$ when there is a probability space $(\Theta, \mathcal{F},\mu)$ and a random variable $y:\Theta \to \mathbb{R}^{LT}_{+}$. With these conditions,  $\left\lbrace p_{t}, y_{t}(\theta)\right\rbrace_{t \in \mathbb{T}}$ almost surely can be strictly rationalized by a continuous, strictly increasing, skew-symmetric, piecewise concave preference function (equivalently, obeys the WARP). Then,

$\rho_{t}(\mathcal{Y}) =\mu \left( \left\lbrace \theta \in \Theta: y_{t}(\theta) \in \mathcal{Y} \right\rbrace \right)$, where $\mathcal{Y}\subseteq \mathbb{R}^{L}_{+}$ is a measurable set of $\cup_{i \in \mathcal{I}_{t}} x_{i \mid t}$.
\end{definition}

 In this definition, $\Theta$ can be interpreted as the population of consumers and $y_{t}(\theta)$ as the demand of consumer type $\theta$ at observation time $t$ when the prevailing price is $p_{t}$. $\mu$ is the true distribution that rationalizes the data, and the consumer types that support $\mu$ must fit the preference model.\footnote{Rationalizing the data set $S$ involves determining whether $S$ is generated by maximizing preference functions over a distribution of such functions that remain stable over the entire observation period. Alternatively, the model can be interpreted such that the preference function of each individual changes over time, but the distribution in the population is stationary.} A consumer's type is an observed demand system $\left\lbrace y_{t}(\theta)\right\rbrace_{t \in \mathbb{T}} $ such that for all $t \in \mathbb{T}$, $y_{t} \in \cup_{i \in \mathcal{I}_{t}} x_{i \mid t}$. A consumer's type is WARP-consistent when $\left(y_{1}(\theta), \ldots, y_{T}(\theta)\right)$ (together with the associated prices) satisfies the WARP. 

\subsection{Weak Axiom of Revealed Stochastic Preference}
Here, I provide a revealed-preference characterization of the weaker version of \cite{mcfadden2006revealed}'s ASRSP that deals with the preference function maximization instead of the utility maximization. If the RPM rationalizes the data set $S= \left\lbrace p_{t}, \rho_{t} \right\rbrace_{t \in \mathbb{T}}$; then there is a realization $r^{*}$ of the preference function $r$ and then the vector $a(r^{*})$ defined as follows:\vspace{0.20cm}
$a(r^{*})=\left(a_{1}(r^{*}), \ldots,  a_{T}(r^{*})\right)^{'}$, where $a(r^{*})$ 
 represents a choice pattern over the choice set represented by $r^{*}$. For all  $t \in \mathbb{T}$ and for all $i \in \mathcal{I}_{t}$:
 \begin{align*}
   a_{i \mid t}(r^{*})=
  \begin{cases}
                                   1 & if ~patch~x_{i \mid t}~\text{is chosen from}~B^{*}_{t}\\
                                   0 &  \text{Otherwise} \\
  \end{cases}
 \end{align*}
The vector $\left(a_{1}, \ldots,  a_{T} \right)^{'}$ is called permissible if it is rationalizable in terms of a preference function. That is if it can be written as: 
 
$$ a_{i \mid t}(r^{*})=\mathbf{1}\left\lbrace  y_{t} \in x_{i \mid t}:~r^{*}(y_{t}, z_{t}) \geq 0, ~p_{t}y_{t}\geq p_{j} z_{t},  \forall z_{t} \in B_{t} \right\rbrace$$
Given a finite discrete set $I$, $\mathcal{M}(I)$ and $\mathcal{M}^{*}(I)$ are sets of ordered subsets in which I allow or do not allow repetition, respectively. For example, given $I=\left\lbrace 2, 3, 4, 5\right\rbrace$, $\mathcal{M}(I)$ includes $\left\lbrace 2, 3\right\rbrace$, $\left\lbrace 2, 3, 4, 4\right\rbrace$, and $\left\lbrace 2, 4, 3, 4\right\rbrace$ and so on. $\mathcal{M}^{*}(I)$ does not include the second and third elements in the above example. $\mathcal{M}(I)$ is infinite, but $\mathcal{M}^{*}(I)$ is finite if $I$ is finite.\vspace{0.25cm}

The following axiom is the counterpart to the \cite{mcfadden2006revealed}'s axiom for (static) stochastic revealed preferences.

\begin{definition}[Weak Axiom of Revealed Stochastic Preference, WARSP]
A stochastic choice function $\rho=\left(\rho_{i \mid t} \right)_{i \in \mathcal{I}_{t}, t \in \mathbb{T}}$ satisfies WARSP if for any $M \in \mathcal{M}(\mathbb{T})$ and for any $x_{i \mid t}$, $t \in M$ and $i \in \mathcal{I}_{t}$,
$$\sum_{t \in M}\rho_{i \mid t} \leq \max_{\theta \in \Theta} \sum_{t \in M} a_{i\mid t}(r^{*}_{\theta})$$.   
\end{definition}

A stochastic choice function $\rho=\left(\rho_{i \mid t} \right)_{i \in \mathcal{I}_{t}, t \in \mathbb{T}}$ satisfies the WARSP with respect to the hypothetical preference functions $\mathcal{R}$ if the condition above holds.\footnote{$\mathcal{R}$ is the set of all continuous, strictly increasing, skew-symmetric, and (strictly) piecewise concave preference functions. $\mathbb{T}=\left\lbrace 1, \ldots, T \right\rbrace $.} The WARSP states that the sum of choice probabilities over a finite sequence of trials $\left\lbrace \left(x_{i \mid t}, B_{t}^{*} \right)\right\rbrace_{i \in \mathcal{I}_{t}, t \in \mathbb{T}}$ is not greater than the maximum number of successes that an admissible preference function can generate. The WARSP provides conditions that allow rationality to be tested without restricting heterogeneity and the number of goods. Compared to the ARSP, the WARSP allows for a less restrictive modeling of preferences. This is an advantage in complex decision-making where transitivity may not apply. \vspace{0.25cm}

The following assumptions are helpful for the stochastic revealed-preference characterization of WARSP for finite data sets.

\begin{assumption}{\label{ass1}}
Every demand correspondence is non-empty.
\end{assumption}

\cite{aguiar2020rationalization} noted that WARP can lead to empty demand correspondences. To address this, they introduced a concept called \textit{total coherency in segments}, which ensures that the demand correspondence is always non-empty. Under this assumption, the entire demand correspondence is not empty and falls into a specific patch.
Given a period $t$ and a budget $B_{t}=\cup_{i \in \mathcal{I}_{t}} x_{i \mid t}$, then $\sum_{i \in \mathcal{I}_{t}} \rho_{i \mid t}=1$, with $\rho_{i \mid t} >0$ for all $i \in \mathcal{I}_{t}$.\vspace{0.25cm}

Suppose $S$ is RPM-rationalized and suppose that with positive probability, there is a tie between two alternatives $y_{t}$ and $z_{t}$ in a given period $t$. Let $\rho_{t} \left(y_{t}, \left\lbrace z_{t}, y_{t} \right\rbrace \right)$ be the probability of choosing $y_{t}$ from the binary menu $\left\lbrace z_{t}, y_{t} \right\rbrace$.

\begin{align*}
\rho_{t} \left(y_{t}, \left\lbrace z_{t}, y_{t} \right\rbrace \right) + \rho_{t} \left(z_{t}, \left\lbrace z_{t}, y_{t} \right\rbrace \right) &=\mu  \left\lbrace \theta \in \Theta: r^{*}_{\theta}(y_{t}, z_{t}) \geq 0 \right\rbrace+ \mu \left\lbrace \theta \in \Theta: r_{\theta}^{*}(z_{t}, y_{t}) \geq 0 \right\rbrace \\
&=\mu  \left\lbrace \theta \in \Theta: r^{*}_{\theta}(y_{t}, z_{t}) > 0 \right\rbrace + \mu \left\lbrace \theta \in \Theta: r_{\theta}^{*}(z_{t}, y_{t}) > 0 \right\rbrace \\
& + 2\mu \left\lbrace \theta \in \Theta: r_{\theta}^{*}(z_{t}, y_{t}) =0\right\rbrace\\
&> 1~ \text{(contradiction)}.
\end{align*}

Thus, there are no ties without loss of generality and ties are excluded by construction, as only strict non transitive preferences with positive probability are realized. This is the purpose of the following assumption. \vspace{0.20cm}

\begin{assumption}[Intersection patches]{\label{ass2}}
 For any $t, s \in \mathbb{T}$, $\rho_{t}\left(y \in X^{*}: y \in B_{t}^{*} \cap B_{s}^{*}\right)=0$
\end{assumption}

There is no probability mass at the intersection of any pair of hyperplane budgets. This assumption allows us to exclude intersection patches in our RPM characterization. Its advantage is that it helps generate sufficient variation in preferences and demand. Under this assumption, rationalization by the WARP and the WGARP leads to the same result.

\subsection{Nonparametric Characterization of RPMs}
The following theorem provides a testable condition over the demand distribution for deciding whether a RPM can rationalize a given repeated cross-sectional data set $S$.

\begin{theorem}{\label{theo1}}
Given $L \geq 2$ and the data $S= \left\lbrace p_{t}, \rho_{t} \right\rbrace_{t \in \mathbb{T}}$satisfy the assumptions \ref{ass1} and \ref{ass2}, the following are equivalent: 
\begin{enumerate}
\item[(i)]$S$ is rationalized by the RPM. \label{theo21}\vspace{0.20cm}
\item[(ii)] There is a probability vector $\nu \in \Delta^{H-1}$ in which the discretized choice probability $\rho$ satisfies the condition~ $\Gamma \nu=\rho$.\footnote{$\Delta^{H-1}=\left\lbrace \left( \nu_{1}, \ldots, \nu_{H-1}\right) \in \mathbb{R}^{H-1}_{+}: \sum_{i=1}^{H-1} \nu_{i} \leq 1;~1 \leq i \leq H-1\right\rbrace $ is a $H$-dimensional simplex and $H < +\infty$.} \label{theo22}
\item[(iii)] $\rho=\left\lbrace \rho_{t}\right\rbrace_{t \in \mathbb{T}}$ satisfies the WARSP.\label{theo23}
\item [(iv)]  The stochastic demand system $P =\left\lbrace P_{t}\right\rbrace_{t \in \mathbb{T}}$ is \textit{stochastically} rationalizable.
\end{enumerate}
\end{theorem}

\begin{proof}
The equivalence between $(i)$ and $(ii)$ is adapted from \cite{deb2023revealed}'s proof of Theorem $1$. The equivalence between $(ii)$ and $(iii)$ is adapted from \cite{mcfadden2006revealed}. The equivalence between $(ii)$ and $(iv)$ is adapted from \cite{kitamura2018nonparametric}'s proof of Theorem $3.1$.
\end{proof}

According to Theorem \ref{theo1}, the data $S$ are only rationalized by the random preference function models if the probability of choice $\rho$ is a convex combination of the columns of $\Gamma$. The weights $\nu$ can be interpreted as the implied population distribution over rational choice types. The dimension of the rational demand matrix $\Gamma $ is $d_{\rho} \times H$, where $d_{\rho}$ is the length of $\rho$ and $H$ is the number of columns or rational types. $\Gamma$ is a matrix of zeros and ones that captures all WARP-consistent types. The above theorem is an analog of Theorem $1$ in DKSQ and of Theorem $3.1$ in KS. The difference lies in the size of the rational demand matrix that has more columns than the rational demand matrix of KS and DKSQ.\vspace{0.20cm}

Statement $2$ in Theorem \ref{theo1} is the empirical condition of the RPM, which is amenable to statistical testing. This statement is a nonparametric characterization of RPMs. It is robust to errors that could result from misspecification of the functional form of the preference function or the unobserved heterogeneity distribution. Statement $3$ reflects that the WARSP characterizes RPMs similarly to how the ARSP characterizes RUMs. Thus, this theorem shows that the WARSP is necessary and sufficient to rationalize stochastic choices through an RPM. This result complements the long-standing assertion that the WARSP is a necessary but not sufficient condition for stochastic demand behavior to be rationalizable in terms of stochastic orderings on the commodity space \cite[][]{cosaert2018nonparametric, hoderlein2014revealed, kawaguchi2017testing, dasgupta2007regular, bandyopadhyay1999stochastic}. The reasoning is that the WARP alone does not reflect utility maximization. However, as the WARP reflects preference function maximization, I can consider multiple choice situations and check the WARSP between all pairs of choices without imposing transitivity as an additional restriction as for the SARSP.\footnote{The SARSP uses transitivity in addition to the property that preferences have a unique maximum. This condition is weaker than ARSP but provides a test that can be implemented with little computational cost \cite[][]{kawaguchi2017testing}.}

Statement $4$ reflects that to establish that a stochastic demand system $P$ is consistent with the NRDM, not all possible Borel sets need to be checked as suggested in definition \ref{def9}. So, considering the cross-sectional probabilities of the patches in the respective budgets is enough to characterize the rationalizability of the stochastic demand system $P$.\vspace{0.20cm}

\begin{example}[Consumers types]
In example \ref{figur3}, which represents a two-budget case, there are three rational types of consumers. The fourth figure has an irrational type since the WARP is violated there. Hence, the only excluded behavior is $\left(0, 1, 1, 0 \right)$. The set of rational types are $\Theta=\left\lbrace \theta_{1}, \theta_{2}, \theta_{3} \right\rbrace $. If I disregard the intersection patches, then both of the matrices for the WARP- and WGARP-consistent types are defined as follows:
$$\Gamma^{*}=
\left( {\begin{array}{ccc}
1 & 0& 1\\
0 & 1 & 0\\
1&0 &0 \\
0 &1 &1
\end{array} }\right)\begin{array} {c}
x_{0 \mid 2}\\
x_{1 \mid 2}\\
x_{0 \mid 1}\\
x_{1 \mid 1}
\end{array}.$$
\end{example}
\subsection{Relationship Between RUMs and RPMs}
In this section, I examine the relationship between rationalization by the RUM and rationalization by the RPM. The following proposition states that rationalization by the RUM implies rationalization by the RPM.

\begin{proposition} {\label{prop1}}
    Given $L \geq 2$ and $S= \left\lbrace p_{t}, \rho_{t} \right\rbrace_{t \in \mathbb{T}}$ satisfying the assumptions \ref{ass1} and \ref{ass2}; if $S$ is rationalized by the RUM, then it is also rationalized by the RPM.
\end{proposition}
Thus, if a repeated cross-sectional data set $S$ is RUM-rationalized, it is also RPM-rationalized. The reverse is not necessarily true since the WARP is necessary but not sufficient for utility maximization. 

The RPM and the RUM are identical if the dimension of the commodity space is restricted to two. When dealing with more than two commodities, RPM makes perfect sense, and the RUM is only a particular case of the RPM. Similarly, the ARSP is a specific case of the WARSP if the preference function is defined as follows: $r^{*}_{\theta}(y, z)=u^{*}_{\theta}(y)-u^{*}_{\theta}(z)$. In this case, the event where a consumer picks an alternative $x_{i \mid t}$ from the menu $B_{t}^{*}$ is defined as follows:
$$ a_{i \mid t}(r^{*}_{\theta})=\mathbf{1}\left\lbrace  y_{t}\in \argmax_{y \in B_{t},~ p_{t}^{'}y=1} u^{*}_{\theta}(y)=x_{i \mid t} \right\rbrace$$
 The following corollary reflects the relationship between the test for stochastic rationality for the data set $S=\left\lbrace p_{t}, \rho_{t} \right\rbrace_{t \in \mathbb{T}}$ and the test for stochastic rationality for all pairs of observations.
 
\begin{corollary}{\label{cor1}}
Given $L \geq 2$ and the data set $S= \left\lbrace p_{t}, \rho_{t} \right\rbrace_{t \in \mathbb{T}}$ satisfy assumptions \ref{ass1} and \ref{ass2}; if the data are rationalized by the RPM, there is a unique probability vector $u_{ts} \in \Delta^{3}$ for all $t, s \in \mathbb{T}$, such that $\Gamma^{*} u_{ts}=\rho_{ts}$, where $\rho_{ts}=\left (\rho_{0 \mid t}, \rho_{1 \mid t}, \rho_{0 \mid s}, \rho_{1 \mid s}\right)^{'}$.
\end{corollary}
\begin{proof}
See Appendix.
\end{proof}

\begin{remark}
The inverse of Corollary \ref{cor1} is not true if I consider the following counterexample: three goods and three budgets.\footnote{I thank Victor Aguiar for providing me with this counterexample.} The distribution of demand and price are defined as follows:
$$\rho=\left(0, 1/2, 1/2, 0, 1/2, 0, 0, 1/2, 0, 1/2, 1/2, 0\right)^{'}$$
$P=(p_{1}, p_{2}, p_{3}),~~\text{where}~~p_{1}=\left(1/2, 1/4, 1/4\right),~~p_{2}=\left(1/4, 1/2, 1/4\right),~~\text{and}~~p_{3}=\left(1/4, 1/4, 1/2\right)$.
The matrix of WARP-consistent types $\Gamma$ has $12$ rows and $27$ columns. This rational matrix has the same number of rows but more columns than the rational matrix in KS that has $12$ rows and $25$ columns and is defined in the example $3.2$. I also consider the same set of prices like in the example $3.2$ of KS.\vspace{0.25cm}
I define the budget hyperplanes as follows: $B(p_{t})=\left\lbrace y_{t}=(y^{1}_{t}, y^{2}_{t}, y^{3}_{t}) \in \mathbb{R}^{3}_{+}: p_{t}^{'}y_{t}=1\right\rbrace,~t=1, 2, 3.$
I now prove that the rationalization by the RPM holds for all pairs of observations. The following inequalities between the choice probabilities result from the condition: $\Gamma \nu=\rho$.
\begin{itemize}
\item $B_{1}$ and $B_{2}$: $\rho_{0\mid 1}+\rho_{2 \mid 1} \geq  \rho_{3 \mid 2}$ and $\rho_{3 \mid 1} \leq \rho_{0 \mid 2}+\rho_{2 \mid 2}$
\item $B_{1}$ and $B_{3}$: $\rho_{2 \mid 3}+\rho_{3 \mid 3} \leq \rho_{0 \mid 1}+\rho_{1 \mid 1}$ and $\rho_{0 \mid 3} + \rho_{1 \mid 3} \geq \rho_{3 \mid 1}$
\item $B_{2}$ and $B_{3}$: $\rho_{2 \mid 3}+\rho_{0 \mid 3} \geq \rho_{3 \mid 2}$ and $\rho_{3 \mid 3} \leq \rho_{0 \mid 2}+\rho_{1 \mid 2}$
\end{itemize} \vspace{0.25cm}
From this condition, the following occurs:
$$\rho_{12}=
\left( {\begin{array}{cccc}
\rho_{0 \mid 1}+ \rho_{2 \mid 1}\\
\rho_{3 \mid 1}\\
\rho_{3 \mid 2} \\
\rho_{0 \mid 2}+ \rho_{2 \mid 2}
\end{array} }\right)~~\rho_{13}=
\left({\begin{array}{cccc}
\rho_{0 \mid 1}+ \rho_{1 \mid 1}\\
\rho_{3 \mid 1}\\
\rho_{0 \mid 3}+ \rho_{1 \mid 3} \\
\rho_{2 \mid 3}+ \rho_{3 \mid 3}
\end{array} }\right)~~\rho_{23}=
\left( {\begin{array}{cccc}
\rho_{0 \mid 2}+ \rho_{1 \mid 2}\\
\rho_{3 \mid 2}\\
\rho_{3 \mid 3} \\
\rho_{0 \mid 3}+ \rho_{2 \mid 3}\\
\end{array} }\right)$$
\vspace{0.5cm}
So, there are some vectors of probabilities $u_{1\mid 2}$, $u_{1\mid 3}$ and $u_{2\mid 3}$ such that:
$$\Gamma^{*} u_{1 2}=\rho_{1 2},~~\Gamma^{*} u_{13}=\rho_{13},~~\Gamma^{*} u_{2 3}=\rho_{23}$$

I now claim that the test for stochastic rationality does not apply if the vector of choice probabilities is $\rho=\left(0, 1/2, 1/2, 0, 1/2, 0, 0, 1/2, 0, 1/2, 1/2, 0\right)^{'}$. In fact, there is no positive vector $\nu \in \left[0, 1\right]^{27}$, so that $\Gamma \nu=\rho$. I obtain a computational certificate of infeasibility. See Appendix \ref{matlab}.\vspace{0.25cm}
This result shows that rationalization by the RPM in all pairs of observations is necessary but not sufficient. In other words, verifying the WARP in each pair of budget hyperplanes is necessary, but more is needed for this approach to be equivalent to verifying the WARP for all consumer types when considering all observations together.
\end{remark}

Given all the results that have been derived so far, under what conditions would the RPM and the RUM be equivalent when there are more than two goods? Recent work by \cite{cherchye2018transitivity} provides the necessary and sufficient conditions for price variation under which the WARP is equivalent to the GARP and the SARP. They refer to this condition as a \textit{triangular configuration}. It turns out that the RPM and the RUM can be equivalent if the set of normalized price vectors is a\textit{triangular configuration}.\vspace{0.25cm}

\begin{definition}\cite[][]{cherchye2018transitivity}
A set of normalized price vectors $P=\left\lbrace \tilde{p}_{t} \right\rbrace_{t \in \mathbb{T}}$, where $\tilde{p}_{t}=p_{t}/p_{t}y_{t}$,
is a \textit{triangular configuration} if for any three normalized price vectors $\tilde{p}_{t}$, $\tilde{p}_{s}$, $\tilde{p}_{k}$ in $P$, there are a number of $\lambda \in \left[0, 1 \right]$ and a permutation 
$\sigma: \left\lbrace t, s, k\right\rbrace \to \left\lbrace t, s, k\right\rbrace $ such that the following condition holds:
$$\tilde{p}_{\sigma(t)} \leq \lambda \tilde{p}_{\sigma(s)} + (1-\lambda) \tilde{p}_{\sigma(k)}~~\text{or}~~\tilde{p}_{\sigma(t)} \geq \lambda \tilde{p}_{\sigma(s)} + (1-\lambda) \tilde{p}_{\sigma(k)}$$
\end{definition}

Therefore, verifying whether a set of price vectors is a \textit{triangular configuration} simply requires checking the linear inequalities above for any possible combination of three price vectors. Under this condition, dropping the transitivity condition makes the WARP and the SARP indistinguishable from each other. \vspace{0.25cm}
The following theorem shows the equivalence between the RUM and the RPM.

\begin{proposition}{\label{prop2}}
  Given $L \geq 2$, if the data $S= \left\lbrace p_{t}, \rho_{t} \right\rbrace_{t \in \mathbb{T}}$ satisfy the assumptions \ref{ass1} and \ref{ass2} and the set of normalized price vectors $P=\left\lbrace \tilde{p}_{t} \right\rbrace_{t \in \mathbb{T}}$ is a \textit{triangular configuration}, the following are equivalent: 
\begin{enumerate}
\item[(i)]$S$ is rationalized by the RPM. \label{theo12}
\item[(ii)] $S$ is rationalized by the RUM.\label{theo22}
\end{enumerate}  
\end{proposition}
\begin{proof}
    See Appendix
\end{proof}
 This finding expands on the results of \cite{cherchye2018transitivity} regarding the equivalence between the WARP and the SARP in the stochastic environment and answers the question of when a rationalization by the RUM is empirically equivalent to a rationalization by the RPM when the number of goods exceeds two. Thus, dropping the transitivity condition makes the RUM and the RPM empirically equivalent under the \textit{triangular configuration}.
 
\subsection{Econometric Testing}
In this subsection, I discuss the statistical testing procedure. The second statement of Theorem \ref{theo1} simplifies the problem of testing RPMs to testing a null hypothesis concerning a finite vector of probabilities. Given $\left\lbrace a_{1}, \ldots, a_{H}\right\rbrace$ for the column vectors of $\Gamma$, the set $$\mathcal{C}=cone (\Gamma)=\left\lbrace \Gamma \nu=a_{1}\nu_{1}+ \ldots+ a_{H} \nu_{H}: \nu_{l} \geq 0 ~~\text{for all}~~l \in \left\lbrace 1, \ldots, H \right\rbrace \right\rbrace $$
is a finitely generated cone. Since the matrix of WARP-consistent types has more columns than the matrix of SARP-consistent types used in KS and DKSQ, the cone under the WARP is generated by more vectors than that under the SARP.\vspace{0.20cm} 
For Theorem \ref{theo1}, I want to test the following hypothesis: $H_{0}$: There are  $\nu \in \Delta ^{H-1}$ such that $ \Gamma \nu=\rho$. This null hypothesis is equivalent to
$$\min_{\nu \in \mathbb{R}^{H}_{+}}\parallel \rho-\Gamma \nu \parallel^{2}=0$$
 The solution $\eta_{0}$ of $H_{0}$ is the projection of $\rho \in \mathbb{R}^{I}_{+}$ onto the cone $\mathcal{C}$ under the weighted norm $\parallel y\parallel=\sqrt{y^{t} y}$. \footnote{$\rho \in \mathbb{R}^{d_{\rho}}_{+}$ with $d_{\rho}$ being the number of patches of the rational demand matrix $\Gamma$.} The objective function's corresponding value is the projection residual's square. Note that $\eta_{0}$ is unique, but the corresponding $\nu$ is not; stochastic rationality holds when the length of the residual is zero. An inherent sample representation of the objective function in $H_{0}$ would be 
$\min_{\eta \in  \mathcal{C}}\parallel \hat{\rho}-\eta \parallel^{2}$, where $\hat{\rho}$ estimates $\rho$, for example, by sample choice frequencies. Thus, as in KS, the scaling yields:
\begin{align*}\label{test}
\mathcal{J}_{N} &= N\min_{\eta \in \mathcal{C}}\parallel \hat{\rho}-\eta \parallel^{2}
\end{align*}
where $N=\min_{t}N_{t}$ with $N_{t}$ being the number of observations on budget $B_{t}$. Note that $\nu$ is not unique at the optimum, while $\eta$ is. The null hypothesis occurs when a population of rational consumers generates a sample of cross-sectional demand distributions. In other words, the estimated distributions of the preference functions may have arisen from an RPM up to sampling uncertainty. I use KS's modified bootstrap method to calculate the critical values for this test (See it described here \ref{stat}).\vspace{0.25cm}

\section{Empirical Application}
To illustrate the proposed methodology, I use a subsample data set from the UK Family Expenditure Survey (FES). The FES data are a repeated cross-sectional survey comprising around $7,000$ households who report their consumption expenditures for various commodity groups each year. These data are frequently used for the nonparametric demand estimations (See, for instance, \citealp{blundell2008best}; \citealp{kitamura2018nonparametric}; \citealp{deb2023revealed}). I analyze the coarsest partition of three and five goods and maintain the same categorization as in KS. I consider food, services, and nondurable goods for a 3-dimensional commodity space. For a 5-dimensional commodity space, I consider food, services, nondurable goods, clothing, and alcohol. My analysis focuses on consumers meeting specific criteria, including owning at least one car, having one child, having an annual nonzero income, and spending on the nonzero goods under consideration. The number of data points varies from $2087$ (in 1975) to $1833$ (in $1999$) for a total of $46,391$. I have a larger data set than KS because I consider both couples and singles households.\footnote{The number of data points used in KS varies from $715$ (in $1997$) to $1509$ (in $1975$) for a total of $26,341$. My choice for couples and singles households is motivated by \cite{aguiar2021stochastic} and \cite{browning2010uncertainty}, who found that couples and singles households are more likely to inconsistently behave with the dynamic utility maximization theory provided some values of exponential discounting.} \vspace{0.20cm}

Before testing the RUM and the RPM on the previous subsample of the FES data set, similar to \citealp{blundell2008best}, \cite{kitamura2018nonparametric}, and \cite{deb2023revealed} I assume that consumers are subjected to identical prices within a year, and I use similar price data. I represent each consumer's income using total expenditures on the goods under consideration and the income per period by the median expenditure of all consumers in that period. This is because budgets overlap considerably at the median expenditure level. Since the demand distributions are estimated for a fixed (mean) level of total expenditure, I also use an instrumental variable procedure with a control function approach to adjust for the dependence of total expenditure on prices.\footnote{The test proposed by \cite{mcfadden1990stochastic} requires the observation of large samples of consumers who face the exact prices and have identical total expenditures on the observed goods. However, this property does not hold where consumer demand (and thus total expenditures) for the observed goods is usually price-dependent.} Therefore, following the work of KS and \cite{blundell2008best}, I also use the annual household income as an instrumental variable. \footnote{I also use KS's replication file in which the procedure for nonparametric estimation of instrumental variables is well described. I further modify the Floyd-Warshall algorithm to detect only cycles of length two.}\vspace{0.25cm}

Tables \ref{tab:mytable1} and \ref{tab:mytable2} hold summaries of the empirical results for both the RPM and the RUM. In Table \ref{tab:mytable1}, I present the results for the three composite goods and the blocks of seven consecutive periods. In Table \ref{tab:mytable2}, I increase the dimensionality of the commodity space to five commodities, but I reduce the number of consecutive periods to six. Both tables display test statistics ($\mathcal{J}_{N}$), p-values ($\textit{p}$), the number of patches $\textit{I}$, and the number of demand vectors $\textit{H}$ that can be rationalized. The rational demand matrix has the size ($\textit{I} \times \textit{H}$). The findings are consistent with the theory as the rational matrix under the WARP has the same number of rows but more columns than the rational matrix under the SARP. \vspace{0.25cm}

Overall, in Tables \ref{tab:mytable1} and \ref{tab:mytable2}, the estimated choice probabilities are consistent with the RPM, and the non-rejection of the null hypothesis ${H}_{0}$ is also statistically significant. These results indicate strong evidence for the null hypothesis. However, the estimated choice probabilities are not always consistent with the RUM. This is true in entries 86--92, 87--93, and 88--94 of Table \ref{tab:mytable1}. The p-value of each entry is less than $0.05$, indicating evidence against the null hypothesis. Since I ran $19$ test statistics to obtain these results, I further adjusted all p-values with the Bonferroni correction. Specifically, I multiplied each p-value by the number of test statistics. All adjusted p-values that exceed one were then reduced to one. Thus, for entries 88--94 of Table \ref{tab:mytable1}, the estimated choice probabilities were still inconsistent with the RUM, and the rejection of the null hypothesis ${H}_{0}$ was statistically significant at the $10\%$ level.

The same observation holds for entries 87--92 and 88--93 in Table \ref{tab:mytable2} for the RUM in which the p-value of each entry is less than $0.05$, indicating strong evidence against the null hypothesis. Applying the Bonferroni correction once again to the p-value $0.002$ generates an adjusted p-value of $0.04$.Thus, the estimated choice probabilities in the entries 88--93 are still inconsistent with the RUM, and the rejection of the null hypothesis is also statistically significant at the $5\%$ level. \vspace{0.25cm}

Overall, my results support the hypothesis that consumers behave in accordance with the RPM. In contrast, the RUM cannot describe the behavior of the population. My results are also consistent with Proposition $1$ because RPMs explain the data when RUMs do. The fact that the RUM could not explain household behaviors while the RPM could confirm that households can sometimes exhibit intransitive behaviors.\vspace{0.25cm}

It follows from these results that RPMs are not only weaker than RUMs but also empirically more successful. To the best of my knowledge, this finding is the first empirical evidence of the superiority of the WARSP over ARSP when the dimension of the commodity space is greater than two. Before closing this section, it is important to mention that this paper is one of many that provide empirical evidence against the RUM. For example, \cite{aguiar2018does} found that the RUM cannot explain population behavior when consideration of all alternatives is costly.

\begin{table}
  \centering
  \begin{tabular}{ccccccc}
    \hline
    \hline
    & \multicolumn{3}{c}{\textbf{RPM}} & \multicolumn{3}{c}{\textbf{RUM}} \\
    \cmidrule(lr){2-4}
    \cmidrule(lr){5-7} 
    Periods &$\Gamma$&$\mathcal{J}_{N}$ & $\textit{p}$ &$\Gamma$& $\mathcal{J}_{N}$ & $\textit{p}$ \\
   \hline \\
   $75-81$ & $33 \times 39859$ & $0.0033$ & $0.661$&$33 \times 3666$&$0.646$& $0.289$\\
$76-82$ & $35 \times 32061$ & $0.0043$ & $0.759$ &$35 \times 3620$&$0.8582$&$0.396$ \\
$77-83$ & $39 \times 101610$ & $0.0019$ & $0.530$&$39 \times 8410$&$3.3263$& $0.138$ \\
$78-84$ & $39 \times 122479$ & $0.002$ & $0.368$&$39 \times 7479$& $5.5957$&$0.163$ \\
$79-85$ & $40 \times 122589$ & $0.003$ & $0.613$ &$40 \times 10486$&$8.5847$& $0.095$\\
$80-86$ & $34 \times 50016$ & $0.0029$ & $0.301$ &$34 \times 4931$& $10$&$0.116$\\
$81-87$ & $20 \times 540$ & $0.0411$ & $0.059$&$20 \times 175$& $5.9795$&$0.17$ \\
$82-88$ & $13 \times 36$ & $\num{5.57e-04}$ & $0.255$&$13 \times 21$&$2.2138$& $0.193$ \\
$83-89$ & $11 \times 12$ & $\num{5.61e-04}$ & $0.216$&$11 \times 7$&$2.1799$& $0.198$ \\
$84-90$ & $11 \times 10$ & $2.1509$ & $0.195$&$11 \times 7$&$2.1509$ & $0.195$ \\
$85-91$ & $11 \times 12$ & $0$ & $1$ &$11 \times 7$&$0.5163$& $0.255$\\
$86-92$ & $24 \times 1680$ & $0.0013$ & $0.297$&$24 \times 333$&$10.9469$& $0.011$ \\
$87-93$ & $39 \times 91625$ & $0.0012$ & $0.348$&$39 \times 6780$& $9.1404$& $0.007$\\
$88-94$ & $45 \times 341579$ & $0.0011$ & $0.35$&$45 \times 17446$&$10.7697$&  $0.005$\\
$89-95$ & $37 \times 33912$ & $0.0018$ & $0.05$&$37\times 2890$&$0.0018$& $0.072$ \\
$90-96$ & $29 \times 7472$ & $0$ & $1$ &$29\times 917$&$0$& $1$\\
$91-97$ & $25 \times 2112$ & $0$ & $1$&$25 \times 475$& $0.0105$& $0.273$\\
$92-98$ & $19 \times 288$ & $0$ & $1$&$19 \times 108$&$0.006$ & $0.547$ \\
$93-99$ & $13 \times 32$ & $0$ & $1$&$13 \times 15$&$0.0067$& $0.284$ \\
  \hline
    \hline
  \end{tabular}
  \caption{Testing Results: Three goods and blocks of seven consecutive periods}
  \label{tab:mytable1}
\footnotesize{I = number of patches; H = number of rationalizable discrete demand vectors; $\Gamma=I \times H$; $\mathcal{J}_{N}=\text{test statistic}$; p = p-value. Number of bootstrap replications = 1, 000.}
\end{table}
\begin{table}
  \centering
  \begin{tabular}{ccccccc}
    \hline
     \hline
    & \multicolumn{3}{c}{\textbf{RPM}} & \multicolumn{3}{c}{\textbf{RUM}} \\
    \cmidrule(lr){2-4}
    \cmidrule(lr){5-7}
   Periods &$\Gamma$&$\mathcal{J}_{N}$ & $\textit{p}$ &$\Gamma$& $\mathcal{J}_{N}$ & $\textit{p}$ \\
   \hline\\
   $75-80$ & $33 \times 15920$ & $0.001$ & $0.323$&$33 \times 1907$&$0.0092$& $0.495$\\
 $76-81$ & $33 \times 12969$ & $0.0073$ & $0.568$&$33 \times 1817$&$1.242$& $0.113$\\
   $77-82$ & $43 \times 85848$ & $\num{7.82e-04}$ & $0.427$&$43 \times 7380$&$1.0446$& $0.188$\\
      $78-83$ & $50 \times 252739$ & $\num{2.81e-04}$ & $0.409$&$50 \times 10730$&$1.1463$& $0.1920$\\
      $79-84$ & $49 \times 147867$ & $\num{3.26e-04}$ & $0.481$&$49 \times 9161$&$0.0044$& $0.822$\\
          $80-85$ & $41 \times 66428$ & $0.0012$ & $0.18$&$41 \times 5033$&$0.0031$& $0.752$\\
          $81-86$ & $39 \times 47603$ & $\num{4.66e-04}$ & $0.2$&$39 \times 4176$&$0.0135$& $0.62$\\
     $82-87$ & $20 \times 576$ & $0$ & $1$&$20 \times 145$&$0.0127$& $0.602$\\
       $83-88$ & $12 \times 27$ & $0$ & $1$&$12 \times 11$&$0.0141$& $0.562$\\
       $84-89$ & $12 \times 18$ & $0$ & $1$&$12 \times 11$&$0.014$& $0.563$\\
        $85-90$ & $8 \times 3$ & $0$ & $1$&$8 \times 3$&$0$& $1$\\
         $86-91$ & $12 \times 32$ & $0$ & $1$&$12\times 15$&$0.9592$& $0.184$\\
          $87-92$ & $25\times 2991$ & $0$ & $1$&$25 \times 435$&$11.8048$& $0.006$\\
      $88-93$ & $52\times 263485$ & $\num{2.9e-05}$ & $0.842$&$52 \times 12486$&$11.613$& $0.002$\\
   $89-94$ & $60 \times 635742$ & $\num{2.8e-05}$ & $0.888$&$60 \times 26982$&$0.4917$& $0.309$\\
    $90-95$ & $45 \times 91902$ & $0$ & $1$&$45 \times 3818$&$0.0324$& $0.466$\\
      $91-96$ & $53 \times 165334$ & $0.0232$ & $0.395$&$53 \times 15131$&$0.0234$& $0.481$\\
       $92-97$ & $24\times 854$ & $0$ & $1$&$24 \times 269$&$0$& $1$\\
        $93-98$ & $14\times 96$ & $0$ & $1$&$14 \times 37$&$0.005$& $0.432$\\
         $94-99$ & $12 \times 36 $ & $0$ & $1$&$12 \times 17$&$0.0056$& $0.456$\\
   \hline
    \hline
  \end{tabular}
  \caption{Testing Results: Five goods and blocks of six consecutive periods}
  \label{tab:mytable2}
\footnotesize{I = number of patches; H = number of rationalizable discrete demand vectors; $\Gamma=I \times H$; $\mathcal{J}_{N}=\text{test statistic}$; p = p-value. Number of bootstrap replications = 1, 000.}
\end{table}

\subsection{Monte Carlo Simulations}
The previous empirical application shows that the RPM is empirically more successful than the RUM if the dimension of the commodity space is at least three. Unlike the RUM, the RPM is not consistently rejected by the UK data that can be attributed to various factors. Most notably, the RPM is less restrictive than the RUM. To confirm that my testing procedure has power against the RPM, I assess the performance of my model using Monte Carlo simulations. In particular, I show that my test can reject the false null hypothesis of data being consistent with the RPM in finite samples comparable to the sample size of the population of consumers considered in my empirical application. \vspace{0.25cm}

I run the simulations using the $12$ by $27$ matrix of WARP-consistent types of $\Gamma$ (referenced as \ref{matrix}) that I  used in the previous counterexample in Section $4$. This matrix differs from the one used in KS as it has two additional columns.\footnote{\cite{kitamura2018nonparametric}'s $12$ by $25$ rational matrix is a matrix of SARP-consistent types on which they tested the WARP. In their example $3.1$, consumer types are consistent with the WARP but not with the SARP.} I also consider the matrix $X$ used in KS's Monte Carlo simulations. Data on choice probabilities are generated using $\rho_{true}$, where $\rho_{true}$ is the weighted sum of $\rho_{inside}$ and $\rho_{outside}$, calculated as:
$$\rho_{true}=0.95 \rho_{inside}+ 0.05 \rho_{outside}$$

$\rho_{\text{inside}}$ and $\rho_{\text{outside}}$ are defined such that $\rho_{\text{inside}}$ strongly satisfies the null hypothesis, while $\rho_{\text{outside}}$ does not satisfy the null hypothesis. I use and modify KS's replication file for the Monte Carlo simulations. I consider four sample sizes---$N_{t} \in \left\lbrace 100, 200, 500, 1,000, 2,500\right\rbrace$---and perform $499$ simulations with a bootstrap size of $1,000$.  Table \ref{tab:mytable3} presents the results for the power simulations. It displays the proportion of rejected simulations at the $5\%$ and $1\%$ significance levels. As expected, the fraction of rejections is bigger for larger samples. For $N_{t} \in \left\lbrace 200, 500, 1,000, 2,500 \right\rbrace$, the rejection rate is $100\%$. Thus, as the population increases, so should the probability of rejecting the null that the populations’ behavior is generated by the RPM. These results show that the RPM can be rejected with a power of one with a comparable sample size as in my empirical application.
\begin{table}
  \centering
  \begin{tabular}{cccc}
    \hline
     \hline
  Sample size  & $\alpha=5\%$ & $\alpha=1\%$\\
    \hline\\
     $N=100$ & $0.99$ & $0.952$ \\
   $N=200$ & $1$ & $1$ \\
    $N=500$ & $1$ & $1$  \\
     $N=1,000$ & $1$ & $1$  \\
     $N=2,500$ & $1$ & $1$  \\
      \hline
     \hline 
  \end{tabular}
  \caption{Proportion of rejected simulations. Number of bootstrap replications = 1,000. Number of simulations = 499}
  \label{tab:mytable3}
\end{table}

\section{Welfare Comparisons and Counterfactuals in RPMs}
The natural next steps after the rationality check are the extrapolation of the (bounds on) counterfactual demand distributions and the welfare analysis. Consider repeated cross-sectional data that are consistent with the RPM but inconsistent with the RUM. In this situation, the interpretation of the classical stochastic revealed-preference relation is unclear, yet I show in this section that meaningful welfare and counterfactual analyses are possible. The main results are an analogy of the results of DKSQ and \cite{kitamura2019nonparametric} for the RAUM (welfare comparisons) and the RUM (nonparametric counterfactuals), respectively. \vspace{0.20cm}

\subsection{Welfare Comparisons}
With the help of the rationalization test in which a distribution across different consumer types is determined, I can infer the preference distribution in the population. Such insights are valuable for the welfare analysis, especially when considering whether the substitution or addition of certain products affects consumer welfare. Suppose a government authorizes the sale of a particular good on the market that was previously prohibited, such as cannabis (marijuana). What matters for the government's prospects of being re-elected is the proportion of consumers who are better off as a result of this change in their choice set. I develop a method below to obtain information about this proportion. \vspace{0.25cm}

The next definition introduces the patch-revealed dominance criterion. It is a binary relation between the elements of $X^{*}=\cup_{i \in \mathcal{I}_{t}, t \in \mathbb{T}} \left\lbrace x_{i \mid t} \right\rbrace$ that are complete but not necessarily transitive. With this criterion, I can compare any pair of patches belonging to different hyperplane budgets.
 \begin{definition}[Patch-Revealed Dominance]
    I say that patch $x_{i \mid t}$ is dominant over $x_{i^{'} \mid s}$ and is denoted by $x_{i \mid t} > x_{i^{'} \mid s}$, if for all $y \in x_{i \mid t}$, $p_{t}y_{t} \geq p_{t}z_{s}$ and $p_{s}z_{s}<p_{s}y_{t}$ for all $z_{s} \in x_{i^{'} \mid s}$.
\end{definition}
Since patches are regarded as alternatives, $x_{i \mid t} > x_{i^{'} \mid s}$ means that the best nonstochastic choice of $x_{i \mid t}$ is directly revealed to be preferred over the best nonstochastic choice in $x_{i^{'} \mid s}$.\footnote{As all nonstochastic demand systems belonging to the same patch lead to the same directly revealed preference result, the patch-revealed dominance criterion means that all bundles in the dominant patch are directly revealed to be preferred to those in the dominated patch.}  Figure \ref{fig1} illustrates that $x_{0 \mid 2}$ is dominant over $x_{0 \mid 1}$  and $x_{1 \mid 1}$ is dominant over $x_{1 \mid 2}$. \vspace{0.20cm}

\begin{definition}
Given a rational type $\theta$, the budget $B_{t}^{*}=\left\lbrace x_{i \mid t}\right\rbrace_{i \in \mathcal{I}_{t}}$ is revealed to be preferred to budget $B_{s}^{*}=\left\lbrace x_{i \mid s}\right\rbrace_{i \in \mathcal{I}_{s}}$ that is denoted as $B_{t}^{*} >^{\theta} B_{s}^{*}$, if there are $l \in \mathcal{I}_{t}$ and $k \in \mathcal{I}_{s}$ such that $x_{l \mid t}$ is revealed to be dominant over $ x_{k \mid s}$, $x_{l \mid t} >^{\theta}  x_{k \mid s}$.
\end{definition}

According to this definition, since the revealed preference compares subsets of hyperplane budgets, if budgets $B_{t}^{*}$ and $B_{s}^{*}$ overlap; then budget $B_{t}^{*}$ can be revealed to be preferred to the budget $B_{s}^{*}$ and vice versa.\footnote{The notation $B_{t}^{*} >^{\theta} B_{s}^{*}$ also means that the consumer is better off at $B_{t}^{*}$ than at $B_{s}^{*}$. If the two hyperplane budgets $B_{t}^{*}$ and $B_{s}^{*}$ do not overlap, I can still apply the Path-Revealed Dominance criterion. I do not consider intersection patches in this analysis because the revealed dominance criterion may no longer be consistent.} To determine the welfare effect of a change in the choice set or menu from $B_{s}^{*}$ to $B_{t}^{*}$, let $ \mathbbm{1}_{B_{t}^{*} >^{\theta} B_{s}^{*}}$ denote the row vector with its length equal to the number of rational types such that the t-th element is one if, for the rational type corresponding to column $t$ of $\Gamma$, $B_{t}^{*}$ is revealed to be preferred to $B_{s}^{*}$ and zero otherwise. Therefore, $ \mathbbm{1}_{B_{t}^{*} >^{\theta} B_{s}^{*}}$ enumerates the set of rational types for which $B_{t}^{*}$ is revealed to be preferred to $B_{s}^{*}$.\vspace{0.20cm}

\begin{definition}
Given a rational type $\theta$, if $\mu$ is the true distribution rationalizing data $\left\lbrace p_{t}, then y_{t}(\theta) \right\rbrace_{t \in \mathbb{T}}$ is the proportion of consumers who are revealed to be better off under budget set $B_{t}^{*}$ than under budget set $B_{s}^{*}$. This proportion comes from $\mu \left(\left\lbrace  \theta \in \Theta: (l, k) \in \mathcal{I}_{t} \times \mathcal{I}_{s} ~\text{such that}~x_{l \mid t} >^{\theta}  x_{k \mid s} \right\rbrace \right)$.
\end{definition}
If the true distribution $\mu$ is known, this definition can be used to identify the proportion of consumers who prefer the choice set $B_{t}^{*}$ to the choice set $B_{s}^{*}$. I can partially identify this proportion since the true distribution $\mu$ is not always observed. The following proposition allows me to achieve this identification. 

\begin{proposition}[Welfare Comparison]{\label{prop3}}
Given $L \geq 2$, let $S=\left\lbrace \rho_{t}, p_{t} \right\rbrace_{t \in \mathbb{T}}$ be a repeated cross-section of data rationalized by the RPM and satisfying the assumptions \ref{ass1} and \ref{ass2}.
\begin{itemize}
 \item [(i)] Then, for every $\lambda \in \left[\underline{\gamma}, \bar{\gamma} \right] $, there is a rationalization of $S$ for which $\lambda$ is the proportion of consumers who are revealed better off under budget set $B_{t}^{*}$ than under budget set $B_{s}^{*}$.
 \item [(ii)] For any rationalization of $S$, there is a proportion of consumers who are better off at $B^{*}_{t}$ compared to $B_{s}^{*}$. This proportion can have any value in the interval $\left[\underline{\gamma}, 1-\underline{\beta}\right].$ 
\end{itemize}
\begin{align*}
\bar{\gamma}  & = \max_{\nu} \left\lbrace \mathbbm{1}_{B_{t}^{*} >^{\theta} B_{s}^{*}} \cdot \nu\right\rbrace,~~\text{subject to} ~\Gamma\nu=\rho,~\nu \geq 0; \\ 
\underline{\gamma}  & = \min_{\nu} \left\lbrace \mathbbm{1}_{B_{t}^{*} >^{\theta} B_{s}^{*}} \cdot \nu\right\rbrace,~~\text{subject to}~\Gamma\nu=\rho,~\nu \geq 0; \\
\underline{\beta}  & = \min_{\nu} \left\lbrace \mathbbm{1}_{B_{s}^{*} >^{\theta} B_{t}^{*}} \cdot \nu\right\rbrace,~~\text{subject to}~\Gamma\nu=\rho,~\nu \geq 0. 
\end{align*}
Furthermore, $\underline{\gamma} $ and $\bar{\gamma}$ are, respectively, the lower and upper bounds of the proportion of consumers who are revealed to be better off under budget set $B_{t}^{*}$ than under budget set $B_{s}^{*}$. These welfare bounds cannot be improved without further information.
\end{proposition}
This proposition is analogous to proposition $4$ of DKSQ that also deals with the recoverability of random preferences. The main difference is that DKSQ's proposition partially identifies the proportion of consumers who prefer one price over another and uses the generalized axiom of price preference (GAPP) to define rationality over prices that require preference relations to be transitive.

\begin{example}
 In the example \ref{figur3}, in rational type $1$, $x_{0 \mid 2} >^{\theta_{1}} x_{0 \mid 1}$ and the proportion of consumers who are revealed to be better off with $B_{2}^{*}$ than $B_{1}^{*}$ is $\rho_{0 \mid 1}$. In rational type $2$, $x_{1 \mid 1} >^{\theta_{2}} x_{1 \mid 2}$ and the proportion of consumers who are revealed to be better off with $B_{1}^{*}$ than $B_{2}^{*}$ is $\rho_{1 \mid 2}$. The point identification of the parameter of interest takes place because the rational matrix has a full rank. However, the proportion of consumers who are better off with $B_{2}^{*}$ than $B_{1}^{*}$ is $\left[\rho_{0 \mid 1}, \rho_{0 \mid 2}\right]$ and the proportion of consumers who are better off with $B_{1}^{*}$ than $B_{2}^{*}$ is $\left[\rho_{1 \mid 2}, \rho_{1 \mid 1}\right]$.
\end{example}

\subsection{Nonparametric Counterfactuals in RPMs}
In this subsection, I bound the features of counterfactual choices in the nonparametric RPM of demand when observable choices are repeated across sections with unrestricted, unobserved heterogeneity. Given a rational type $\theta$, I take rationalizable stochastic choices $\rho=(\rho_{1}, \ldots, \rho_{T})$ as given, and ask what discipline they place on
$$\rho_{0}(x_{i \mid 0})=\int \mathbf{1}\left\lbrace y(\theta) \in x_{i \mid 0}:~r_{\theta}(y(\theta), z(\theta)) > 0, ~p_{0}y(\theta) \geq p_{0}z(\theta),~\forall z(\theta) \in B_{0} \right\rbrace d\mu(r),~\forall~i \in \mathcal{I}_{0}$$
where $\rho_{0}(x_{i \mid 0})$ is the stochastic choice at some counterfactual budget $B_{0}^{*}=\left\lbrace x_{i \mid 0}\right\rbrace_{i \in \mathcal{I}_{0}}$ that corresponds to the counterfactual price $p_{0}$. $B_{0}=\left\lbrace y \in X: p_{0}y \leq w_{0}\right\rbrace$ is the budget set in period $0$. As in \cite{kitamura2019nonparametric}, this discipline will typically take the form of bounds, and a distribution $\rho_{0}$ of demands on $B^{*}_{0}$ is inside the bounds if $(\rho_{0}, \rho_{1} \ldots, \rho_{T})$ are jointly consistent with the RPM.\vspace{0.25cm}
Consider a repeated cross-sectional data set $S= \left\lbrace p_{t}, \rho_{t}\right\rbrace_{t \in \mathbb{T}}$ in which each observation consists of the prevailing price $p_{t}$ and each observation corresponds to the distribution of choices made by a population of consumers.The data set $S$ satisfies the assumptions \ref{ass1} and \ref{ass2} and is rationalized by the RPM.
The following theorem is analogous to Theorem $3$ of \cite{kitamura2019nonparametric} that deals with nonparametric counterfactuals in RUMs. This theorem bounds the expectation of the (functionals of) counterfactual stochastic demand. 
\begin{theorem}{\label{theor3}}
Given a rational type $\theta$ and any known function $h \colon \mathbb{R}^{L} \to  \mathbb{R} $ is bounded on $B_{0}^{*}$. 
\begin{align*}
\underline{h}_{i \mid 0}=\inf\left\lbrace h(y(\theta)): y(\theta) \in x_{i \mid 0}\right\rbrace, ~ i \in \mathcal{I}_{0}\\
\bar{h}_{i \mid 0}=\sup \left\lbrace h(y(\theta)): y(\theta) \in x_{i \mid 0}\right\rbrace, ~i \in \mathcal{I}_{0}
\end{align*}
Then the lower and upper bounds of $\mathbb{E}\left(h(y(\theta))\right) $ are defined as follows:
\begin{align*}
 \min \left\lbrace \left(\underline{h}_{0 \mid 0},\ldots~ \underline{h}_{I_{0} \mid 0}\right)  \Gamma_{0} \nu:~ \tilde{\Gamma} \nu=\tilde{\rho},~ \nu \geq 0 \right\rbrace  \\
  \leq \mathbb{E}\left(h(y(\theta))\right) \leq \\
   \max \left\lbrace \left(\bar{h}_{0 \mid 0},\ldots~ \bar{h}_{I_{0} \mid 0}\right)  \Gamma_{0} \nu:~ \tilde{\Gamma}\nu=\tilde{\rho},~ \nu \geq 0 \right\rbrace, 
\end{align*}
 
 where $\tilde{\rho}=\left[\rho_{1} \ldots \rho_{T}\right]^{'}$ and 
  $\tilde{\Gamma}=\left[\Gamma_{1} \ldots \Gamma_{T}\right]^{'}$
  
 They are sharp and can only be improved on with further information.
\end{theorem}
I can use the result from the previous theorem to bound the expected value of $h(y(\theta))=z^{'}_{0}y(\theta)$, where $z_{0} \in \mathbb{R}^{L}_{+}$ is a vector specified by the user. For example, $z_{0}= (1, 2,\ldots,0)^{'}$ can be thought of as extracting the demand for goods $1$ and $2$, and the upper and lower bounds are the welfare bounds on the expected stochastic demand for goods $1$ and $2$. $z_{0}= (p_{0}^{1}/p_{0}y(\theta), 0, \ldots, 0)^{'}$ can be thought of as extracting the share of expenditure on good $1$, and then the previous upper and lower bounds are the welfare bounds for the share of expenditure on good $1$. \vspace{0.25cm}

It is important to note that the bounds derived in Proposition \ref{prop3} and Theorem \ref{theor3} would be more informative if the data $S$ were RUM-rationalized due to the additional constraint imposed by the transitivity assumption. If the data $S$ is rationalized by the RUM and the RPM, the bounds derived in Proposition \ref{prop3} and Theorem \ref{theor3} and their counterparts under the RUMs will usually be nested. This is because replacing the WARP with the SARP does not significantly improve the bounds of revealed preferences \cite[][]{blundell2015sharp, cosaert2018nonparametric, aguiar2021cardinal}. Moreover, \cite{aguiar2021cardinal} in experimental and scanner datasets find that the performance of WARP is relatively equivalent to that of GARP.

\begin{remark}
I can improve the welfare bounds as well as the bounds derived in Proposition \ref{prop3} and Theorem \ref{theor3} by introducing \textit{triangular configuration}. To achieve this improvement, I must construct and work with another rational demand matrix in which each column obeys the WARP and the \textit{triangular configuration}. Due to the additional price restrictions, the resulting rational demand matrix will undoubtedly have fewer columns than the matrix of WARP-consistent types of $\Gamma$. Ultimately, I can obtain more informative bounds that are still robust to deviations from the transitivity of preferences.
\end{remark}

\section{Conclusion}
This paper extends the scope of the preference function models by adding random preference models. To characterize the latter, I have developed a WARP-based approach in a setting where heterogeneity of preferences is not directly observable. I show that the RPM characterization is amenable to statistical testing using KS’s statistical tools. The null hypothesis to be tested was that the data were generated by an RPM characterizing a heterogeneous population in which the only behavioral constraints were "more is better" and the WARP. Based on an analysis of a subset of the FES data, I find evidence that contradicts the RUM and supports the RPM. In addition, I have shown how the results of the RPM test can help identify nontransitive preferences when there is variation in the menu. \vspace{0.20cm}

My approach has three attractive features. First, the results are completely nonparametric, which means they are independent of the functional form of the underlying preference functions. Second, I impose minimal requirements on the structure of individual unobserved heterogeneity. Third, I do not restrict the number of goods. In contrast to the approaches of \cite{deb2023revealed} and \cite{kitamura2019nonparametric}, my analysis of welfare and counterfactuals does not rely on the full transitivity requirements. Therefore, it can provide wider bounds, although these can be improved with further information or constraints such as \textit{triangular configuration}. \vspace{0.20cm}

Overall, RPMs prove to be robust alternatives to traditional RUMs as they provide valuable insights into welfare analysis that go beyond what RUMs can provide. These results should be helpful for practitioners of revealed preferences, as WARP is much easier to handle than SARP and is more successful from an empirical perspective.\vspace{0.25cm}

\section{Appendix}
\textbf{Proof of the Lemma \ref{lem1}}
\begin{definition}\cite[]{matzkin1991testing}
    Consider a finite data set  $O^{T}=\left\lbrace p_{t}, y_{t} \right\rbrace_{t \in \mathbb{T}} $. The
following statements are equivalent:
\begin{itemize}
    \item[(i)] The data $O^{T}$ can be strictly rationalized by a utility function.
    \item [(ii)] The data $O^{T}$ satisfy SARP.
    \item [(iii)] The numbers $u_{t}>0$ and $\lambda_{t} > 0$ for all $t \in \mathbb{T}$ exist such that the inequalities:
    \begin{align*}
        \text{if}~ y_{t} & \neq y_{s}, ~\text{then}~ u_{t}-u_{s} > \lambda_{t}p_{t} \left(y_{t}-y_{s}\right), \\
          \text{if}~ y_{t} & = y_{s}, ~\text{then}~ u_{t}-u_{s} =0, 
    \end{align*}
    hold for all $s, t \in \mathbb{T}$.
(iv) The numbers $v_{t}$ for all $t \in \mathbb{T}$ exist such that the inequalities: 
 \begin{align*}
        \text{if}~ y_{t} & \neq y_{s} ~\text{and}~p_{t} \left(y_{t}-y_{s}\right)\geq 0~ \text{then}, ~v_{t}-v_{s} >0,\\
          \text{if}~ y_{t} & = y_{s}, ~\text{then}~ v_{t}-v_{s} =0, 
    \end{align*}
      hold for all $s, t \in \mathbb{T}$.
\end{itemize}
\end{definition}
\begin{proof}
The proof of this theorem is adapted from \cite{aguiar2020rationalization}. \vspace{0.25cm}
Assume that the WARP condition holds. If so, then I can show that a continuous, skew-symmetric, strictly piecewise concave preference function rationalizes the data. 
If the WARP condition holds, then the data $O^{T}$ can be broken into a pairwise data set $O^{2}_{ts}$. Consider the $T^{2}$ data sets $O^{2}_{ts}$ for each pair of observations $s, t \in \mathbb{T}$. Since $O^{2}_{ts}$ obeys SARP, it follows from the theorem of \cite{matzkin1991testing} that $O^{2}_{ts}$ can be strictly rationalized by a continuous, strictly increasing and strictly concave utility function. The definition of it is $u_{st}: X \to \mathbb{R}$. Now, I define the mapping of $r_{st}: X \times X \to \mathbb{R}$ as follows:
\[ r_{st}(x, y)= \left\{
\begin{array}{ll}
 u_{st}(x)-u_{st}(y) & s \neq t \\
 p_{t}((x-y)-\epsilon \left(g(x-x_{t})-g(y-x_{t})\right) & s=t \\
\end{array}
\right. \]
for any small $\epsilon$ and where the function $g$ is defined as in \cite{matzkin1991testing}. Each function $r_{st}$ is continuous, strictly concave, skew-symmetric, and strictly increasing because the continuous $u_{st}$ strictly increases and is strictly concave. Next, I define the preference function $r$ as follows: 
For all  $x, y \in X$, 
\begin{align*}
r(x, y)=\max_{\lambda \in \Delta}\min_{\mu \in \Delta}\displaystyle\Sigma _{s=1}^{T}\displaystyle\Sigma _{t=1}^{T} \lambda_{s}\mu_{t}r_{st}(x, y) \\
r(x, y)=\min_{\mu \in \Delta} \max_{\lambda \in \Delta}\displaystyle\Sigma _{s=1}^{T}\displaystyle\Sigma _{t=1}^{T} \lambda_{s}\mu_{t}r_{st}(x, y) \\
\end{align*}
where $\Delta=\left\lbrace \lambda \in \mathbb{R}^{L}_{+}: \sum_{t=1}^{T} \lambda_{t}=1 \right\rbrace$ is a $T-1$ dimensional simplex.
Next, I prove that $r$ strictly rationalizes the data set $O^{T}$. Consider $y \in X$ and some fixed $t\in \mathbb{T}$ such that $y_{t}=y$ and $p_{t}y_{t} \geq p_{t}y$. Let $\mu \in \Delta$ is the vector such that $\mu_{j}^{t}=0$ if $j \neq t$ and $\mu_{j}^{t}=1$ if $j=t$.
Then,
\begin{align*}
r(y_{t}, y) &=\max_{\lambda \in \Delta}\min_{\mu \in \Delta}\displaystyle\Sigma_{s=1}^{T}\displaystyle\Sigma_{t=1}^{T} \lambda_{s}\mu_{t} r_{st}(y_{t}, y) \\
&\geq \min_{\lambda \in \Delta}\displaystyle\Sigma_{i=1}^{T}\displaystyle\Sigma_{j=1}^{T} \lambda_{i}\mu^{j}_{t}r_{ij}(y_{t}, y) \\
&= \min_{\lambda \in \Delta}\displaystyle\Sigma_{i=1}^{T} \lambda_{i}r_{it}(y_{t}, y)
\end{align*}
Since for each data set $O^{2}_{ts}$, $p_{t}y_{t} \geq p_{t}y$ and $u_{st}$ strictly rationalizes $O^{2}_{ts}$, $u_{st}(y_{t}) > u_{st}(y)$. Hence, $r_{st}(y_{t}, y) > 0$. It follows that $r(y_{t}, y) > 0$ for each data set $O^{2}_{ts}$. So, $r$ rationalizes $O^{T}$.
I now prove that the preference function $r$ is skew-symmetric and strictly increases.
First, I show the skew-symmetry.
\begin{align*}
-r(x, y) &=-\min_{\lambda \in \Delta}\max_{\mu \in \Delta}\Sigma _{s=1}^{T}\Sigma _{t=1}^{T} \lambda_{s}\mu_{t}r_{st}(x, y) \\
&=\max_{\lambda \in \Delta}\min_{\mu \in \Delta}\Sigma _{s=1}^{T}\Sigma _{t=1}^{T} \lambda_{s}\mu_{t}\left(-r_{st}(x, y) \right) \\
&=\max_{\lambda \in \Delta}\min_{\mu \in \Delta}\Sigma _{s=1}^{T}\Sigma _{t=1}^{T} \lambda_{s}\mu_{t} r_{st}(y, x) ~\text{since $r_{st}$ is skew-symmetric}\\
&=\min_{\mu \in \Delta}\max_{\lambda \in \Delta}\Sigma _{s=1}^{T}\Sigma _{t=1}^{T} \lambda_{s}\mu_{t} r_{st}(y, x) \\
&=r(y, x)
\end{align*} 
Since $-r(x, y)=r(y, x)$, the preference function $r$ is skew-symmetric.
The preference function $r$ strictly increases. 
 
 Consider any $x, y, z \in X$ such that $x > y$. Then:
\begin{align*}
r_{st}(x, z) &=u_{st}(x)-u_{st}(z) \\
 &>u_{st}(y)-u_{st}(z) ~~\text{because}~u_{st} ~\text{strictly increasing}\\
 &=r_{st}(y, z)
\end{align*}
This equation means that: $\max_{\lambda \in \Delta}\Sigma_{s=1}^{T}\Sigma_{t=1}^{T} \lambda_{s}\mu_{t} r_{st}(x, z) >\max_{\lambda \in \Delta}\Sigma_{s=1}^{T}\Sigma_{t=1}^{T} \lambda_{s}\mu_{t} r_{st}(y, z)$ for all $\lambda \in \Delta$. It follows that $r(x, z) > r(y, z)$. Thus, $r$ strictly increases. \vspace{0.20cm}

Now, I prove that the preference function $r$ is continuous. 
By definition, the simplex $\Delta$ consists of a finite number of elements. $\Delta$ is compact because any finite set in a metric space is compact.
The above has shown that $r_{st}$ is continuous because $u_{st}$ is continuous. Hence, for every $\lambda, \mu \in \Delta$, $$f(x, y; \lambda, \mu)=\Sigma_{s \in \mathbb{T}}\Sigma_{t \in \mathbb{T}} \lambda_{s}\mu_{t} r_{st}(x, y)~\text{is continuous}$$
By a direct application of Berge’s maximum theorem, it follows that:

$r(x, y) = \min_{\mu \in \Delta} \left\lbrace \max_{\lambda \in \Delta}\Sigma_{s \in \mathbb{T}}\Sigma _{t \in \mathbb{T}} \lambda_{s}\mu_{t}r_{st}(x, y)\right\rbrace $ is a continuous function for all $x, y \in X$.\vspace{0.20cm}

Now I prove that the preference function $r$ is strictly piecewise concave. Because $u_{st}$ is strictly concave, I know that $r_{st}$ is strictly concave by construction. Consequently, for every $\lambda, \mu \in \Delta$
$$f(x, y; \lambda, \mu)=\Sigma_{s \in \mathbb{T}}\Sigma_{t \in \mathbb{T}} \lambda_{s}\mu_{t} r_{st}(x, y)~~\text{is strictly concave}$$
It follows that: $$f(x, y; \lambda, \mu)=\min_{\mu \in \Delta}
\Sigma_{s \in \mathbb{T}}\Sigma_{t \in \mathbb{T}} \lambda_{s}\mu_{t} r_{st}(x, y)~~\text{is strictly concave}$$
Thus, $$r(x, y)=\max_{\lambda \in \Delta} \left\lbrace \min_{\mu \in \Delta}\Sigma _{s=1}^{T}\Sigma _{t=1}^{T} \lambda_{s}\mu_{t}r_{st}(x, y)\right\rbrace~~\text{is strictly piecewise concave}$$
I have just proved that the preference function $r$ is continuous, skew-symmetric, strictly increasing, strictly piecewise-concave and also rationalizes the data $O^{T}$.\vspace{0.25cm}

Conversely, assume the data $O^{T}$ can be rationalized by a continuous, strictly increasing, and skew-symmetric preference function. I now show that $O^{T}$ obeys the WARP.
If the WARP is violated, there are two observations $s, t \in \mathbb{T}$ such that $p_{t}y_{t} \geq p_{t}y_{s}$ and $p_{s}y_{s}> p_{s}y_{t}$. By strict rationalization in definition \ref{def8}, I have $r(y_{t}, y_{s})>0$ and $r(y_{s}, y_{t})>0$, and these inequalities contradict the skew-symmetric property of $r$. 
\end{proof}
\vspace{0.5cm}

\textbf{Proof of Theorem \ref{theo1}}
\begin{proof}
Suppose the repeated cross-sectional data set $S$ is rationalized by the RPM. I show that there is a $\nu$ such that $\Gamma \nu =\rho$. 
If the repeated cross-sectional data set $S=\left\lbrace p_{t}, \rho_{t}\right\rbrace_{t \in \mathbb{T}}$ is rationalized by the RPM, then there is a probability space $(\Theta, \mathcal{F},\mu)$ and a random variable $y:\Theta \to \mathbb{R}^{LT}_{+}$ so that almost surely, $\left\lbrace p_{t}, y_{t}(\theta)\right\rbrace_{t \in \mathbb{T}}$ obeys the WARP. I can, therefore, construct the matrix of all potential types $\mathcal{A}$ and the matrix of WARP-consistent types of $\Gamma$.
Suppose $S$ can be rationalized by a distribution $\mu$. Let $\nu_{a}$ denote the mass of consumer types of $a$ in the population. For a given observation $t$, let $\mathcal{A}_{i \mid t} = \left\lbrace a \in \mathcal{A}: a_{i \mid t}=1\right\rbrace$, where $\mathcal{A}_{i \mid t}$ is the subset of WARP-consistent types that have their demand in the patch $x_{i \mid t}$ at observation $t$ and
 \begin{align*}
 a_{i \mid t}(r)=\mathbf{1}\left\lbrace y_{t} \in x_{l \mid t}:~r(y_{t}, z) \geq 0, ~p_{t}^{'}y_{t}\geq p_{t}z, \forall z \in B_{t} \right\rbrace.
\end{align*}
 Thus, the proportion of the population whose type belongs to $\mathcal{A}_{i \mid t}$ is
 $$ \mu\left(\left\lbrace \theta \in \Theta: y_{t}(\theta) \in x_{i \mid t}\right\rbrace \right) = \displaystyle\sum_{a \in \mathcal{A}_{i \mid t}} \nu_{a}=\displaystyle\sum_{a \in \mathcal{A}} \nu_{a} a_{i \mid t}=1, $$
 $$ \nu_{a}=\mu\left( \left\lbrace \Theta \in \theta: y_{t}(\theta) \in x_{i \mid t},~ if~ a_{i \mid t}=1,~\text{for all}~t \in \mathbb{T} \right\rbrace \right). $$
 Since $S$ is rationalized by $\mu$ and since by definition $\rho_{t}(\mathcal{Y})=\mu \left( \left\lbrace \theta \in \Theta: y_{t}(\theta) \in \mathcal{Y} \right\rbrace \right)$,
 
 setting $\mathcal{Y}=\left\lbrace y \in \mathbb{R}^{L}_{+}: y \in x_{i \mid t} \right\rbrace $, I obtain $\rho_{i \mid t}=\displaystyle\sum_{a \in \mathcal{A}} \nu_{a} a_{i \mid t}$. Therefore, the observed probability of choices falling on $x_{i \mid t}$ must equal the mass of WARP-consistent classes implied by $\mu$. This condition must hold for all patches $x_{i \mid t}$, and thus the condition $\rho_{i \mid t}=\displaystyle\sum_{ a \in \mathcal{A}} \nu_{a} a_{i \mid t}$ can be extended and written as $\Gamma \nu =\rho$, where $\nu$ is the column vector $\left(\nu_{a}\right)_{a \in \mathcal{A}}$ and $\rho$ is the vector of observed patch probabilities. \vspace{0.25cm}
 
$(ii) \Rightarrow (i)$
Suppose there is a probability vector $\nu \in \Delta^{H-1}$ such that the discretized choice probability $\rho$ satisfies the condition~ $\Gamma \nu=\rho$. Let us show that $S$ is rationalized by the RPM.

Given that $\rho_{t}=\left(\rho_{i \mid t} \right)_{i \in \mathcal{I}_{t}}$ is a probability measure for $B^{*}_{t}$, I define $\tilde{\rho}_{i \mid t}(y)=\rho_{t}(y \mid y \in K_{i \mid t})$ as the conditional distribution of demand at observation $t$ if it is restricted to the cone 

$K_{i \mid t}=\left\lbrace \gamma \cdot y: y \in x_{i \mid t}, \gamma >0 \right\rbrace $. If $\mathcal{Y}$ is a measurable subset of $\mathbb{R}^{L}_{+}$, then
 $\rho_{t}(\mathcal{Y} \mid \mathcal{Y} \in K_{i \mid t})=\frac{\rho \left(\mathcal{Y} \cap  K_{i \mid t}\right)}{\rho \left(K_{i \mid t} \right)}$. From this equation the following can be written,
 $\rho \left(\mathcal{Y} \cap K_{i \mid t}\right)=\rho_{i \mid t}\tilde{\rho}_{i \mid t}({\mathcal{Y}})$, where $\rho \left(K_{i \mid t}\right)=\rho_{i \mid t}$. 
 If $ \mathcal{Y} \cap K_{i \mid t}= \emptyset$, then $\rho \left( \mathcal{Y} \cap K_{i \mid t}\right)=0$. \vspace{0.25cm}
 
Given $a$ and $t$, there is a unique $i \in \mathcal{I}_{t}$ such that $a_{i \mid t}=1$.
Let $x_{t}^{a}=K_{i \mid t}$ and let $\tilde{\rho}_{t}^{a}$ be the probability measure on $B_{t}^{*}$ such that:
 $\tilde{\rho}_{t}^{a}=\tilde{\rho}_{i \mid t}$. It is clear that $\tilde{\rho}_{t}^{a}(x_{t}^{a})=1$.
 Let $\lambda_{a}$ be the product measure on $\left( B^{*}\right)^{T} \subseteq \left(\mathbb{R}^{L}_{+}\right)^{T}$ given by $\lambda_{a}=\bigtimes_{t \in \mathbb{T}} \tilde{\rho}_{t}^{a} $ and $\left( B^{*}\right)^{T}=\left(B_{1}^{*} \times B_{2}^{*} \times \ldots \times B_{T}^{*}\right)$. From the definition of $x_{t}^{a}$ it follows that:
 $$\bigtimes_{t \in \mathbb{T}} \tilde{\rho}_{t}^{a} \subseteq \left\lbrace y \in \left( \mathbb{R}^{L}_{+}\right)^{T}: \left\lbrace \left( p_{t}, y_{t}\right) \right\rbrace_{t \in \mathbb{T}} ~\text{satisfies WARP}\right\rbrace, $$
 and since $\tilde{\rho}_{t}^{a}(x_{t}^{a})=1$ for all $t$, it follows that:
 $$\lambda_{a}\left(\left\lbrace y \in \left( \mathbb{R}^{L}_{+}\right)^{T}: \left\lbrace \left( p_{t}, y_{t}\right) \right\rbrace_{t \in \mathbb{T}} ~\text{satisfies WARP} \right\rbrace \right) =1$$ Note that $y_{t}$ refers to the $t$-th entry of $y$.
Define $\Omega= \mathcal{A} \times \left( B^{*}\right)^{T}$ and the probability measure $\mu$ for $\Omega$ by $\mu\left( \left\lbrace a \right\rbrace \times \mathcal{Y}\right)=\nu_{a} \lambda_{a} (\mathcal{Y})$ for every measurable set $\mathcal{Y} \subseteq \left( B^{*}\right)^{T}$, where $\nu_{a}$ refers to the $a_{th}$ entry of $\nu$. Finally, I define:
 $\mathcal{X} : \Omega \to \left( \mathbb{R}^{L}_{+} \right)^{T}$
 as $\mathcal{X}\left(a, y \right)=y$. Setting $A= \mu \left( \left\lbrace (a, y) \in \Omega: \left\lbrace \left( p_{t}, \mathcal{X}_{t}(a, y) \right) \right\rbrace _{t \in \mathbb{T}}~~\text{satisfies WARP} \right\rbrace \right) $,
 \begin{align*}
A=\sum_{a \in \mathcal{A}} \nu_{a} \lambda_{a}\left( \left\lbrace y \in \left( \mathbb{R}^{L}_{+}\right)^{T}: \left\lbrace \left( p_{t}, \mathcal{X}_{t}(a, y) \right) \right\rbrace_{t \in \mathbb{T}}~~\text{satisfies WARP} \right\rbrace \right)=\sum_{a \in \mathcal{A}} \nu_{a} =1.
 \end{align*}
Next, I show that $ \mu \left( \left\lbrace (a, y) \in \Omega: \mathcal{X}_{t}(a, y) \in \mathcal{Y} \cap K_{i \mid t} \right\rbrace \right)= \rho_{t} \left( \mathcal{Y} \cap K_{i \mid t}\right)$ for all patch $K_{i \mid t}$.
For each measurable set $\mathcal{Y}$ in $\mathbb{R}^{L}_{+}$ and for each $K_{i \mid t},$
 \begin{align*}
 B &= \mu \left( \left\lbrace (a, y) \in \Omega: \mathcal{X}_{t}(a, y) \in \mathcal{Y} \cap K_{i \mid t} \right\rbrace \right) \\
 &=\sum_{a \in \mathcal{A}} \nu_{a} \lambda_{a}\left( \left\lbrace y \in \left( \mathbb{R}^{L}_{+} \right)^{T}: \mathcal{X}_{t}(a, y) \in \mathcal{Y} \cap K_{i \mid t} \right\rbrace \right) \\
 &=\sum_{a \in \mathcal{A}} \nu_{a} \lambda_{a}\left( \left\lbrace y \in \left( \mathbb{R}^{L}_{+} \right)^{T}: y_{t} \in \mathcal{Y} \cap K_{i \mid t} \right\rbrace \right) \\
 &=\sum_{a \in \mathcal{A}} \nu_{a} \tilde{\rho}_{t}^{a} \left( \left\lbrace y_{t} \in \mathbb{R}^{L}_{+} : y_{t} \in \mathcal{Y} \cap K_{i \mid t} \right\rbrace \right).
 \end{align*}
 Recall that $\tilde{\rho}_{t}^{a} \left( \left\lbrace y_{t} \in \mathbb{R}^{L}_{+} : y_{t} \in \mathcal{Y} \cap K_{i \mid t} \right\rbrace \right) =0$ if $a \notin \mathcal{A}_{i \mid t}$. From this equation the following occurs:
 \begin{align*}
 \sum_{a \in \mathcal{A}} \nu_{a} \tilde{\rho}_{t}^{a} \left( \left\lbrace y_{t} \in \mathbb{R}^{L}_{+} : y_{t} \in \mathcal{Y} \cap K_{i \mid t} \right\rbrace \right) & = \sum_{a \in \mathcal{A}_{i \mid t}} \nu_{a} \tilde{\rho}_{t}^{a} \left( \left\lbrace y_{t} \in \mathbb{R}^{L}_{+}: y_{t} \in \mathcal{Y} \cap K_{i \mid t} \right\rbrace \right) \\
 &=\frac{\rho_{t}\left( \mathcal{Y} \cap K_{i \mid t}\right) }{\rho_{i \mid t}} \sum_{a \in \mathcal{A}_{i \mid t}} \nu_{a}\\
 &=\rho_{t} \left(\mathcal{Y} \cap K_{i \mid t}\right).
 \end{align*}
 The last equality follows from the fact that $\rho_{i \mid t}=\sum_{a \in \mathcal{A}_{i \mid t}} \nu_{a}=\sum_{a \in \mathcal{A}} \nu_{a} a_{i \mid t}$.\vspace{0.5cm}

$(ii) \Leftrightarrow (iii)$
To prove the equivalence between $(ii)$ and $(iii)$, I can follow that of \cite{mcfadden2006revealed} for RUMs. The following are equivalents.
\begin{itemize}
    \item[(a)] There is a probability vector $\nu \in \Delta^{H-1}$, so that the discretized choice probability $\rho$ satisfies the condition~ $\Gamma \nu=\rho$.
    \item [(b)] The system of linear inequalities $\rho \leq \Gamma \nu$, $\nu >0$, $\mathbf{1}^{'} \nu \leq 1$ has a solution.
    \item [(c)] The linear program $\min_{\nu, s} \mathbf{1}^{'}s$ subject to, $\nu \geq 0$, $s\geq 0$, $\Gamma \nu +s \geq \rho$, $\mathbf{1}^{'} \nu \leq 1$ has an optimal solution with $s = 0$
     \item [(d)] The linear program $\max_{r, t}~(r^{'} \rho - t)$ subject to $0 \leq r\leq 1$, $t \geq 0$, and $r^{'}\Gamma \leq t\mathbf{1}^{'} $ has no positive solution.
\item[(v)] $\rho$ satisfies WARSP.
\end{itemize}
\vspace{0.5cm}

$(ii) \Leftrightarrow (iv)$
The proof uses nonstochastic demand systems that can be identified with vectors:
$(d_{1}, \ldots, d_{T}) \in \bigtimes_{t \in \mathbb{T}} B_{t}$. This demand system can be rationalized if there is a preference function $r$ such that $r(d_{t}, y)>0$, $ \forall y \in B_{t}$ and $\forall t \in \mathbb{T}$. \vspace{0.25cm}

Set $P =\left\lbrace P_{t}\right\rbrace_{t \in \mathbb{T}}$. The set $\tilde{\mathcal{Y}}$ contains one element of each patch. Let $\tilde{P}=\left\lbrace \tilde{P}_{t}\right\rbrace_{t \in \mathbb{T}}$ be the only stochastic demand system that is centered on $\tilde{\mathcal{Y}}$ and has the same vector representation as $P$. Since all nonstochastic demand systems of a given budget that are on the same patch lead to the same directly revealed preferences, nonstochastic demand systems are either all rationalizable or not. As a result, the demand systems within patches can be arbitrarily perturbed. It follows that $P$ is rationalizable if $\tilde{P}$ is. Therefore, the rationalizability of $P$ can be determined using its vector representation $\rho=\left\lbrace \rho_{t}\right\rbrace_{t \in \mathbb{T}}$. Thus, any stochastic demand system is rationalizable if and only if it is a mixture of rationalizable nonstochastic systems. Since I have a finite number of budgets; $\tilde{\mathcal{Y}}$ is finite, and there are many finite systems of nonstochastic demand that are based on it. Among these nonstochastic demand systems, a subset is rationalizable and characterized by binary vector representations corresponding to the columns of $\Gamma$.
\end{proof} 
\vspace{0.5cm}

\textbf{Proof of Proposition \ref{prop1}}

Figure \ref{fig0} shows the relationship between the RUM and the RPM.\vspace{0.25cm}
 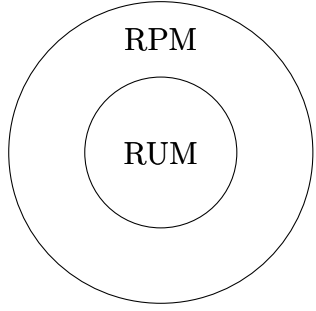
\begin{figure}
 \caption{Relationship between the RPM and the RUM when the dimension of the commodity space is larger than two}{\label{fig0}}
\begin{tikzpicture}
    \pgfmathsetmacro{\numCircles}{2}
    
    \foreach \i in {1,...,\numCircles} {
        \draw (0,0) circle (\i);
        \node at (0,0) {RUM}; 
         \node at (0,1.5) {RPM};
    }
\end{tikzpicture}
 \end{figure}
 
\begin{proof}
    The proof of this proposition is straightforward and follows from the fact that SARP comes from the WARP. If the data $S= \left\lbrace \rho_{t}, p_{t}\right\rbrace_{t \in \mathbb{T}}$ are rationalized by the RUM, then there is a probability space $(\Theta, \mathcal{F},\mu)$ and a random variable $y:\Theta \to \mathbb{R}^{LT}_{+}$ so that almost surely, $\left\lbrace p_{t}, y_{t}(\theta)\right\rbrace_{t \in \mathbb{T}}$ obeys SARP. But since SARP comes from the WARP, it follows that there is a probability space $(\Theta, \mathcal{F},\mu)$ and a random variable $y:\Theta \to \mathbb{R}^{LT}_{+}$ so that almost surely, $\left\lbrace p_{t}, y_{t}(\theta)\right\rbrace_{t \in \mathbb{T}}$ obeys WARP. Thus, the data set $S$ is rationalized by the RPM.
\end{proof}
\vspace{0.5cm}

\textbf{Proof of Proposition \ref{prop2}}

\begin{proposition}\cite[][]{cherchye2018transitivity}

   A set of normalized price vectors $P=\left\lbrace \tilde{p}_{t} \right\rbrace_{t \in \mathbb{T}}$
is WARP-reducible if and only if it is a \textit{triangular configuration}.
\end{proposition}

\begin{proof}
    From Proposition \ref{prop1}, I know that if $S$ is rationalized by the RUM, it is also rationalized by the RPM. To establish the equivalence, I need to show that if $S$ is rationalized by the RPM, then it is rationalized by the RUM. If $S$ is rationalized by the RPM, then there is a probability space $(\Theta, \mathcal{F},\mu)$ and a random variable $y:\Theta \to \mathbb{R}^{LT}_{+}$ so that almost surely, $\left\lbrace p_{t}, y_{t}(\theta)\right\rbrace_{t \in \mathbb{T}}$ obeys the WARP. But since WARP is equivalent to the SARP under \textit{triangular configuration}, it follows that there is a probability space $(\Theta, \mathcal{F},\mu)$ and a random variable $y:\Theta \to \mathbb{R}^{LT}_{+}$ so that almost surely, $\left\lbrace p_{t}, y_{t}(\theta)\right\rbrace_{t \in \mathbb{T}}$ obeys the SARP.
\end{proof}
\vspace{0.5cm}

\textbf{Proof of Corollary \ref{cor1}}
\begin{proof}
Assume that the data are rationalized by the RPM. I show that there is a unique probability vector $u_{ts} \in \Delta^{3}$ such that $\Gamma^{*} u_{ts}=\rho_{ts}$ for all $t, s \in \mathbb{T} $, $\rho_{ts}=\left(\rho_{0 \mid t}, \rho_{1 \mid t}, \rho_{0 \mid s}, \rho_{1 \mid s}\right)^{'}$.\vspace{0.25cm}

\textbf{Case 1}: The budgets $B_{s}$ and $B_{t}$ overlap each other.
If $\Gamma \eta=\rho$, then for all $t \in \mathbb{T}$ and for all $i \in \mathcal{I}_{t}$,
 $\rho_{i \mid t}=\sum_{a \in \mathcal{A}}\eta_{a} a_{i \mid t},$ where $\eta$ is a column vector $\left(\eta_{a}\right)_{a \in \mathcal{A}}$ with $\eta_{a} \geq 0$ for all $a \in \mathcal{A}$. Let $a_{t}=\left\lbrace a_{i \mid t}\right\rbrace_{i \in \mathcal{I}_{t}}$ and $\mathcal{A}= \left\lbrace \mathcal{A}_{i \mid t}\right\rbrace_{l \in \mathcal{I}_{t}, t\in \mathbb{T}}$ with
$\mathcal{A}_{i \mid t}=\left\lbrace a: a_{i \mid t}=1 \right\rbrace $ for all $t \in \mathbb{T}$ that are the subset of WARP-consistent types that have their demands in $x_{i \mid t}$. \vspace{0.25cm}

I define a function that maps $\rho$ to the vector of choice probabilities $\rho_{ts}$ as follows: $A_{t \mid s}\rho=\rho_{ts}$, $A_{t \mid s}$ is a matrix with four rows, and its number of columns corresponds to the total number of patches over the $T$ periods. $A_{t \mid s}$ is composed of $0$ and $1$ and is not necessarily full rank. \vspace{0.25cm}

For all $t,s \in \mathbb{T}$, the condition $\Gamma^{*} u_{ts}=\rho_{ts}$ is therefore equivalent to $\Gamma^{*} u_{t \mid s}= A_{t \mid s} \rho= A_{t \mid s} \left(\Gamma \eta\right)$
 The matrix $\Gamma^{*}$ is full rank, so it is invertible. It follows that for all $t, s \in \mathbb{T}$, $$u_{ts}=\Gamma^{*^{-1}}A_{t \mid s} \left(\Gamma \eta\right)$$
 The vector $u$ is unique because the inverse of the matrix $\Gamma^{*}$ exists and is positive by construction.
 
\textbf{Case 2}: The hyperplane budgets $B_{t}$ and $B_{s}$ do not intersect. I assume w.l.o.g that $B_{s}$ lies above $B_{t}$. The rational matrix is therefore $\Gamma^{*}=(1 ~1)^{'}$. Thus, $\Gamma^{*} u=\rho_{ts}$, where $u=(u_{t}, u_{s})^{'}$ and $\rho_{ts}=(\rho_{t}, \rho_{s})^{'}$. Therefore, $$u_{t}=\rho_{t}=A_{t}\rho=A_{t}\left(\Gamma \eta\right)~\text{and}~u_{s}=\rho_{s}=A_{s}\rho=A_{s}\left(\Gamma \eta\right) $$
where $A_{t}$ and $A_{s}$ are the functions that map the vector of choice probabilities $\rho$ to the vector of choice probabilities $\rho_{t}$ and $\rho_{s}$, respectively.
\end{proof}

\vspace{0.5cm}

\textbf{Econometric Testing} 

\underline{Simulating a Critical Value for $\mathcal{J}_{N}$}:
This approach deals with sampling variability. Sampling variability results from the fact that $\rho$ can only be consistently estimated by the realized choice frequencies $\hat{\rho}$.
Let $\hat{\rho}^{*(r)}$, with $r=1, \ldots R$ be the bootstrap replications of $\hat{\rho}$. Let $\iota$ be a vector of ones of dimension $H$ and $\tau_{N}=\sqrt{\frac{\log N}{N}}$ be a tuning parameter.\vspace{0.25cm}

\begin{itemize}
\item[(i)] Obtain the $\tau_{N}$-tightened restricted estimator $\hat{\eta}_{\tau_{N}}$ that solves 
$$min_{\eta \in \mathcal{C}_{\tau_{N}}} \parallel \hat{\rho}-\eta \parallel^{2}=min_{\left[\nu-\tau_{N}\iota/ H \right] \in \mathbb{R}_{+}^{H}} \parallel \hat{\rho}-\Gamma\nu \parallel^{2};$$ 
\item[(iii)] The bootstrap test statistic is: 
$$\mathcal{J}_{N}^{*(r)}=\min_{\left[\nu-\tau_{N}\iota/H\right] \in \mathbb{R}_{+}^{H}} \parallel \hat{\rho}^{*(r)}-\hat{\rho}+\hat{\eta}_{\tau_{N}}-\eta \parallel^{2},~~r=1, \ldots, R;$$ 
\item[(vi)]. Use the empirical distribution of $\mathcal{J}_{N}^{*(r)}$, $r=1, \ldots, R$, to obtain the critical value for $\mathcal{J}_{N}$.\footnote{Please, see KS to learn more about this test statistic.} \label{stat}
\end{itemize}

\vspace{0.5cm}

\textbf{Proof of Theorem \ref{theor3}}
\begin{proof}
The proof of this Theorem is omitted as it is similar to that of Theorem $3$ in \cite{kitamura2019nonparametric}.
By the law of iterated expectation,
\begin{align*}
\mathbb{E}\left(h(y(\theta)\right) &=\sum_{i=0}^{I_{0}} \rho_{i\mid 0} h_{i \mid 0}   \\
  &=\left( h_{0 \mid 0},\ldots, h_{I_{0} \mid 0}\right) \rho_{0}, 
\end{align*}
where $\mathbb{E}\left(h(y((\theta) \mid y \in x_{i \mid 0} \right)=h_{i \mid 0}$ if $\rho_{i \mid 0} \neq 0$, and I assign it any other value otherwise.
If $i=0$, then $\underline{h}_{0 \mid 0}=\inf\left\lbrace h(y(\theta)), y(\theta)\in x_{0 \mid 0}\right\rbrace$, where $y(\theta)$ is an element of $x_{0 \mid 0}$ for which the function $h$ has the smallest value. Similarly, for $l=0$: $\bar{h}_{0 \mid 0}=\sup\left\lbrace h(y(\theta)), y(\theta) \in x_{0 \mid 0}\right\rbrace$, where $\tilde{y}$ is an element of $x_{0 \mid 0}$ for which the function $h$ has the highest value. For all $i \in \mathcal{I}_{0}$, $\underline{h}_{i \mid 0}$ and $\bar{h}_{i \mid 0}$ are tractable if the function $h$ is continuous. If the function $h$ is, in addition, linear, then computing the bounds requires only linear programming \cite[][]{kitamura2019nonparametric}.
I can also bound probabilities of arbitrary events as follows:
\begin{corollary}{\label{cor3}}
Given a fixed event $x_{i \mid 0} \subseteq B_{0}^{*}=\left\lbrace x_{i \mid 0}\right\rbrace_{i \in \mathcal{I}_{0}}$, the bounds are:
 \begin{align*}
 \min \left\lbrace \Sigma_{i=0}^{I_{0}}  ~e^{'}_{i}\Gamma_{0} \nu:
 \tilde{\Gamma}\nu=\tilde{\rho},~ \nu \geq 0 \right\rbrace  \\
  \leq \mu\left(y(\theta) \in x_{i \mid 0})\right) \leq \\
   \max \left\lbrace \Sigma_{i=0}^{I_{0}}  ~e^{'}_{i}\Gamma_{0} \nu:
\tilde{\Gamma}\nu=\tilde{\rho},~ \nu \geq 0  \right\rbrace 
\end{align*}
 
 where $\tilde{\rho}=\left[\rho_{1} \ldots \rho_{T}\right]^{'}$,  
  $\tilde{\Gamma}=\left[\Gamma_{1} \ldots \Gamma_{T}\right]^{'}$, and $e^{'}_{i}$ are the $i^{th}$ canonical basis vector in $\mathbb{R}^{I_{0}}$, and $\mu$ is the true distribution rationalizing the data. 
\end{corollary}
This corollary provides the bounds for the probability that a given nonstochastic demand $y(\theta)$ falls within the patch $x_{i \mid 0}$. In the current literature, the lower and upper bounds cannot be determined if the data $S$ is not RUM-rationalized because the interval in which $\mu\left(y(\theta) \in x_{i \mid 0})\right)$ belongs is empty. Theorem \ref{theor3} shows that if $S$ is RPM-rationalized, one can recover these bounds and obtain counterfactual bounds that are robust against transitivity violations.
\end{proof} 

\vspace{0.5cm}

\textbf{Counter-Example, Remark 2}\label{matlab}

I ran the following minimization problem in Matlab:
$$\min \nu, \text{subject to}~~\Gamma \nu=\rho~ and~ \nu \geq 0$$
and I received the following message: Certificate\_infeasibility,  ``No feasible solution found" and ``Linprog stopped because no point satisfies the constraints." \vspace{0.25cm}
 A certificate of infeasibility belongs to the following set:
 $$S=\left\lbrace y \in \mathbb{R}^{12}: y^{'}\Gamma=0 ~\text{and}~y^{'} \rho >0 \right\rbrace $$
 
 \vspace{0.5cm}
 
\textbf{Example 4:} 
 \textbf{Generalized Shafer (1974)'s Preference Function}
 
 Given that $\alpha \in (0, 1)$,  $r_{\alpha}(x, y)=(x^{2})^{\alpha} (y^{1})^{\alpha-1}+ \ln (x^{3})-(y^{2})^{\alpha}(x^{1})^{\alpha-1}-\ln (y_{3})$ 
 
 is the general version of the \cite{shafer1974nontransitive}'s preference function.
\begin{itemize}
    \item [(i)] $r_{\alpha}$ is continuous on its domain $X \times X$.
    \item [(ii)] $r_{\alpha}$ is skew-symmetric. 
    \begin{align*}
r_{\alpha}(x, y) &=(x^{2})^{\alpha} (y^{1})^{\alpha-1}+ \ln (x^{3})-(y^{2})^{\alpha}(x^{1})^{\alpha-1}-\ln (y^{3}) \\
 &= -\left((y^{2})^{\alpha}(x^{1})^{\alpha-1}+\ln (y^{3})-(x^{2})^{\alpha}(y^{1})^{\alpha-1}-\ln (x^{3})\right) \\
 &=-r_{\alpha}(y, x)
    \end{align*}

    \item [(iii)] $r_{\alpha}(x, y)$ is strictly increasing. Given $\alpha \in (0, 1)$ and $x, y, z \in X$ such that $x > z$
   
    $r_{\alpha}(x, y)-r_{\alpha}(z, y)=(y^{1})^{\alpha-1}\left((x^{2})^{\alpha}-(z^{2})^{\alpha}\right)+(y^{2})^{\alpha}\left((z^{1})^{\alpha-1}-(x^{1})^{\alpha-1}\right)+\ln(x^{3})-\ln(z^{3})$
Since $x > z$,  $\left(x^{i}-z^{i}\right) > 0$ for all $i=1, 2, 3$ and $x \neq z$. Hence, $x^{3} > z^{3}$ and $\ln(x^{3}) > \ln(z^{3})$ because the natural logarithm function strictly increases with $\mathbb{R}_{+}\setminus\left\lbrace 0\right\rbrace$.  $\left(x^{i}-z^{i}\right) > 0$ implies that $\left((x^{i})^{\alpha}-(z^{i})^{\alpha} \right) > 0$ and $\left((z^{1})^{\alpha-1}-(x^{1})^{\alpha-1}\right) > 0$. It follows that $r_{\alpha}(x, y) > r_{\alpha}(z, y)$.
\item [(iv)] The preference function $r_{\alpha}$ is concave for $\alpha \in \left(0, 1\right)$. Given $y \in X$

$\frac{\partial r_{\alpha}}{\partial x^{1}}=-(\alpha-1)(y^{2})^{\alpha}(x^{1})^{\alpha-2}>0$ ~and~ $\frac{\partial^{2} r_{\alpha}}{\partial (x^{1})^{2}}=-(\alpha-1)(\alpha-2) (y^{2})^{\alpha} (x^{1})^{\alpha-3} <0 $\\ 
\vspace{0.25cm}
$\frac{\partial r_{\alpha}}{\partial x^{2}}=\alpha (x_{2})^{\alpha-1}(y_{1})^{\alpha-1}>0$ ~and~ $\frac{\partial^{2} r_{\alpha}}{\partial (x^{2})^{2}}=\alpha (\alpha-1)(x^{2})^{\alpha-2}(y^{1})^{\alpha-1} <0 $\\ 
\vspace{0.25cm}
$\frac{\partial r_{\alpha}}{\partial x^{3}}=1/x^{3}>0$ ~and~ $\frac{\partial^{2} r_{\alpha}}{\partial (x^{3})^{2}}=-1/(x^{3})^{2}<0$ ~and~ $\frac{\partial^{2} r_{\alpha}}{\partial x^{1}x^{2}}=\frac{\partial^{2} r_{\alpha}}{\partial x^{1}x^{3}}=\frac{\partial^{2} r_{\alpha}}{\partial x^{2}x^{3}}=0$ 
\end{itemize}
\vspace{0.5cm}

\textbf{Monte Carlo Simulations}
 
$$X=
\left( {\begin{array}{ccc}
0 & -1& -1\\
0 & 1& -1\\
0 & -1& 1\\
0 & 1& 1\\
-1 & 0& -1\\
1& 0& -1\\
-1 & 0& 1\\
1 & 0& 1\\
-1 & -1& 0\\
1 & -1& 0\\
-1 & 1& 0\\
1 & 1& 0\\
\end{array} } \right)$$~~~~\begin{table}\label{table3}
  \centering
  \begin{tabular}{cccc}
    \hline \\
  
   &$\rho_{outside}$ & $\rho_{inside}$ & $\rho_{true}$  \\
     \hline
       \hline \\
    &$0.3$ & $0.181$ & $0.1869$  \\
    &$0.2$ & $0.2262$ & $0.2249$  \\
     &$0.2$ & $0.2262$ & $0.2249$  \\
      &$0.3$ & $0.3667$ & $0.3633$  \\
        &$0.3$ & $0.181$ & $0.1869$  \\
       &$0.2$ & $0.2262$ & $0.2249$  \\
        &$0.2$ & $0.2262$ & $0.2249$  \\
        &$0.3$ & $0.3667$ & $0.3633$  \\
         &$0.3$ & $0.181$ & $0.1869$  \\
          &$0.2$ & $0.2262$ & $0.2249$  \\
          &$0.2$ & $0.2262$ & $0.2249$  \\
           &$0.3$ & $0.3667$ & $0.3633$  \\
    \hline
     \hline
  \end{tabular}
  \caption{Vectors used for the Monte Carlo Simulations}
  \label{tab:mytable4}
\end{table}

\vspace{0.75cm}

$$\Gamma=
\left( {\begin{array}{cccccccccccccccccccccccccccc}{\label{matrix}}
1 & 1& 1&1&1&1&1&1&1&1&1&1&0&0&0&0&0&0&0&0&0&0&0&0&0&0&0\\
0 & 0& 0&0&0&0&0&0&0&0&0&0&1&1&1&1&1&1&0&0&0&0&0&0&0&0&0\\
0 & 0& 0&0&0&0&0&0&0&0&0&0&0&0&0&0&0&0&1&1&1&1&1&1&0&0&0\\
0 & 0& 0&0&0&0&0&0&0&0&0&0&0&0&0&0&0&0&0&0&0&0&0&0&1&1&1\\
1 & 1& 1&1&0&0&0&0&0&0&0&0&1&1&1&1&0&0&1&1&0&0&0&0&1&1&0\\
0 & 0& 0&0&1&1&1&1&0&0&0&0&0&0&0&0&0&0&0&0&1&1&0&0&0&0&0\\
0 & 0& 0&0&0&0&0&0&1&1&0&0&0&0&0&0&1&1&0&0&0&0&1&0&0&0&1\\
0 & 0& 0&0&0&0&0&0&0&0&1&1&0&0&0&0&0&0&0&0&0&0&0&1&0&0&0\\
1 & 0& 0&0&1&0&0&0&1&0&1&0&1&0&0&0&1&0&1&0&1&0&1&1&1&0&1\\
0 & 1& 0&0&0&1&0&0&0&0&0&0&0&1&0&0&0&0&0&1&0&0&1&0&0&1&0\\
0 & 0& 1&0&0&0&1&0&0&1&0&1&0&0&1&0&0&1&0&0&0&0&0&0&0&0&0\\
0 & 0& 0&1&0&0&0&1&0&0&0&0&0&0&0&1&0&0&0&0&0&0&0&0&0&0&0
\end{array} } \right)$$
\[ \Gamma \nu=\rho \Leftrightarrow \left\{
\begin{array}{ll}
\Sigma_{l=1}^{12} \nu =\rho_{0 \mid 1}\\
\nu_{13}+\nu_{14}+\nu_{15}+\nu_{16}+\nu_{17}+\nu_{18} =\rho_{1 \mid 1}\\
\nu_{19}+\nu_{20}+\nu_{21}+\nu_{22}+\nu_{23}+\nu_{24}=\rho_{2 \mid 1} \\
\nu_{25}+\nu_{26}+\nu_{27}=\rho_{3 \mid 1} \\
\nu_{1}+\nu_{2}+\nu_{3}+\nu_{4}+\nu_{13}+\nu_{14}+\nu_{15}+\nu_{16}+\nu_{19}+\nu_{20}+\nu_{25}+\nu_{26}=\rho_{0 \mid 2} \\
\nu_{5}+\nu_{6}+\nu_{7}+\nu_{8}+\nu_{21}+\nu_{22}=\rho_{1 \mid 2}\\
\nu_{9}+\nu_{10}+\nu_{17}+\nu_{18}+\nu_{23}+\nu_{27}=\rho_{2\mid 2}\\
\nu_{11}+\nu_{12}+\nu_{24}=\rho_{3 \mid 2} \\
\nu_{1}+\nu_{5}+\nu_{9}+\nu_{11}+\nu_{13}+\nu_{17}+\nu_{19}+\nu_{21}+\nu_{23}+\nu_{24}+\nu_{25}+\nu_{27}=\rho_{0 \mid 3} \\
\nu_{2}+\nu_{6}+\nu_{14}+\nu_{20}+\nu_{23}+\nu_{26}=\rho_{1 \mid 3} \\
\nu_{3}+\nu_{7}+\nu_{10}+\nu_{12}+\nu_{15}+\nu_{18}=\rho_{2 \mid 3} \\
\nu_{4}+\nu_{8}+\nu_{16}=\rho_{3 \mid 3}
\end{array}
\right. \]
\vspace{1cm}
\begin{figure}
 \caption{Consumers' types}{\label{figur3}}
 \begin{subfigure}{5cm}
 \caption{Rational type $1$}
\begin{tikzpicture}[scale=0.6]Axis
\draw [->] (0,0) node [below] {0} -- (0,0) -- (6,0) node [right] {$y_{1}$};
\draw [->] (0,0) node [below] {0} -- (0,0) -- (0,6) node [above] {$y_{2}$};
\node [below][scale=0.7] at (5,0) {$B_{1}$};
\node [left][scale=0.7] at (0, 2.5) {};
\node [below][scale=0.7] at (2.5,0) {$B_{2}$};
\node [left][scale=0.7] at (0, 5) {};
\draw [ultra thick] (0, 2.5)--(5,0) node[above][scale=0.7]{};
\draw [ultra thick] (0,5)--(2.5, 0) node[above left][scale=0.7]{};
\draw[red, thick] (0, 5) -- (1.6, 1.8);
\draw[red, thick] (0, 2.5) -- (1.6, 1.8);
\node [below][scale=0.7] at (1.2, 4) {$x_{0 \mid 2}$};
\node [below][scale=0.7] at (1.7, 0.9) {$x_{1 \mid 2}$};
\node [below][scale=0.7] at (0.8, 2) {$x_{0 \mid 1}$};
\node [below][scale=0.7] at (4.2, 1.2) {$x_{1 \mid 1}$};
\end{tikzpicture}
\end{subfigure}
  \begin{subfigure}{5cm}
  \caption{Rational type $2$}
\begin{tikzpicture}[scale=0.6]Axis
\draw [->] (0,0) node [below] {0} -- (0,0) -- (6,0) node [right] {$y_{1}$};
\draw [->] (0,0) node [below] {0} -- (0,0) -- (0,6) node [above] {$y_{2}$};
\node [below][scale=0.7] at (5,0) {$B_{1}$};
\node [left][scale=0.7] at (0, 2.5) {};
\node [below][scale=0.7] at (2.5,0) {$B_{2}$};
\node [left][scale=0.7] at (0, 5) {};
\draw [ultra thick] (0, 2.5)--(5,0) node[above][scale=0.7]{};
\draw [ultra thick] (0,5)--(2.5, 0) node[above left][scale=0.7]{};
\draw[red, thick] (1.6, 1.8) -- (2.5, 0);
\draw[red, thick] (1.6, 1.8) -- (5, 0);
\node [below][scale=0.7] at (1.2, 4) {$x_{0 \mid 2}$};
\node [below][scale=0.7] at (1.6, 1) {$x_{1 \mid 2}$};
\node [below][scale=0.7] at (0.8, 2) {$x_{0 \mid 1}$};
\node [below][scale=0.7] at (4, 1.4) {$x_{1 \mid 1}$};
\end{tikzpicture}
  \end{subfigure}
  \begin{subfigure}{5cm}
  \caption{Rational type $3$}
\begin{tikzpicture}[scale=0.6]Axis
\draw [->] (0,0) node [below] {0} -- (0,0) -- (6,0) node [right] {$y_{1}$};
\draw [->] (0,0) node [below] {0} -- (0,0) -- (0,6) node [above] {$y_{2}$};
\node [below][scale=0.7] at (5,0) {$B_{1}$};
\node [left][scale=0.7] at (0, 2.5) {};
\node [below][scale=0.7] at (2.5,0) {$B_{2}$};
\node [left][scale=0.7] at (0, 5) {};
\draw [ultra thick] (0, 2.5)--(5,0) node[above][scale=0.7]{};
\draw [ultra thick] (0,5)--(2.5, 0) node[above left][scale=0.7]{};
\draw[red, thick] (0, 5) -- (1.6, 1.8);
\draw[red, thick] (1.6, 1.8) -- (5, 0);
\node [below][scale=0.7] at (1.2, 4) {$x_{0 \mid 2}$};
\node [below][scale=0.7] at (1.6, 1) {$x_{1 \mid 2}$};
\node [below][scale=0.7] at (0.8, 2) {$x_{0 \mid 1}$};
\node [below][scale=0.7] at (4, 1.4) {$x_{1 \mid 1}$};
\end{tikzpicture}
  \end{subfigure}
  
  \begin{subfigure}{5cm}
  \caption{Irrational type}
\begin{tikzpicture}[scale=0.6]Axis
\draw [->] (0,0) node [below] {0} -- (0,0) -- (6,0) node [right] {$y_{1}$};
\draw [->] (0,0) node [below] {0} -- (0,0) -- (0,6) node [above] {$y_{2}$};
\node [below][scale=0.7] at (5,0) {$B_{1}$};
\node [left][scale=0.7] at (0, 2.5) {};
\node [below][scale=0.7] at (2.5,0) {$B_{2}$};
\node [left][scale=0.7] at (0, 5) {};
\draw [ultra thick] (0, 2.5)--(5,0) node[above][scale=0.7]{};
\draw [ultra thick] (0,5)--(2.5, 0) node[above left][scale=0.7]{};
\draw[red, thick] (1.6, 1.8) -- (2.5, 0);
\draw[red, thick] (0, 2.5) -- (1.6, 1.8);
\node [below][scale=0.7] at (1.2, 4) {$x_{0 \mid 2}$};
\node [below][scale=0.7] at (1.6, 1) {$x_{1 \mid 2}$};
\node [below][scale=0.7] at (0.8, 2) {$x_{0 \mid 1}$};
\node [below][scale=0.7] at (4, 1.4) {$x_{1 \mid 1}$};
\end{tikzpicture}
  \end{subfigure}
\end{figure}
  
\newpage
\bibliographystyle{apacite}
\bibliography{summerpaper}

\begin{thebibliography}{}

\bibitem [\protect \citeauthoryear {%
Aguiar%
, Boccardi%
, Kashaev%
\BCBL {}\ \BBA {} Kim%
}{%
Aguiar%
\ \protect \BOthers {.}}{%
{\protect \APACyear {2018}}%
}]{%
aguiar2018does}
\APACinsertmetastar {%
aguiar2018does}%
\begin{APACrefauthors}%
Aguiar%
, Boccardi, M\BPBI J.%
, Kashaev, N.%
\BCBL {}\ \BBA {} Kim, J.%
\end{APACrefauthors}%
\unskip\
\newblock
\APACrefYearMonthDay{2018}{}{}.
\newblock
{\BBOQ}\APACrefatitle {Does random consideration explain behavior when choice is hard? {E}vidence from a large-scale experiment} {Does random consideration explain behavior when choice is hard? {E}vidence from a large-scale experiment}.{\BBCQ}
\newblock
\APACjournalVolNumPages{arXiv preprint arXiv:1812.09619}{}{}{}.
\PrintBackRefs{\CurrentBib}

\bibitem [\protect \citeauthoryear {%
Aguiar%
, Boccardi%
, Kashaev%
\BCBL {}\ \BBA {} Kim%
}{%
Aguiar%
\ \protect \BOthers {.}}{%
{\protect \APACyear {2023}}%
}]{%
aguiar2023random}
\APACinsertmetastar {%
aguiar2023random}%
\begin{APACrefauthors}%
Aguiar%
, Boccardi, M\BPBI J.%
, Kashaev, N.%
\BCBL {}\ \BBA {} Kim, J.%
\end{APACrefauthors}%
\unskip\
\newblock
\APACrefYearMonthDay{2023}{}{}.
\newblock
{\BBOQ}\APACrefatitle {Random utility and limited consideration} {Random utility and limited consideration}.{\BBCQ}
\newblock
\APACjournalVolNumPages{Quantitative Economics}{14}{1}{71--116}.
\PrintBackRefs{\CurrentBib}

\bibitem [\protect \citeauthoryear {%
Aguiar%
, Hjertstrand%
\BCBL {}\ \BBA {} Serrano%
}{%
Aguiar%
\ \protect \BOthers {.}}{%
{\protect \APACyear {2020}}%
}]{%
aguiar2020rationalization}
\APACinsertmetastar {%
aguiar2020rationalization}%
\begin{APACrefauthors}%
Aguiar%
, Hjertstrand, P.%
\BCBL {}\ \BBA {} Serrano, R.%
\end{APACrefauthors}%
\unskip\
\newblock
\APACrefYearMonthDay{2020}{}{}.
\newblock
{\BBOQ}\APACrefatitle {A rationalization of the weak axiom of revealed preference} {A rationalization of the weak axiom of revealed preference}.{\BBCQ}
\newblock

\PrintBackRefs{\CurrentBib}

\bibitem [\protect \citeauthoryear {%
Aguiar%
\ \BBA {} Kashaev%
}{%
Aguiar%
\ \BBA {} Kashaev%
}{%
{\protect \APACyear {2021}}%
}]{%
aguiar2021stochastic}
\APACinsertmetastar {%
aguiar2021stochastic}%
\begin{APACrefauthors}%
Aguiar%
\BCBT {}\ \BBA {} Kashaev, N.%
\end{APACrefauthors}%
\unskip\
\newblock
\APACrefYearMonthDay{2021}{}{}.
\newblock
{\BBOQ}\APACrefatitle {Stochastic revealed preferences with measurement error} {Stochastic revealed preferences with measurement error}.{\BBCQ}
\newblock
\APACjournalVolNumPages{The Review of Economic Studies}{88}{4}{2042--2093}.
\PrintBackRefs{\CurrentBib}

\bibitem [\protect \citeauthoryear {%
Aguiar%
\ \BBA {} Serrano%
}{%
Aguiar%
\ \BBA {} Serrano%
}{%
{\protect \APACyear {2021}}%
}]{%
aguiar2021cardinal}
\APACinsertmetastar {%
aguiar2021cardinal}%
\begin{APACrefauthors}%
Aguiar%
\BCBT {}\ \BBA {} Serrano, R.%
\end{APACrefauthors}%
\unskip\
\newblock
\APACrefYearMonthDay{2021}{}{}.
\newblock
{\BBOQ}\APACrefatitle {Cardinal revealed preference: Disentangling transitivity and consistent binary choice} {Cardinal revealed preference: Disentangling transitivity and consistent binary choice}.{\BBCQ}
\newblock
\APACjournalVolNumPages{Journal of Mathematical Economics}{94}{}{102462}.
\PrintBackRefs{\CurrentBib}

\bibitem [\protect \citeauthoryear {%
Al{\'o}s-Ferrer%
, Fehr%
\BCBL {}\ \BBA {} Garagnani%
}{%
Al{\'o}s-Ferrer%
\ \protect \BOthers {.}}{%
{\protect \APACyear {2022}}%
}]{%
alos2022identifying}
\APACinsertmetastar {%
alos2022identifying}%
\begin{APACrefauthors}%
Al{\'o}s-Ferrer, C.%
, Fehr, E.%
\BCBL {}\ \BBA {} Garagnani, M.%
\end{APACrefauthors}%
\unskip\
\newblock
\APACrefYearMonthDay{2022}{}{}.
\newblock
\APACrefbtitle {Identifying nontransitive preferences} {Identifying nontransitive preferences}\ \APACbVolEdTR{}{\BTR{}}.
\newblock
\APACaddressInstitution{}{Working Paper}.
\PrintBackRefs{\CurrentBib}

\bibitem [\protect \citeauthoryear {%
Apesteguia%
\ \BBA {} Ballester%
}{%
Apesteguia%
\ \BBA {} Ballester%
}{%
{\protect \APACyear {2018}}%
}]{%
apesteguia2018monotone}
\APACinsertmetastar {%
apesteguia2018monotone}%
\begin{APACrefauthors}%
Apesteguia, J.%
\BCBT {}\ \BBA {} Ballester, M\BPBI A.%
\end{APACrefauthors}%
\unskip\
\newblock
\APACrefYearMonthDay{2018}{}{}.
\newblock
{\BBOQ}\APACrefatitle {Monotone stochastic choice models: The case of risk and time preferences} {Monotone stochastic choice models: The case of risk and time preferences}.{\BBCQ}
\newblock
\APACjournalVolNumPages{Journal of Political Economy}{126}{1}{74--106}.
\PrintBackRefs{\CurrentBib}

\bibitem [\protect \citeauthoryear {%
Bandyopadhyay%
, Dasgupta%
\BCBL {}\ \BBA {} Pattanaik%
}{%
Bandyopadhyay%
\ \protect \BOthers {.}}{%
{\protect \APACyear {1999}}%
}]{%
bandyopadhyay1999stochastic}
\APACinsertmetastar {%
bandyopadhyay1999stochastic}%
\begin{APACrefauthors}%
Bandyopadhyay, T.%
, Dasgupta, I.%
\BCBL {}\ \BBA {} Pattanaik, P\BPBI K.%
\end{APACrefauthors}%
\unskip\
\newblock
\APACrefYearMonthDay{1999}{}{}.
\newblock
{\BBOQ}\APACrefatitle {Stochastic revealed preference and the theory of demand} {Stochastic revealed preference and the theory of demand}.{\BBCQ}
\newblock
\APACjournalVolNumPages{Journal of Economic Theory}{84}{1}{95--110}.
\PrintBackRefs{\CurrentBib}

\bibitem [\protect \citeauthoryear {%
Ben-Akiva%
\ \BBA {} Lerman%
}{%
Ben-Akiva%
\ \BBA {} Lerman%
}{%
{\protect \APACyear {1985}}%
}]{%
ben1985discrete}
\APACinsertmetastar {%
ben1985discrete}%
\begin{APACrefauthors}%
Ben-Akiva, M\BPBI E.%
\BCBT {}\ \BBA {} Lerman, S\BPBI R.%
\end{APACrefauthors}%
\unskip\
\newblock
\APACrefYear{1985}.
\newblock
\APACrefbtitle {Discrete choice analysis: theory and application to travel demand} {Discrete choice analysis: theory and application to travel demand}\ (\BVOL~9).
\newblock
\APACaddressPublisher{}{MIT press}.
\PrintBackRefs{\CurrentBib}

\bibitem [\protect \citeauthoryear {%
Berry%
, Levinsohn%
\BCBL {}\ \BBA {} Pakes%
}{%
Berry%
\ \protect \BOthers {.}}{%
{\protect \APACyear {1993}}%
}]{%
berry1993automobile}
\APACinsertmetastar {%
berry1993automobile}%
\begin{APACrefauthors}%
Berry, S\BPBI T.%
, Levinsohn, J\BPBI A.%
\BCBL {}\ \BBA {} Pakes, A.%
\end{APACrefauthors}%
\unskip\
\newblock
\APACrefYearMonthDay{1993}{}{}.
\newblock
\APACrefbtitle {Automobile prices in market equilibrium: Part I and II.} {Automobile prices in market equilibrium: Part i and ii.}
\newblock
\APACaddressPublisher{}{National Bureau of Economic Research Cambridge, Mass., USA}.
\PrintBackRefs{\CurrentBib}

\bibitem [\protect \citeauthoryear {%
Bhattacharya%
, Dupas%
\BCBL {}\ \BBA {} Kanaya%
}{%
Bhattacharya%
\ \protect \BOthers {.}}{%
{\protect \APACyear {2019}}%
}]{%
bhattacharya2019demand}
\APACinsertmetastar {%
bhattacharya2019demand}%
\begin{APACrefauthors}%
Bhattacharya, D.%
, Dupas, P.%
\BCBL {}\ \BBA {} Kanaya, S.%
\end{APACrefauthors}%
\unskip\
\newblock
\APACrefYearMonthDay{2019}{}{}.
\newblock
\APACrefbtitle {Demand and welfare analysis in discrete choice models with social interactions} {Demand and welfare analysis in discrete choice models with social interactions}\ \APACbVolEdTR{}{\BTR{}}.
\newblock
\APACaddressInstitution{}{National Bureau of Economic Research}.
\PrintBackRefs{\CurrentBib}

\bibitem [\protect \citeauthoryear {%
Blundell%
\ \protect \BOthers {.}}{%
Blundell%
\ \protect \BOthers {.}}{%
{\protect \APACyear {2015}}%
}]{%
blundell2015sharp}
\APACinsertmetastar {%
blundell2015sharp}%
\begin{APACrefauthors}%
Blundell%
, Browning, M.%
, Cherchye, L.%
, Crawford, I.%
, De~Rock, B.%
\BCBL {}\ \BBA {} Vermeulen, F.%
\end{APACrefauthors}%
\unskip\
\newblock
\APACrefYearMonthDay{2015}{}{}.
\newblock
{\BBOQ}\APACrefatitle {Sharp for SARP: nonparametric bounds on counterfactual demands} {Sharp for sarp: nonparametric bounds on counterfactual demands}.{\BBCQ}
\newblock
\APACjournalVolNumPages{American Economic Journal: Microeconomics}{7}{1}{43--60}.
\PrintBackRefs{\CurrentBib}

\bibitem [\protect \citeauthoryear {%
Blundell%
, Browning%
\BCBL {}\ \BBA {} Crawford%
}{%
Blundell%
\ \protect \BOthers {.}}{%
{\protect \APACyear {2008}}%
}]{%
blundell2008best}
\APACinsertmetastar {%
blundell2008best}%
\begin{APACrefauthors}%
Blundell%
, Browning, M.%
\BCBL {}\ \BBA {} Crawford, I.%
\end{APACrefauthors}%
\unskip\
\newblock
\APACrefYearMonthDay{2008}{}{}.
\newblock
{\BBOQ}\APACrefatitle {Best nonparametric bounds on demand responses} {Best nonparametric bounds on demand responses}.{\BBCQ}
\newblock
\APACjournalVolNumPages{Econometrica}{76}{6}{1227--1262}.
\PrintBackRefs{\CurrentBib}

\bibitem [\protect \citeauthoryear {%
Browning%
, Chiappori%
\BCBL {}\ \BBA {} Weiss%
}{%
Browning%
\ \protect \BOthers {.}}{%
{\protect \APACyear {2010}}%
}]{%
browning2010uncertainty}
\APACinsertmetastar {%
browning2010uncertainty}%
\begin{APACrefauthors}%
Browning, M.%
, Chiappori, P\BHBI A.%
\BCBL {}\ \BBA {} Weiss, Y.%
\end{APACrefauthors}%
\unskip\
\newblock
\APACrefYearMonthDay{2010}{}{}.
\newblock
{\BBOQ}\APACrefatitle {Uncertainty and dynamics in the collective model} {Uncertainty and dynamics in the collective model}.{\BBCQ}
\newblock

\PrintBackRefs{\CurrentBib}

\bibitem [\protect \citeauthoryear {%
Cherchye%
, Demuynck%
\BCBL {}\ \BBA {} De~Rock%
}{%
Cherchye%
\ \protect \BOthers {.}}{%
{\protect \APACyear {2018}}%
}]{%
cherchye2018transitivity}
\APACinsertmetastar {%
cherchye2018transitivity}%
\begin{APACrefauthors}%
Cherchye, L.%
, Demuynck, T.%
\BCBL {}\ \BBA {} De~Rock, B.%
\end{APACrefauthors}%
\unskip\
\newblock
\APACrefYearMonthDay{2018}{}{}.
\newblock
{\BBOQ}\APACrefatitle {Transitivity of preferences: When does it matter?} {Transitivity of preferences: When does it matter?}{\BBCQ}
\newblock
\APACjournalVolNumPages{Theoretical Economics}{13}{3}{1043--1076}.
\PrintBackRefs{\CurrentBib}

\bibitem [\protect \citeauthoryear {%
Cherchye%
, Demuynck%
\BCBL {}\ \BBA {} De~Rock%
}{%
Cherchye%
\ \protect \BOthers {.}}{%
{\protect \APACyear {2019}}%
}]{%
cherchye2019bounding}
\APACinsertmetastar {%
cherchye2019bounding}%
\begin{APACrefauthors}%
Cherchye, L.%
, Demuynck, T.%
\BCBL {}\ \BBA {} De~Rock, B.%
\end{APACrefauthors}%
\unskip\
\newblock
\APACrefYearMonthDay{2019}{}{}.
\newblock
{\BBOQ}\APACrefatitle {Bounding counterfactual demand with unobserved heterogeneity and endogenous expenditures} {Bounding counterfactual demand with unobserved heterogeneity and endogenous expenditures}.{\BBCQ}
\newblock
\APACjournalVolNumPages{Journal of econometrics}{211}{2}{483--506}.
\PrintBackRefs{\CurrentBib}

\bibitem [\protect \citeauthoryear {%
Cosaert%
\ \BBA {} Demuynck%
}{%
Cosaert%
\ \BBA {} Demuynck%
}{%
{\protect \APACyear {2018}}%
}]{%
cosaert2018nonparametric}
\APACinsertmetastar {%
cosaert2018nonparametric}%
\begin{APACrefauthors}%
Cosaert, S.%
\BCBT {}\ \BBA {} Demuynck, T.%
\end{APACrefauthors}%
\unskip\
\newblock
\APACrefYearMonthDay{2018}{}{}.
\newblock
{\BBOQ}\APACrefatitle {Nonparametric welfare and demand analysis with unobserved individual heterogeneity} {Nonparametric welfare and demand analysis with unobserved individual heterogeneity}.{\BBCQ}
\newblock
\APACjournalVolNumPages{Review of Economics and Statistics}{100}{2}{349--361}.
\PrintBackRefs{\CurrentBib}

\bibitem [\protect \citeauthoryear {%
Dardanoni%
, Manzini%
, Mariotti%
\BCBL {}\ \BBA {} Tyson%
}{%
Dardanoni%
\ \protect \BOthers {.}}{%
{\protect \APACyear {2020}}%
}]{%
dardanoni2020inferring}
\APACinsertmetastar {%
dardanoni2020inferring}%
\begin{APACrefauthors}%
Dardanoni, V.%
, Manzini, P.%
, Mariotti, M.%
\BCBL {}\ \BBA {} Tyson, C\BPBI J.%
\end{APACrefauthors}%
\unskip\
\newblock
\APACrefYearMonthDay{2020}{}{}.
\newblock
{\BBOQ}\APACrefatitle {Inferring cognitive heterogeneity from aggregate choices} {Inferring cognitive heterogeneity from aggregate choices}.{\BBCQ}
\newblock
\APACjournalVolNumPages{Econometrica}{88}{3}{1269--1296}.
\PrintBackRefs{\CurrentBib}

\bibitem [\protect \citeauthoryear {%
Dasgupta%
\ \BBA {} Pattanaik%
}{%
Dasgupta%
\ \BBA {} Pattanaik%
}{%
{\protect \APACyear {2007}}%
}]{%
dasgupta2007regular}
\APACinsertmetastar {%
dasgupta2007regular}%
\begin{APACrefauthors}%
Dasgupta, I.%
\BCBT {}\ \BBA {} Pattanaik, P\BPBI K.%
\end{APACrefauthors}%
\unskip\
\newblock
\APACrefYearMonthDay{2007}{}{}.
\newblock
{\BBOQ}\APACrefatitle {‘Regular’choice and the weak axiom of stochastic revealed preference} {‘regular’choice and the weak axiom of stochastic revealed preference}.{\BBCQ}
\newblock
\APACjournalVolNumPages{Economic Theory}{31}{}{35--50}.
\PrintBackRefs{\CurrentBib}

\bibitem [\protect \citeauthoryear {%
David%
}{%
David%
}{%
{\protect \APACyear {1963}}%
}]{%
david1963method}
\APACinsertmetastar {%
david1963method}%
\begin{APACrefauthors}%
David, H\BPBI A.%
\end{APACrefauthors}%
\unskip\
\newblock
\APACrefYear{1963}.
\newblock
\APACrefbtitle {The method of paired comparisons} {The method of paired comparisons}\ (\BVOL~12).
\newblock
\APACaddressPublisher{}{London}.
\PrintBackRefs{\CurrentBib}

\bibitem [\protect \citeauthoryear {%
Deb%
, Kitamura%
, Quah%
\BCBL {}\ \BBA {} Stoye%
}{%
Deb%
\ \protect \BOthers {.}}{%
{\protect \APACyear {2023}}%
}]{%
deb2023revealed}
\APACinsertmetastar {%
deb2023revealed}%
\begin{APACrefauthors}%
Deb, R.%
, Kitamura, Y.%
, Quah, J\BPBI K.%
\BCBL {}\ \BBA {} Stoye, J.%
\end{APACrefauthors}%
\unskip\
\newblock
\APACrefYearMonthDay{2023}{}{}.
\newblock
{\BBOQ}\APACrefatitle {Revealed price preference: theory and empirical analysis} {Revealed price preference: theory and empirical analysis}.{\BBCQ}
\newblock
\APACjournalVolNumPages{The Review of Economic Studies}{90}{2}{707--743}.
\PrintBackRefs{\CurrentBib}

\bibitem [\protect \citeauthoryear {%
Diecidue%
\ \BBA {} Somasundaram%
}{%
Diecidue%
\ \BBA {} Somasundaram%
}{%
{\protect \APACyear {2017}}%
}]{%
diecidue2017regret}
\APACinsertmetastar {%
diecidue2017regret}%
\begin{APACrefauthors}%
Diecidue, E.%
\BCBT {}\ \BBA {} Somasundaram, J.%
\end{APACrefauthors}%
\unskip\
\newblock
\APACrefYearMonthDay{2017}{}{}.
\newblock
{\BBOQ}\APACrefatitle {Regret theory: A new foundation} {Regret theory: A new foundation}.{\BBCQ}
\newblock
\APACjournalVolNumPages{Journal of Economic Theory}{172}{}{88--119}.
\PrintBackRefs{\CurrentBib}

\bibitem [\protect \citeauthoryear {%
Domencich%
\ \BBA {} McFadden%
}{%
Domencich%
\ \BBA {} McFadden%
}{%
{\protect \APACyear {1975}}%
}]{%
domencich1975urban}
\APACinsertmetastar {%
domencich1975urban}%
\begin{APACrefauthors}%
Domencich, T\BPBI A.%
\BCBT {}\ \BBA {} McFadden, D.%
\end{APACrefauthors}%
\unskip\
\newblock
\APACrefYearMonthDay{1975}{}{}.
\newblock
\APACrefbtitle {Urban travel demand-a behavioral analysis} {Urban travel demand-a behavioral analysis}\ \APACbVolEdTR{}{\BTR{}}.
\PrintBackRefs{\CurrentBib}

\bibitem [\protect \citeauthoryear {%
Echenique%
, Lee%
\BCBL {}\ \BBA {} Shum%
}{%
Echenique%
\ \protect \BOthers {.}}{%
{\protect \APACyear {2011}}%
}]{%
echenique2011money}
\APACinsertmetastar {%
echenique2011money}%
\begin{APACrefauthors}%
Echenique, F.%
, Lee, S.%
\BCBL {}\ \BBA {} Shum, M.%
\end{APACrefauthors}%
\unskip\
\newblock
\APACrefYearMonthDay{2011}{}{}.
\newblock
{\BBOQ}\APACrefatitle {The money pump as a measure of revealed preference violations} {The money pump as a measure of revealed preference violations}.{\BBCQ}
\newblock
\APACjournalVolNumPages{Journal of Political Economy}{119}{6}{1201--1223}.
\PrintBackRefs{\CurrentBib}

\bibitem [\protect \citeauthoryear {%
Falmagne%
}{%
Falmagne%
}{%
{\protect \APACyear {1978}}%
}]{%
falmagne1978representation}
\APACinsertmetastar {%
falmagne1978representation}%
\begin{APACrefauthors}%
Falmagne, J\BHBI C.%
\end{APACrefauthors}%
\unskip\
\newblock
\APACrefYearMonthDay{1978}{}{}.
\newblock
{\BBOQ}\APACrefatitle {A representation theorem for finite random scale systems} {A representation theorem for finite random scale systems}.{\BBCQ}
\newblock
\APACjournalVolNumPages{Journal of Mathematical Psychology}{18}{1}{52--72}.
\PrintBackRefs{\CurrentBib}

\bibitem [\protect \citeauthoryear {%
Fishburn%
}{%
Fishburn%
}{%
{\protect \APACyear {1982}}%
}]{%
fishburn1982nontransitive}
\APACinsertmetastar {%
fishburn1982nontransitive}%
\begin{APACrefauthors}%
Fishburn, P\BPBI C.%
\end{APACrefauthors}%
\unskip\
\newblock
\APACrefYearMonthDay{1982}{}{}.
\newblock
{\BBOQ}\APACrefatitle {Nontransitive measurable utility} {Nontransitive measurable utility}.{\BBCQ}
\newblock
\APACjournalVolNumPages{Journal of Mathematical Psychology}{26}{1}{31--67}.
\PrintBackRefs{\CurrentBib}

\bibitem [\protect \citeauthoryear {%
Fishburn%
}{%
Fishburn%
}{%
{\protect \APACyear {1991}}%
}]{%
fishburn1991nontransitive}
\APACinsertmetastar {%
fishburn1991nontransitive}%
\begin{APACrefauthors}%
Fishburn, P\BPBI C.%
\end{APACrefauthors}%
\unskip\
\newblock
\APACrefYearMonthDay{1991}{}{}.
\newblock
{\BBOQ}\APACrefatitle {Nontransitive preferences in decision theory} {Nontransitive preferences in decision theory}.{\BBCQ}
\newblock
\APACjournalVolNumPages{Journal of Risk and Uncertainty}{4}{}{113--134}.
\PrintBackRefs{\CurrentBib}

\bibitem [\protect \citeauthoryear {%
Fosgerau%
\ \BBA {} Rehbeck%
}{%
Fosgerau%
\ \BBA {} Rehbeck%
}{%
{\protect \APACyear {2023}}%
}]{%
fosgerau2023nontransitive}
\APACinsertmetastar {%
fosgerau2023nontransitive}%
\begin{APACrefauthors}%
Fosgerau, M.%
\BCBT {}\ \BBA {} Rehbeck, J.%
\end{APACrefauthors}%
\unskip\
\newblock
\APACrefYearMonthDay{2023}{}{}.
\newblock
{\BBOQ}\APACrefatitle {Nontransitive Preferences and Stochastic Rationalizability: A Behavioral Equivalence} {Nontransitive preferences and stochastic rationalizability: A behavioral equivalence}.{\BBCQ}
\newblock
\APACjournalVolNumPages{arXiv preprint arXiv:2304.14631}{}{}{}.
\PrintBackRefs{\CurrentBib}

\bibitem [\protect \citeauthoryear {%
Fountain%
}{%
Fountain%
}{%
{\protect \APACyear {1981}}%
}]{%
fountain1981consumer}
\APACinsertmetastar {%
fountain1981consumer}%
\begin{APACrefauthors}%
Fountain, J.%
\end{APACrefauthors}%
\unskip\
\newblock
\APACrefYearMonthDay{1981}{}{}.
\newblock
{\BBOQ}\APACrefatitle {Consumer surplus when preferences are intransitive: Analysis and interpretation} {Consumer surplus when preferences are intransitive: Analysis and interpretation}.{\BBCQ}
\newblock
\APACjournalVolNumPages{Econometrica: Journal of the Econometric Society}{}{}{379--394}.
\PrintBackRefs{\CurrentBib}

\bibitem [\protect \citeauthoryear {%
Frick%
, Iijima%
\BCBL {}\ \BBA {} Strzalecki%
}{%
Frick%
\ \protect \BOthers {.}}{%
{\protect \APACyear {2019}}%
}]{%
frick2019dynamic}
\APACinsertmetastar {%
frick2019dynamic}%
\begin{APACrefauthors}%
Frick, M.%
, Iijima, R.%
\BCBL {}\ \BBA {} Strzalecki, T.%
\end{APACrefauthors}%
\unskip\
\newblock
\APACrefYearMonthDay{2019}{}{}.
\newblock
{\BBOQ}\APACrefatitle {Dynamic random utility} {Dynamic random utility}.{\BBCQ}
\newblock
\APACjournalVolNumPages{Econometrica}{87}{6}{1941--2002}.
\PrintBackRefs{\CurrentBib}

\bibitem [\protect \citeauthoryear {%
Hausman%
\ \BBA {} Newey%
}{%
Hausman%
\ \BBA {} Newey%
}{%
{\protect \APACyear {2016}}%
}]{%
hausman2016individual}
\APACinsertmetastar {%
hausman2016individual}%
\begin{APACrefauthors}%
Hausman, J\BPBI A.%
\BCBT {}\ \BBA {} Newey, W\BPBI K.%
\end{APACrefauthors}%
\unskip\
\newblock
\APACrefYearMonthDay{2016}{}{}.
\newblock
{\BBOQ}\APACrefatitle {Individual heterogeneity and average welfare} {Individual heterogeneity and average welfare}.{\BBCQ}
\newblock
\APACjournalVolNumPages{Econometrica}{84}{3}{1225--1248}.
\PrintBackRefs{\CurrentBib}

\bibitem [\protect \citeauthoryear {%
Hoderlein%
\ \BBA {} Stoye%
}{%
Hoderlein%
\ \BBA {} Stoye%
}{%
{\protect \APACyear {2014}}%
}]{%
hoderlein2014revealed}
\APACinsertmetastar {%
hoderlein2014revealed}%
\begin{APACrefauthors}%
Hoderlein, S.%
\BCBT {}\ \BBA {} Stoye, J.%
\end{APACrefauthors}%
\unskip\
\newblock
\APACrefYearMonthDay{2014}{}{}.
\newblock
{\BBOQ}\APACrefatitle {Revealed preferences in a heterogeneous population} {Revealed preferences in a heterogeneous population}.{\BBCQ}
\newblock
\APACjournalVolNumPages{Review of Economics and Statistics}{96}{2}{197--213}.
\PrintBackRefs{\CurrentBib}

\bibitem [\protect \citeauthoryear {%
Houthakker%
}{%
Houthakker%
}{%
{\protect \APACyear {1950}}%
}]{%
houthakker1950revealed}
\APACinsertmetastar {%
houthakker1950revealed}%
\begin{APACrefauthors}%
Houthakker, H\BPBI S.%
\end{APACrefauthors}%
\unskip\
\newblock
\APACrefYearMonthDay{1950}{}{}.
\newblock
{\BBOQ}\APACrefatitle {Revealed preference and the utility function} {Revealed preference and the utility function}.{\BBCQ}
\newblock
\APACjournalVolNumPages{Economica}{17}{66}{159--174}.
\PrintBackRefs{\CurrentBib}

\bibitem [\protect \citeauthoryear {%
Im%
\ \BBA {} Rehbeck%
}{%
Im%
\ \BBA {} Rehbeck%
}{%
{\protect \APACyear {2022}}%
}]{%
im2022non}
\APACinsertmetastar {%
im2022non}%
\begin{APACrefauthors}%
Im, C.%
\BCBT {}\ \BBA {} Rehbeck, J.%
\end{APACrefauthors}%
\unskip\
\newblock
\APACrefYearMonthDay{2022}{}{}.
\newblock
{\BBOQ}\APACrefatitle {Non-rationalizable individuals and stochastic rationalizability} {Non-rationalizable individuals and stochastic rationalizability}.{\BBCQ}
\newblock
\APACjournalVolNumPages{Economics Letters}{219}{}{110786}.
\PrintBackRefs{\CurrentBib}

\bibitem [\protect \citeauthoryear {%
John%
}{%
John%
}{%
{\protect \APACyear {2001}}%
}]{%
john2001concave}
\APACinsertmetastar {%
john2001concave}%
\begin{APACrefauthors}%
John, R.%
\end{APACrefauthors}%
\unskip\
\newblock
\APACrefYearMonthDay{2001}{}{}.
\newblock
{\BBOQ}\APACrefatitle {The concave nontransitive consumer} {The concave nontransitive consumer}.{\BBCQ}
\newblock
\APACjournalVolNumPages{Journal of Global Optimization}{20}{3}{297--308}.
\PrintBackRefs{\CurrentBib}

\bibitem [\protect \citeauthoryear {%
Kalai%
, Rubinstein%
\BCBL {}\ \BBA {} Spiegler%
}{%
Kalai%
\ \protect \BOthers {.}}{%
{\protect \APACyear {2002}}%
}]{%
kalai2002rationalizing}
\APACinsertmetastar {%
kalai2002rationalizing}%
\begin{APACrefauthors}%
Kalai, G.%
, Rubinstein, A.%
\BCBL {}\ \BBA {} Spiegler, R.%
\end{APACrefauthors}%
\unskip\
\newblock
\APACrefYearMonthDay{2002}{}{}.
\newblock
{\BBOQ}\APACrefatitle {Rationalizing choice functions by multiple rationales} {Rationalizing choice functions by multiple rationales}.{\BBCQ}
\newblock
\APACjournalVolNumPages{Econometrica}{70}{6}{2481--2488}.
\PrintBackRefs{\CurrentBib}

\bibitem [\protect \citeauthoryear {%
Kashaev%
\ \BBA {} Aguiar%
}{%
Kashaev%
\ \BBA {} Aguiar%
}{%
{\protect \APACyear {2021}}%
}]{%
kashaev2021random}
\APACinsertmetastar {%
kashaev2021random}%
\begin{APACrefauthors}%
Kashaev, N.%
\BCBT {}\ \BBA {} Aguiar, V\BPBI H.%
\end{APACrefauthors}%
\unskip\
\newblock
\APACrefYearMonthDay{2021}{}{}.
\newblock
{\BBOQ}\APACrefatitle {A Random Attention and Utility Model} {A random attention and utility model}.{\BBCQ}
\newblock
\APACjournalVolNumPages{arXiv preprint arXiv:2105.11268}{}{}{}.
\PrintBackRefs{\CurrentBib}

\bibitem [\protect \citeauthoryear {%
Kashaev%
, Gauthier%
\BCBL {}\ \BBA {} Aguiar%
}{%
Kashaev%
\ \protect \BOthers {.}}{%
{\protect \APACyear {2023}}%
}]{%
kashaev2023dynamic}
\APACinsertmetastar {%
kashaev2023dynamic}%
\begin{APACrefauthors}%
Kashaev, N.%
, Gauthier, C.%
\BCBL {}\ \BBA {} Aguiar, V\BPBI H.%
\end{APACrefauthors}%
\unskip\
\newblock
\APACrefYearMonthDay{2023}{}{}.
\newblock
{\BBOQ}\APACrefatitle {Dynamic and Stochastic Rational Behavior} {Dynamic and stochastic rational behavior}.{\BBCQ}
\newblock
\APACjournalVolNumPages{arXiv preprint arXiv:2302.04417}{}{}{}.
\PrintBackRefs{\CurrentBib}

\bibitem [\protect \citeauthoryear {%
Kawaguchi%
}{%
Kawaguchi%
}{%
{\protect \APACyear {2017}}%
}]{%
kawaguchi2017testing}
\APACinsertmetastar {%
kawaguchi2017testing}%
\begin{APACrefauthors}%
Kawaguchi, K.%
\end{APACrefauthors}%
\unskip\
\newblock
\APACrefYearMonthDay{2017}{}{}.
\newblock
{\BBOQ}\APACrefatitle {Testing rationality without restricting heterogeneity} {Testing rationality without restricting heterogeneity}.{\BBCQ}
\newblock
\APACjournalVolNumPages{Journal of Econometrics}{197}{1}{153--171}.
\PrintBackRefs{\CurrentBib}

\bibitem [\protect \citeauthoryear {%
Kihlstrom%
, Mas-Colell%
\BCBL {}\ \BBA {} Sonnenschein%
}{%
Kihlstrom%
\ \protect \BOthers {.}}{%
{\protect \APACyear {1976}}%
}]{%
kihlstrom1976demand}
\APACinsertmetastar {%
kihlstrom1976demand}%
\begin{APACrefauthors}%
Kihlstrom, R.%
, Mas-Colell, A.%
\BCBL {}\ \BBA {} Sonnenschein, H.%
\end{APACrefauthors}%
\unskip\
\newblock
\APACrefYearMonthDay{1976}{}{}.
\newblock
{\BBOQ}\APACrefatitle {The demand theory of the weak axiom of revealed preference} {The demand theory of the weak axiom of revealed preference}.{\BBCQ}
\newblock
\APACjournalVolNumPages{Econometrica: Journal of the Econometric Society}{}{}{971--978}.
\PrintBackRefs{\CurrentBib}

\bibitem [\protect \citeauthoryear {%
Kitamura%
\ \BBA {} Stoye%
}{%
Kitamura%
\ \BBA {} Stoye%
}{%
{\protect \APACyear {2018}}%
}]{%
kitamura2018nonparametric}
\APACinsertmetastar {%
kitamura2018nonparametric}%
\begin{APACrefauthors}%
Kitamura, Y.%
\BCBT {}\ \BBA {} Stoye, J.%
\end{APACrefauthors}%
\unskip\
\newblock
\APACrefYearMonthDay{2018}{}{}.
\newblock
{\BBOQ}\APACrefatitle {Nonparametric analysis of random utility models} {Nonparametric analysis of random utility models}.{\BBCQ}
\newblock
\APACjournalVolNumPages{Econometrica}{86}{6}{1883--1909}.
\PrintBackRefs{\CurrentBib}

\bibitem [\protect \citeauthoryear {%
Kitamura%
\ \BBA {} Stoye%
}{%
Kitamura%
\ \BBA {} Stoye%
}{%
{\protect \APACyear {2019}}%
}]{%
kitamura2019nonparametric}
\APACinsertmetastar {%
kitamura2019nonparametric}%
\begin{APACrefauthors}%
Kitamura, Y.%
\BCBT {}\ \BBA {} Stoye, J.%
\end{APACrefauthors}%
\unskip\
\newblock
\APACrefYearMonthDay{2019}{}{}.
\newblock
{\BBOQ}\APACrefatitle {Nonparametric counterfactuals in random utility models} {Nonparametric counterfactuals in random utility models}.{\BBCQ}
\newblock
\APACjournalVolNumPages{arXiv preprint arXiv:1902.08350}{}{}{}.
\PrintBackRefs{\CurrentBib}

\bibitem [\protect \citeauthoryear {%
Koida%
\ \BBA {} Shirai%
}{%
Koida%
\ \BBA {} Shirai%
}{%
{\protect \APACyear {2024}}%
}]{%
koida2024dual}
\APACinsertmetastar {%
koida2024dual}%
\begin{APACrefauthors}%
Koida, N.%
\BCBT {}\ \BBA {} Shirai, K.%
\end{APACrefauthors}%
\unskip\
\newblock
\APACrefYearMonthDay{2024}{}{}.
\newblock
{\BBOQ}\APACrefatitle {A dual approach to nonparametric characterization for random utility models} {A dual approach to nonparametric characterization for random utility models}.{\BBCQ}
\newblock
\APACjournalVolNumPages{arXiv preprint arXiv:2403.04328}{}{}{}.
\PrintBackRefs{\CurrentBib}

\bibitem [\protect \citeauthoryear {%
Loomes%
, Starmer%
\BCBL {}\ \BBA {} Sugden%
}{%
Loomes%
\ \protect \BOthers {.}}{%
{\protect \APACyear {1991}}%
}]{%
loomes1991observing}
\APACinsertmetastar {%
loomes1991observing}%
\begin{APACrefauthors}%
Loomes, G.%
, Starmer, C.%
\BCBL {}\ \BBA {} Sugden, R.%
\end{APACrefauthors}%
\unskip\
\newblock
\APACrefYearMonthDay{1991}{}{}.
\newblock
{\BBOQ}\APACrefatitle {Observing violations of transitivity by experimental methods} {Observing violations of transitivity by experimental methods}.{\BBCQ}
\newblock
\APACjournalVolNumPages{Econometrica: Journal of the Econometric Society}{}{}{425--439}.
\PrintBackRefs{\CurrentBib}

\bibitem [\protect \citeauthoryear {%
Loomes%
\ \BBA {} Sugden%
}{%
Loomes%
\ \BBA {} Sugden%
}{%
{\protect \APACyear {1982}}%
}]{%
loomes1982regret}
\APACinsertmetastar {%
loomes1982regret}%
\begin{APACrefauthors}%
Loomes, G.%
\BCBT {}\ \BBA {} Sugden, R.%
\end{APACrefauthors}%
\unskip\
\newblock
\APACrefYearMonthDay{1982}{}{}.
\newblock
{\BBOQ}\APACrefatitle {Regret theory: An alternative theory of rational choice under uncertainty} {Regret theory: An alternative theory of rational choice under uncertainty}.{\BBCQ}
\newblock
\APACjournalVolNumPages{The economic journal}{92}{368}{805--824}.
\PrintBackRefs{\CurrentBib}

\bibitem [\protect \citeauthoryear {%
Lu%
}{%
Lu%
}{%
{\protect \APACyear {2019}}%
}]{%
lu2019bayesian}
\APACinsertmetastar {%
lu2019bayesian}%
\begin{APACrefauthors}%
Lu, J.%
\end{APACrefauthors}%
\unskip\
\newblock
\APACrefYearMonthDay{2019}{}{}.
\newblock
{\BBOQ}\APACrefatitle {Bayesian identification: a theory for state-dependent utilities} {Bayesian identification: a theory for state-dependent utilities}.{\BBCQ}
\newblock
\APACjournalVolNumPages{American Economic Review}{109}{9}{3192--3228}.
\PrintBackRefs{\CurrentBib}

\bibitem [\protect \citeauthoryear {%
Luce%
}{%
Luce%
}{%
{\protect \APACyear {1956}}%
}]{%
luce1956semiorders}
\APACinsertmetastar {%
luce1956semiorders}%
\begin{APACrefauthors}%
Luce, R\BPBI D.%
\end{APACrefauthors}%
\unskip\
\newblock
\APACrefYearMonthDay{1956}{}{}.
\newblock
{\BBOQ}\APACrefatitle {Semiorders and a theory of utility discrimination} {Semiorders and a theory of utility discrimination}.{\BBCQ}
\newblock
\APACjournalVolNumPages{Econometrica, Journal of the Econometric Society}{}{}{178--191}.
\PrintBackRefs{\CurrentBib}

\bibitem [\protect \citeauthoryear {%
Manski%
}{%
Manski%
}{%
{\protect \APACyear {2003}}%
}]{%
manski2003partial}
\APACinsertmetastar {%
manski2003partial}%
\begin{APACrefauthors}%
Manski, C\BPBI F.%
\end{APACrefauthors}%
\unskip\
\newblock
\APACrefYear{2003}.
\newblock
\APACrefbtitle {Partial identification of probability distributions} {Partial identification of probability distributions}\ (\BVOL~5).
\newblock
\APACaddressPublisher{}{Springer}.
\PrintBackRefs{\CurrentBib}

\bibitem [\protect \citeauthoryear {%
Manzini%
\ \BBA {} Mariotti%
}{%
Manzini%
\ \BBA {} Mariotti%
}{%
{\protect \APACyear {2007}}%
}]{%
manzini2007sequentially}
\APACinsertmetastar {%
manzini2007sequentially}%
\begin{APACrefauthors}%
Manzini, P.%
\BCBT {}\ \BBA {} Mariotti, M.%
\end{APACrefauthors}%
\unskip\
\newblock
\APACrefYearMonthDay{2007}{}{}.
\newblock
{\BBOQ}\APACrefatitle {Sequentially rationalizable choice} {Sequentially rationalizable choice}.{\BBCQ}
\newblock
\APACjournalVolNumPages{American Economic Review}{97}{5}{1824--1839}.
\PrintBackRefs{\CurrentBib}

\bibitem [\protect \citeauthoryear {%
Matzkin%
\ \BBA {} Richter%
}{%
Matzkin%
\ \BBA {} Richter%
}{%
{\protect \APACyear {1991}}%
}]{%
matzkin1991testing}
\APACinsertmetastar {%
matzkin1991testing}%
\begin{APACrefauthors}%
Matzkin, R\BPBI L.%
\BCBT {}\ \BBA {} Richter, M\BPBI K.%
\end{APACrefauthors}%
\unskip\
\newblock
\APACrefYearMonthDay{1991}{}{}.
\newblock
{\BBOQ}\APACrefatitle {Testing strictly concave rationality} {Testing strictly concave rationality}.{\BBCQ}
\newblock
\APACjournalVolNumPages{Journal of Economic Theory}{53}{2}{287--303}.
\PrintBackRefs{\CurrentBib}

\bibitem [\protect \citeauthoryear {%
McFadden%
}{%
McFadden%
}{%
{\protect \APACyear {2006}}%
}]{%
mcfadden2006revealed}
\APACinsertmetastar {%
mcfadden2006revealed}%
\begin{APACrefauthors}%
McFadden.%
\end{APACrefauthors}%
\unskip\
\newblock
\APACrefYearMonthDay{2006}{}{}.
\newblock
{\BBOQ}\APACrefatitle {Revealed stochastic preference: {A} synthesis} {Revealed stochastic preference: {A} synthesis}.{\BBCQ}
\newblock
\BIn{} \APACrefbtitle {Rationality and {E}quilibrium} {Rationality and {E}quilibrium}\ (\BPGS\ 1--20).
\newblock
\APACaddressPublisher{}{Springer}.
\PrintBackRefs{\CurrentBib}

\bibitem [\protect \citeauthoryear {%
McFadden%
\ \BBA {} Richter%
}{%
McFadden%
\ \BBA {} Richter%
}{%
{\protect \APACyear {1971}}%
}]{%
mcfadden1971extension}
\APACinsertmetastar {%
mcfadden1971extension}%
\begin{APACrefauthors}%
McFadden%
\BCBT {}\ \BBA {} Richter, M.%
\end{APACrefauthors}%
\unskip\
\newblock
\APACrefYearMonthDay{1971}{}{}.
\newblock
{\BBOQ}\APACrefatitle {On the Extension of a Set Function on a Set of Events to a Probability on the Generated {Boolean} $\sigma$-algebra} {On the extension of a set function on a set of events to a probability on the generated {Boolean} $\sigma$-algebra}.{\BBCQ}
\newblock
\APACjournalVolNumPages{University of California, Berkeley, working paper}{}{}{}.
\PrintBackRefs{\CurrentBib}

\bibitem [\protect \citeauthoryear {%
McFadden%
\ \BBA {} Richter%
}{%
McFadden%
\ \BBA {} Richter%
}{%
{\protect \APACyear {1990}}%
}]{%
mcfadden1990stochastic}
\APACinsertmetastar {%
mcfadden1990stochastic}%
\begin{APACrefauthors}%
McFadden%
\BCBT {}\ \BBA {} Richter, M.%
\end{APACrefauthors}%
\unskip\
\newblock
\APACrefYearMonthDay{1990}{}{}.
\newblock
{\BBOQ}\APACrefatitle {Stochastic rationality and revealed stochastic preference} {Stochastic rationality and revealed stochastic preference}.{\BBCQ}
\newblock
\APACjournalVolNumPages{Preferences, Uncertainty, and Optimality, Essays in Honor of Leo Hurwicz, Westview Press: Boulder, CO}{}{}{161--186}.
\PrintBackRefs{\CurrentBib}

\bibitem [\protect \citeauthoryear {%
Nishimura%
\ \BBA {} Ok%
}{%
Nishimura%
\ \BBA {} Ok%
}{%
{\protect \APACyear {2016}}%
}]{%
nishimura2016utility}
\APACinsertmetastar {%
nishimura2016utility}%
\begin{APACrefauthors}%
Nishimura, H.%
\BCBT {}\ \BBA {} Ok, E\BPBI A.%
\end{APACrefauthors}%
\unskip\
\newblock
\APACrefYearMonthDay{2016}{}{}.
\newblock
{\BBOQ}\APACrefatitle {Utility representation of an incomplete and nontransitive preference relation} {Utility representation of an incomplete and nontransitive preference relation}.{\BBCQ}
\newblock
\APACjournalVolNumPages{Journal of Economic Theory}{166}{}{164--185}.
\PrintBackRefs{\CurrentBib}

\bibitem [\protect \citeauthoryear {%
Ok%
\ \BBA {} Masatlioglu%
}{%
Ok%
\ \BBA {} Masatlioglu%
}{%
{\protect \APACyear {2007}}%
}]{%
ok2007theory}
\APACinsertmetastar {%
ok2007theory}%
\begin{APACrefauthors}%
Ok, E\BPBI A.%
\BCBT {}\ \BBA {} Masatlioglu, Y.%
\end{APACrefauthors}%
\unskip\
\newblock
\APACrefYearMonthDay{2007}{}{}.
\newblock
{\BBOQ}\APACrefatitle {A theory of (relative) discounting} {A theory of (relative) discounting}.{\BBCQ}
\newblock
\APACjournalVolNumPages{Journal of Economic Theory}{137}{1}{214--245}.
\PrintBackRefs{\CurrentBib}

\bibitem [\protect \citeauthoryear {%
Pakes%
}{%
Pakes%
}{%
{\protect \APACyear {1984}}%
}]{%
pakes1984patents}
\APACinsertmetastar {%
pakes1984patents}%
\begin{APACrefauthors}%
Pakes, A.%
\end{APACrefauthors}%
\unskip\
\newblock
\APACrefYearMonthDay{1984}{}{}.
\newblock
\APACrefbtitle {Patents as options: Some estimates of the value of holding European patent stocks} {Patents as options: Some estimates of the value of holding european patent stocks}\ \APACbVolEdTR{}{\BTR{}}.
\newblock
\APACaddressInstitution{}{National Bureau of Economic Research}.
\PrintBackRefs{\CurrentBib}

\bibitem [\protect \citeauthoryear {%
Quah%
}{%
Quah%
}{%
{\protect \APACyear {2006}}%
}]{%
quah2006weak}
\APACinsertmetastar {%
quah2006weak}%
\begin{APACrefauthors}%
Quah, J\BPBI K.%
\end{APACrefauthors}%
\unskip\
\newblock
\APACrefYearMonthDay{2006}{}{}.
\newblock
{\BBOQ}\APACrefatitle {Weak axiomatic demand theory} {Weak axiomatic demand theory}.{\BBCQ}
\newblock
\APACjournalVolNumPages{Economic Theory}{29}{3}{677--699}.
\PrintBackRefs{\CurrentBib}

\bibitem [\protect \citeauthoryear {%
Roelofsma%
\ \BBA {} Read%
}{%
Roelofsma%
\ \BBA {} Read%
}{%
{\protect \APACyear {2000}}%
}]{%
roelofsma2000intransitive}
\APACinsertmetastar {%
roelofsma2000intransitive}%
\begin{APACrefauthors}%
Roelofsma, P\BPBI H.%
\BCBT {}\ \BBA {} Read, D.%
\end{APACrefauthors}%
\unskip\
\newblock
\APACrefYearMonthDay{2000}{}{}.
\newblock
{\BBOQ}\APACrefatitle {Intransitive intertemporal choice} {Intransitive intertemporal choice}.{\BBCQ}
\newblock
\APACjournalVolNumPages{Journal of Behavioral Decision Making}{13}{2}{161--177}.
\PrintBackRefs{\CurrentBib}

\bibitem [\protect \citeauthoryear {%
Rubinstien%
}{%
Rubinstien%
}{%
{\protect \APACyear {1988}}%
}]{%
rubinstien1988similarity}
\APACinsertmetastar {%
rubinstien1988similarity}%
\begin{APACrefauthors}%
Rubinstien, A.%
\end{APACrefauthors}%
\unskip\
\newblock
\APACrefYearMonthDay{1988}{}{}.
\newblock
{\BBOQ}\APACrefatitle {Similarity and Decision making under Risk} {Similarity and decision making under risk}.{\BBCQ}
\newblock
\APACjournalVolNumPages{Journal of Economic Theory}{46}{}{145--153}.
\PrintBackRefs{\CurrentBib}

\bibitem [\protect \citeauthoryear {%
Samuelson%
}{%
Samuelson%
}{%
{\protect \APACyear {1938}}%
}]{%
samuelson1938note}
\APACinsertmetastar {%
samuelson1938note}%
\begin{APACrefauthors}%
Samuelson, P\BPBI A.%
\end{APACrefauthors}%
\unskip\
\newblock
\APACrefYearMonthDay{1938}{}{}.
\newblock
{\BBOQ}\APACrefatitle {A note on the pure theory of consumer's behaviour} {A note on the pure theory of consumer's behaviour}.{\BBCQ}
\newblock
\APACjournalVolNumPages{Economica}{5}{17}{61--71}.
\PrintBackRefs{\CurrentBib}

\bibitem [\protect \citeauthoryear {%
Shafer%
}{%
Shafer%
}{%
{\protect \APACyear {1974}}%
}]{%
shafer1974nontransitive}
\APACinsertmetastar {%
shafer1974nontransitive}%
\begin{APACrefauthors}%
Shafer, W\BPBI J.%
\end{APACrefauthors}%
\unskip\
\newblock
\APACrefYearMonthDay{1974}{}{}.
\newblock
{\BBOQ}\APACrefatitle {The nontransitive consumer} {The nontransitive consumer}.{\BBCQ}
\newblock
\APACjournalVolNumPages{Econometrica: Journal of the Econometric Society}{}{}{913--919}.
\PrintBackRefs{\CurrentBib}

\bibitem [\protect \citeauthoryear {%
Sonnenschein%
}{%
Sonnenschein%
}{%
{\protect \APACyear {1971}}%
}]{%
sonnenschein1971demand}
\APACinsertmetastar {%
sonnenschein1971demand}%
\begin{APACrefauthors}%
Sonnenschein, H.%
\end{APACrefauthors}%
\unskip\
\newblock
\APACrefYearMonthDay{1971}{}{}.
\newblock
{\BBOQ}\APACrefatitle {Demand theory without transitive preferences, with applications to the theory of competitive equilibrium} {Demand theory without transitive preferences, with applications to the theory of competitive equilibrium}.{\BBCQ}
\newblock
\APACjournalVolNumPages{Preferences, utility, and demand}{}{}{215--223}.
\PrintBackRefs{\CurrentBib}

\bibitem [\protect \citeauthoryear {%
Stoye%
}{%
Stoye%
}{%
{\protect \APACyear {2019}}%
}]{%
stoye2019revealed}
\APACinsertmetastar {%
stoye2019revealed}%
\begin{APACrefauthors}%
Stoye, J.%
\end{APACrefauthors}%
\unskip\
\newblock
\APACrefYearMonthDay{2019}{}{}.
\newblock
{\BBOQ}\APACrefatitle {Revealed Stochastic Preference: A one-paragraph proof and generalization} {Revealed stochastic preference: A one-paragraph proof and generalization}.{\BBCQ}
\newblock
\APACjournalVolNumPages{Economics Letters}{177}{}{66--68}.
\PrintBackRefs{\CurrentBib}

\bibitem [\protect \citeauthoryear {%
Tebaldi%
, Torgovitsky%
\BCBL {}\ \BBA {} Yang%
}{%
Tebaldi%
\ \protect \BOthers {.}}{%
{\protect \APACyear {2023}}%
}]{%
tebaldi2023nonparametric}
\APACinsertmetastar {%
tebaldi2023nonparametric}%
\begin{APACrefauthors}%
Tebaldi, P.%
, Torgovitsky, A.%
\BCBL {}\ \BBA {} Yang, H.%
\end{APACrefauthors}%
\unskip\
\newblock
\APACrefYearMonthDay{2023}{}{}.
\newblock
{\BBOQ}\APACrefatitle {Nonparametric estimates of demand in the california health insurance exchange} {Nonparametric estimates of demand in the california health insurance exchange}.{\BBCQ}
\newblock
\APACjournalVolNumPages{Econometrica}{91}{1}{107--146}.
\PrintBackRefs{\CurrentBib}

\bibitem [\protect \citeauthoryear {%
Turansick%
}{%
Turansick%
}{%
{\protect \APACyear {2022}}%
}]{%
turansick2022identification}
\APACinsertmetastar {%
turansick2022identification}%
\begin{APACrefauthors}%
Turansick, C.%
\end{APACrefauthors}%
\unskip\
\newblock
\APACrefYearMonthDay{2022}{}{}.
\newblock
{\BBOQ}\APACrefatitle {Identification in the random utility model} {Identification in the random utility model}.{\BBCQ}
\newblock
\APACjournalVolNumPages{Journal of Economic Theory}{203}{}{105489}.
\PrintBackRefs{\CurrentBib}

\bibitem [\protect \citeauthoryear {%
Tversky%
}{%
Tversky%
}{%
{\protect \APACyear {1969}}%
}]{%
tversky1969intransitivity}
\APACinsertmetastar {%
tversky1969intransitivity}%
\begin{APACrefauthors}%
Tversky, A.%
\end{APACrefauthors}%
\unskip\
\newblock
\APACrefYearMonthDay{1969}{}{}.
\newblock
{\BBOQ}\APACrefatitle {Intransitivity of preferences.} {Intransitivity of preferences.}{\BBCQ}
\newblock
\APACjournalVolNumPages{Psychological Review}{76}{1}{31}.
\PrintBackRefs{\CurrentBib}

\bibitem [\protect \citeauthoryear {%
Wei%
}{%
Wei%
}{%
{\protect \APACyear {2024}}%
}]{%
wei2024random}
\APACinsertmetastar {%
wei2024random}%
\begin{APACrefauthors}%
Wei, D.%
\end{APACrefauthors}%
\unskip\
\newblock
\APACrefYearMonthDay{2024}{}{}.
\newblock
{\BBOQ}\APACrefatitle {Random Attention Span} {Random attention span}.{\BBCQ}
\newblock
\APACjournalVolNumPages{arXiv preprint arXiv:2405.11578}{}{}{}.
\PrintBackRefs{\CurrentBib}

\bibitem [\protect \citeauthoryear {%
Yang%
\ \BBA {} Kopylov%
}{%
Yang%
\ \BBA {} Kopylov%
}{%
{\protect \APACyear {2023}}%
}]{%
yang2023random}
\APACinsertmetastar {%
yang2023random}%
\begin{APACrefauthors}%
Yang, E.%
\BCBT {}\ \BBA {} Kopylov, I.%
\end{APACrefauthors}%
\unskip\
\newblock
\APACrefYearMonthDay{2023}{}{}.
\newblock
{\BBOQ}\APACrefatitle {Random quasi-linear utility} {Random quasi-linear utility}.{\BBCQ}
\newblock
\APACjournalVolNumPages{Journal of Economic Theory}{209}{}{105650}.
\PrintBackRefs{\CurrentBib}

\end{thebibliography}

\end{document}